\documentclass{LMCS}

\def\dOi{10(4:2)2014}
\lmcsheading%
{\dOi}
{1--80}
{}
{}
{Feb.~20, 2013}
{Oct.~30, 2014}
{}

\ACMCCS{[{\bf Theory of computation}]: Models of
  computation---Concurrency---Process calculi;Semantics and
  reasoning---Program semantics---Denotational semantics/Catgorical semantics}

\usepackage{lmcs-preamble}
\renewcommand{\with}[1]{\langle #1 \rangle}
\renewcommand{\QFI}{\maji{IQ}}
\renewcommand{\LLL}{\QFI}
\renewcommand{\SSSL}{\SSS^{\LLL}}
\renewcommand{\TTTL}{\TTT^{\LLL}}

\theoremstyle{plain}

\renewcommand{\stratglobale}{behaviour\xspace} 

\renewcommand{\stratglobales}{behaviours\xspace}

\renewcommand{\stratlocale}{strategy\xspace}  
\renewcommand{\astratlocale}{a strategy\xspace}
\renewcommand{\stratlocales}{strategies\xspace}

\renewcommand{\plays}{\E}
\renewcommand{\views}{\EVi}
\renewcommand{\Plays}{\E}

\begin{document}

\author[T.~Hirschowitz]{Tom Hirschowitz}
\address{CNRS, Universit\'e de Savoie}
\email{tom.hirschowitz@univ-savoie.fr}
\thanks{Partially
    funded by the French ANR projets blancs PiCoq ANR-10-BLAN-0305 and
    R\'ecr\'e ANR-11-BS02-0010}%

\title[Full abstraction for fair testing in CCS]{Full abstraction for fair testing in CCS \\ {\tiny (expanded version)}\rsuper*}

%% required for running head on odd and even pages, use suitable
%% abbreviations in case of long titles and many authors:

%% mandatory lists of keywords and classifications:

\titlecomment{{\lsuper*}An extended abstract of this paper has appeared in CALCO '13.} 
%%%%%%%%%%%%%%%%%%%%%%%%%%%%%%%%%%%%%%%%%%%%%%%%%%%%%%%%%%%%%%%%%%%%%%%%%%%

%% the abstract has to PRECEED the command \maketitle:
%% be sure not to issue the \maketitle command twice!

\begin{abstract}
  \noindent In previous work with Pous, we defined a semantics for CCS
  which may both be viewed as an innocent form of presheaf semantics
  and as a concurrent form of game semantics.  We define in this
  setting an analogue of fair testing equivalence, which we prove
  fully abstract w.r.t.\ standard fair testing equivalence.

  The proof relies on a new algebraic notion called \emph{playground},
  which represents the `rule of the game'. From any playground, we
  derive two languages equipped with labelled transition systems, as
  well as a strong, functional bisimulation between them.
\end{abstract}

\maketitle

\clearpage
%% start the paper here:
\tableofcontents
\clearpage

\section{Introduction}
%\subsection{Semantics for concurrency vs.\ game semantics}
\begin{wrapfigure}[6]{r}{0pt}
  \begin{minipage}[c]{0.28\linewidth}
    \vspace*{-1em}
    \begin{tabular}[t]{c|c}
      Games & Concurrency \\
      \hline
      position & \configuration \\
      player & \agent \\
      move & \action \\
      play & \trace %\\
%      view & \thread
    \end{tabular}
  \end{minipage}
\end{wrapfigure}
This paper is about \emph{game semantics} for CCS~\cite{Milner80}.
Game semantics is originally a very successful approach to
\emph{sequential} denotational
semantics~\cite{DBLP:conf/lfcs/Nickau94,DBLP:journals/iandc/HylandO00,ajm}.
Its basic idea is to interpret programs as strategies for a player in
a game, and the computational environment as an opponent. Composition
of programs is handled by letting the corresponding strategies
interact.  We mostly use game semantical terminology in this paper,
but the above dictionary may help the intuition of concurrency
theorists.

Denotational models of CCS are extremely diverse, and treat various
behavioural equivalences, as surveyed by Winskel and
Nielsen~\cite{WN}. The closest game semantical work seems to be
Laird's model~\cite{DBLP:conf/fsttcs/Laird06}, which achieves full
abstraction w.r.t.\ \emph{trace} (a.k.a.\ \emph{may testing})
equivalence for a fragment of $\pi$.  The goal of the present paper is
to design the first game semantics for a finer equivalence than trace
equivalence, in the simpler setting of CCS (we plan to address the
full $\pi$-calculus in future work).  The reason Laird is limited to
trace equivalence is that the standard notion of strategy is a set of
plays (with well-formedness conditions). Hence, e.g., the famous
coffee machines, $a.b + a.c$ and $a.(b+c)$, are identified.  Following
two recent, yet independent lines of work~\cite{RideauW,HP11}, we
generalise strategies by allowing them to accept plays \emph{in
  several ways}, thus reconciling game semantics with presheaf
models~\cite{DBLP:conf/lics/JoyalNW93}. Winskel et al.'s approach is
only starting to be applied to concrete languages, see for example the
work in progress on an affine, concurrent variant of Idealised
Algol~\cite{CCWGalop14}.  The approach
of~\cite{HP11,2011arXiv1109.4356H} (\citetalias{2011arXiv1109.4356H})
was used to give a game semantics for CCS, and define a semantic
analogue of fair testing equivalence, but no adequacy result was
proved. We here prove full abstraction of semantic fair testing
equivalence w.r.t.\ standard fair testing equivalence.  Our model is
compositional, since (1) all syntactic constructs of CCS have natural
interpretations, and (2) global dynamics may be inferred from local
dynamics, as in any game semantics (see the paragraph on innocence
below and Sections~\ref{subsubsec:strategies}
and~\ref{subsubsec:syntax}).

\subsection{Overview of the approach}
%Let us briefly summarise our approach.
\subsubsection*{Truly concurrent plays}
First of all, as in~\cite{RideauW}, our notion of play is truly
concurrent.  Indeed, it does not keep track of the order in which
(atomic) moves occur.  Instead, it only retains causal dependencies
between them (see Section~\ref{subsec:plays}).  Furthermore, our plays
form a proper category, which enables in particular a smooth treatment
of bound variables. Briefly, plays that differ only up to a
permutation of channels are isomorphic, and by construction strategies
handle them correctly.
% It does resurface 
% though, in our proof of full abstraction~\cite{fullabspi}, when we 
% construct \anlts{} for strategies. We deal with it by treating all 
% outputs as bound, and by introducing new labels to recover the 
% expected behaviour, in a way not 
% unlike~\cite{DBLP:conf/fossacs/CrafaVY12}. 

\subsubsection*{Branching behaviour}
Second, we deal with 
% A further design principle is to generalise the notion of strategy to
% account for 
branching behaviour. Standardly, and ignoring momentarily the previous
paragraph, a strategy is essentially a prefix-closed set of `accepted'
plays.  This is equivalent to functors $\op{\plays} \to 2$, where
$\plays$ is the poset of plays ordered by prefix inclusion, and $2$ is
the poset $0 \leq 1$ ($\plays$ stands for `extension').  A play
$\trasse$ is `accepted' by such a functor $F$ when $F(\trasse) = 1$,
and if $\trasse' \leq \trasse$, then functoriality imposes that
$F(\trasse) \leq F(\trasse')$, hence $F(\trasse') = 1$: this is
prefix-closedness.  In order to allow plays to be accepted in several
ways, we follow \emph{presheaf models}~\cite{DBLP:conf/lics/JoyalNW93}
and move to functors $\op{\plays} \to \set$, where $\set$ is the
category of finite ordinals and all functions between
them\footnote{The author learnt this point of view from a talk by Sam
  Staton.}. Thus, to each play $\trasse \in \plays$, a strategy
associates a \emph{set} of ways to accept it, empty if $\trasse$ is
rejected.  E.g., in the simplistic setting where $\Plays$ denotes the
poset of words over actions, ordered by prefix inclusion, the coffee
machine $a.b + a.c$ is encoded as the presheaf $S$ defined on the left
and pictured on the right:
\begin{center}
    \begin{minipage}[c]{0.26\linewidth}
      \begin{itemize}
      \item $S (\epsilon) = \ens{\star}$,
      \item $S (a) = \ens{x,x'}$,
      \item $S (ab) = \ens{y}$,
      \item $S (ac) = \ens{y'}$,
      \end{itemize}
    \end{minipage}
      \hfil
    \begin{minipage}[c]{0.44\linewidth}
      \begin{itemize}
      \item $S$ empty otherwise,
      \item $S (\epsilon \into a) = \ens{x \mapsto \star, x' \mapsto \star}$,
      \item $S (a \into ab) = \ens{y \mapsto x}$,
      \item $S (a \into ac) = \ens{y' \mapsto x'}$,
      \end{itemize}
    \end{minipage}
      \hfil
      \begin{minipage}[c]{0.28\linewidth}
        \vspace*{-.8em}
        \diag(.15,.8){%
          \& |(root)| \star \& \\
          |(x)| x \& \& |(x')| x' \\
          \\
          |(y)| y \& \& |(y')| y'. %
        }{%
          (root) edge[labelal={a}] (x) %
          edge[labelar={a}] (x') %
          (x) edge[labell={b}] (y) %
          (x') edge[labelr={c}] (y') %
        }
      \end{minipage}
\end{center}
This illustrates what is meant by `accepting a play in several ways':
the play $a$ is here accepted in two ways, $x$ and $x'$.  The other
coffee machine is of course obtained by identifying $x$ and $x'$. In
our setting, plays are considered relative to their initial position
$X$, hence strategies are presheaves $\op{\plays_X} \to \set$ on the
category of plays over $X$.

\subsubsection*{Innocence}
Finally, defining strategies as presheaves on plays is too naive,
which leads us to reincorporate the game semantical idea of
\emph{innocence}.  Example~\ref{ex:noninnocent} below exhibits such a
presheaf in which two players synchronise on a public channel $a$,
without letting others interfere. In CCS, this would amount to a
process like $\abar.P \para a.Q \para a.R$ in which, say, the first
two processes could arrange for ruling out the third.  Considering
such presheaves as valid strategies would break our main result.

In the Hyland-Ong approach, innocent strategies may be defined as
prefix-closed sets of \emph{views}, where views are special
plays representing the information that a player may `access' during
a global play. The global strategy $\exta{X}{S}$ associated to an
innocent strategy $S$ is then recovered by decreeing that
$\exta{X}{S}$ accepts all plays whose views are accepted by $S$.
This leads us to consider a subcategory $\views$ of the category
$\plays$ of plays, whose objects are called \emph{views}. We thus
have for each position $X$ two categories of strategies: the
naive one, the category $[\op{\plays_X},\set]$ of
\emph{\stratglobales} on $X$, consists of presheaves on plays; the
more relevant one, the category $[\op{(\views_X)},\set]$ of
\emph{\stratlocales} on $X$, consists of presheaves on views. 

How, then, do we recover the global \stratglobale associated to
\astratlocale, which is crucial for defining our semantic fair testing
equivalence?  The right answer is given by a standard categorical
construction called right Kan extension (see
Section~\ref{subsubsec:strategies}).  Roughly, for the \stratglobale
$B_S$ associated to a \stratlocale $S$, a way to accept some play
$\trasse \in \plays_X$ is a compatible family of ways for $S$ to
accept all views of $\trasse$.  In the boolean, setting (considering
functors $\op{\plays_X} \to 2$), this reduces to $B_S$ accepting
$\trasse$ iff all its views are accepted by $S$. Our definition thus
generalises Hyland and Ong's.

Finally, game semantical \emph{parallel composition} (different from
CCS parallel composition, though inspired from it) intuitively
lets strategies interact together.  We account for it as follows.  If
we partition the players of a play $X$ into two teams, 
we obtain  two subpositions 
\begin{wrapfigure}[4]{r}{0pt}
  \begin{minipage}[c][3em]{0.4\linewidth}
%    \vspace*{-2.5em}
  \diag(.6,.6){%
    \op{(\views_{X_1})} \&     \op{(\views_{X})} \&     \op{(\views_{X_2})} \\
    \& \set %
  }{%
    (m-1-1) edge[into] (m-1-2) %
    edge[labelbl={S_1}] (m-2-2) %
    (m-1-3) edge[linto] (m-1-2) %
    edge[labelbr={S_2}] (m-2-2) %
    (m-1-2) edge[labelon={[S_1,S_2]},dashed] (m-2-2) %
  }
\end{minipage}\end{wrapfigure}
\noindent $X_1 \into X
\otni X_2$,  each  player of $X$ belonging to $X_1$ or $X_2$ according
to its team. We have that the category $\views_X$ of views on $X$
is isomorphic to the coproduct category $\views_{X_1} +
\views_{X_2}$. The parallel composition of any two \stratlocales $S_1$
and $S_2$ on $X_1$ resp.\ $X_2$ is simply obtained by universal
property of coproduct, as above right.

\subsection{Main result: which behavioural equivalence?}
With our game in place, we easily define a translation of CCS
processes into \stratlocales. It then remains to demonstrate the
adequacy of this translation.  Our \stratlocales are actually rather
intensional, so we cannot hope for adequacy w.r.t.\ equality of
\stratlocales.  Instead, we exploit the rich structure of our model to
define both \anlts{} and an analogue of fair testing equivalence
\emph{on the semantic side}, i.e., for \stratlocales.  We then provide
two results. The most important, in the author's view, is full
abstraction w.r.t.\ standard \emph{fair testing semantics}
(Corollary~\ref{cor:final}). But the second result might be considered
more convincing by many: it establishes that our semantics is fully
abstract w.r.t.\ weak bisimilarity (Corollary~\ref{cor:wbisim}).
A reason why the latter result is here considered less important
originates in the tension between \lts{} semantics and reduction
semantics~\cite{modularLTS}. Briefly, reduction semantics is simple
and intuitive, but it operates on equivalence classes of terms (under
so-called \emph{structural} congruence). On the other hand, designing
\ltss{} is a subtle task, rewarded by easier, more structural
reasoning over reductions. We perceive \lts{} semantics as less
intrinsic than reduction semantics. E.g., for more sophisticated
calculi than CCS, several \ltss{} exist, which yield significantly
different notions of bisimilarity.

\begin{wrapfigure}[3]{r}{0pt}
  \begin{minipage}[c]{0.23\linewidth}
    \vspace*{-1.5em}
    \diag(.3,.3){%
      |(a)| \bullet \& |(b)| \bullet \& |(d)| \bullet
      \& |(f)| \bullet \\
      \& |(c)| \bullet \& |(e)| \bullet \& |(g)| \bullet %
    }{%
      (a) edge[labela={\tau}] (b) %
      edge[labelbl={\tau}] (c) %
      (b) edge[labela={\tau}] (d) %
      edge[labelbl={a}] (e) %
      (d) edge[labela={\tau}] (f) %
      edge[labelbl={b}] (g) %
    }
    % \\ 
    % \diag(.3,.3){%\ 
    %   |(a)| \bullet \& |(b)| \bullet \& |(d)| \bullet 
    %   \& |(f)| \bullet \\ 
    %   \& |(c)| \bullet \& |(e)| \bullet \& |(g)| \bullet. %\ 
    % }{%\ 
    %   (a) edge[labela={\tau}] (b) %\ 
    %   edge[labelbl={\tau}] (c) %\ 
    %   (b) edge[labela={\tau}] (d) %\ 
    %   edge[labelbl={b}] (e) %\ 
    %   (d) edge[labela={\tau}] (f) %\ 
    %   edge[labelbl={a}] (g) %\ 
    % } 
  \end{minipage}
\end{wrapfigure}
Beyond \lts{}-based
equivalences, we see essentially two options: \emph{barbed
  congruence}~\cite{DBLP:books/daglib/0004377} or some \emph{testing
  equivalence}~\cite{DBLP:journals/tcs/NicolaH84}.
Barbed congruence equates processes $P$ and $Q$, roughly, when for all
contexts $C$, $C[P]$ and $C[Q]$ are weakly bisimilar w.r.t.\
\emph{reduction} (i.e., only $\tau$-actions are allowed), and
furthermore they have the same interaction capabilities at all
stages. Barbed congruence is sometimes perceived as too discriminating
w.r.t.\ guarded choice. Consider, e.g., the CCS process $P_1$ pictured
above, and let $P_2$ be the same with $a$ and $b$ swapped.  Both
processes may disable both actions $a$ and $b$, the only difference
being that $P_1$ disables $a$ \emph{before} disabling $b$.  Barbed
congruence distinguishes $P_1$ from $P_2$ (take $C = \square \para
\abar$), which some view as a deficiency.

Another possibility would be \emph{must testing}
equivalence~\cite{DBLP:journals/tcs/NicolaH84}.  Recall that $P$
\emph{must pass} a test process $R$ iff all maximal executions of
$P \para R$ perform, at some point, a fixed `tick'
action~\cite{DBLP:journals/iandc/Gorla10}, here denoted by $\tick$.
Then, $P$ and $Q$ are must testing equivalent iff they must
pass the same tests.  Must testing equivalence is sometimes perceived
as too discriminating w.r.t.\ divergence.  E.g., consider $Q_1 =
{!\tau} \para a$ and $Q_2 = a$. Perhaps surprisingly, $Q_1$ and $Q_2$
are \emph{not} must testing equivalent. Indeed, $Q_2$ must pass the
test $\abar.\tick$, but $Q_1$ does not, due to an infinite, silent
reduction sequence.

We eventually go for fair testing equivalence, which was originally
introduced (for CCS-like calculi) to rectify both the deficiency of
barbed congruence w.r.t.\ choice and that of must testing equivalence
w.r.t.\ divergence.  The idea is that two processes are equivalent
when they \emph{should} pass the same tests. A process $P$
should pass the test $T$ iff their parallel composition $P \para T$
never loses the ability of performing the special `tick' action, after
any tick-free reduction sequence. Fair testing equivalence thus
equates $P_1$ and $P_2$ above, as well as $Q_1$ and $Q_2$.
Cacciagrano et al.~\cite{DBLP:journals/corr/abs-0904-2340} provide an
excellent survey.

\subsection{Plan and overview}\label{subsec:overview}
We now give a bit more detail on the contents.  In
Section~\ref{sec:prelim}, we introduce our notations and some
preliminaries.  Section~\ref{sec:HP} summarises from \citetalias{2011arXiv1109.4356H} the game for
CCS, the notions of strategy and behaviour, the translation
$\Translfun$ of CCS processes into strategies, and semantic fair
testing equivalence.  The rest is devoted to proving that
$\Translfun$, here decomposed as $\translfun \rond \theta$ (see
below), is such that $P \faireqs Q$ iff $\Transl{P} \faireq
\Transl{Q}$, where $\faireqs$ is standard fair testing
equivalence (Corollary~\ref{cor:final}).

\subsubsection{Playgrounds}
Our proof of this result takes a long detour to introduce a new
algebraic gadget called \emph{playground}, which we now motivate.  Our
first attempts at proving the full abstraction result were obscured by
a tight interleaving of
\begin{itemize}
\item results stating common properties of moves in the game, or of
  plays, and
\item results and constructions on strategies derived from those
  (e.g., the \lts{} for strategies).
\end{itemize}

On the other hand, the reasons why our constructions work are
intuitively simple.  Namely, innocent strategies essentially amount to
describing syntax trees by selecting their branches amongst a set of
all possible branches. This enlarges the universe of terms slightly,
but in game semantics, one studies properties of terms which also make
sense for such generalised terms.  Compositionality and the definition
of our semantic fair testing equivalence are examples where using
strategies instead of terms tends to simplify the constructions. E.g.,
associated behaviours are recovered from innocent strategies through
Kan extension, thanks to an expressive notion of morphism between
plays.  Our results essentially follow from this correspondence
between terms and strategies.

\begin{exa}
  To illustrate what we mean by generalised terms, consider standard,
  unlabelled binary trees as a stripped down example of a term
  language.  Such trees admit a description as prefix-closed sets of
  words over $\ens{0,1}$ (their sets of \emph{occurrences}). In order
  to get exactly trees, such sets should be constrained a bit. E.g.,
  the empty set of words, or the set $\ens{(),(0)}$ do not describe
  any tree.
\end{exa}

Playgrounds are a first attempt at a general framework describing
this correspondence between terms and strategies. We develop their theory in
Sections~\ref{sec:playgrounds} and~\ref{sec:strats}, whose main result
is a strong bisimulation between both presentations (i.e., terms vs.\
strategies).  This is then expoited in the next sections to derive the
main results.

The basis for playgrounds are \emph{pseudo double
  categories}~\cite{GrandisPare,GrandisPareAdjoints,LeinsterHC,GarnerPhD},
a weakening of Ehresmann's double
categories~\citep{Ehresmann:double,Ehresmann:double2}.  Playgrounds
are thus pseudo double categories with additional structure.  The
objects of a playground represent positions in the game.  There are
two kinds of morphisms: \emph{vertical} morphisms represent plays,
while \emph{horizontal} ones represent embeddings of positions. E.g.,
there are special objects representing `typical' players; and a player
of a position $X$ is a horizontal morphism $d \to X$ from such a
typical player, in a Yoneda-like way. There are then axioms to model
atomicity (plays may be decomposed into atomic moves) and locality
(plays over a large position may be restricted to any subposition;
each player only sees part of the play). There are finally a few more
technical axioms.

In Section~\ref{sec:playgrounds}, we give the definition and derive a
few basic results and constructions. In particular, we define a naive
notion of strategy, \emph{behaviours}, and a less naive notion,
\emph{strategies}.  Finally, we relate the two by exhibiting a functor
from strategies to behaviours.  In Section~\ref{sec:strats}, we prove
that strategies are in bijective correspondence with infinite terms in
a certain language. We then derive from this \anlts{} $\SSS_\D$ for
strategies.  Furthermore, we define a second language, which is closer
to usual process calculi. And indeed, instantiating this general
language to our game for CCS yields essentially CCS, the only
difference being that channel creation is treated on an equal footing
with input and output. We further equip this language of \emph{process
  terms} with \anlts{} $\TTT_\D$.  Finally, we define a translation
from process terms to strategies $\translfun \colon \TTT_\D \to
\SSS_\D$, which is proved to be a strong bisimulation
(Theorem~\ref{thm:bisim}).

At this point, it remains 
\begin{enumerate}\enlargethispage{\baselineskip}
\item to show that the pseudo double category $\Dccs$ formed by our
  game does satisfy the axioms for playgrounds, and
\item to use the strong bisimulation $\translfun$ to derive our main
  results.
\end{enumerate}

\subsubsection{Graphs with complementarity}
We start with (2), because we feel doing otherwise would disrupt the
flow of the paper. Indeed, it should not be surprising at all that
$\Dccs$ forms a playground; and furthermore the methods employed to
show this are in sharp contrast with the rest of the paper.  The plan
for (2), carried out in Section~\ref{sec:graphs}, is as follows.

First, we reduce semantic fair testing equivalence to fair testing
equivalence in the \lts{} $\SSS_{\Dccs}$, thus bridging the gap
between the game semantical world and \ltss{}. But this is not as
simple as it looks.  Indeed, Hennessy and De Nicola's original setting
for testing equivalences~\citep{DBLP:journals/tcs/NicolaH84} is not
quite expressive enough for our purposes, which leads us to define a
slightly more general one, called \emph{modular graph with
  complementarity}.  First, our setting is `typed', in the sense that
not all tests may be applied to a process $P$, only tests of a type
`compatible' with $P$. Furthermore, in modular graphs with
complementarity, fair testing equivalence relies on a notion of
\emph{complementarity} saying when two transitions may be glued
together to form a \emph{closed-world} transition.  Thus, fair testing
equivalence is `intrinsic', i.e., does not depend on any alphabet.  So
we have a mere \lts{} $\SSS_{\Dccs}$ over an \emph{ad hoc} alphabet
$\QF$ derived from $\Dccs$, and we need promote it into a modular
graph with complementarity.  This goes by refining the original
alphabet $\QF$ with `interfaces', yielding a new alphabet $\QFI$. We
then define a morphism $\chi \colon \QFI \to \QF$, and pull 
$\SSS_{\Dccs}$ back along $\chi$, thus obtaining our modular graph with
complementarity $\SSSL_{\Dccs}$ (which is thus also \anlts{} over
$\QFI$). In passing, we do the same for $\TTT_{\Dccs}$, which yields
$\TTTL_{\Dccs}$: this will be useful later.  We finally prove that
fair testing equivalence in $\SSSL_{\Dccs}$ coincides with semantic
fair testing equivalence (Lemma~\ref{cor:SSSLI}).  Similarly, we
construct a modular graph with complementarity $\ccs$ for CCS, and
show that fair testing equivalence therein coincides with standard
fair testing equivalence (Proposition~\ref{prop:fairccs}).  We are
thus reduced to proving that some composite $\ccs \xto{\theta}
\TTTL_{\Dccs} \xto{\translfun} \SSSL_{\Dccs}$ is \emph{fair}, i.e.,
preserves and reflects fair testing equivalence.

Our second step is to establish a sufficient condition for a relation
$R \colon G \modto H$ to be fair and to apply this to the graph of our
translation $\ccs \to \SSSL_{\Dccs}$.  The idea is to define what an
adequate alphabet $A$ should be in our setting, and to prove that,
essentially, if we can find an adequate alphabet $A$ for $G$ and $H$,
such that $R$ is a relation over $A$, then $R$ is fair as soon as
\begin{itemize}
\item $R$ is included in weak bisimilarity over $A$, and
\item both graphs have enough $A$-\emph{trees}, in a sense
  inspired by the notion of
  \emph{failure}~\cite{DBLP:journals/iandc/RensinkV07}.
\end{itemize}
In order to apply this, we transform $\SSSL_{\Dccs}$ and
$\TTTL_{\Dccs}$ into modular graphs with complementarity over the same
alphabet $\A$ (i.e., set of labels) as $\ccs$.  We proceed by
`relabeling' along some morphism of graphs $\LLL \xto{\xi} \A$.  We
still have our translation $\TTTL_{\Dccs} \xto{\translfun}
\SSSL_{\Dccs}$, which is a strong, functional bisimulation over $\A$.
It thus remains to check that (a) the map $\ccs \xto{\theta}
\TTTL_{\Dccs}$ is included in weak bisimilarity, and (b) both $\ccs$
and $\SSSL_{\Dccs}$ have enough $\A$-trees.  Roughly, $G$ has enough
$A$-trees when, for any $t$ in a certain class of tree-like \ltss{}
over $A$ called $A$-trees, there exists $x_t \in G$ weakly bisimilar
to $t$.  For (b), all three \ltss{} under consideration clearly have
enough $\A$-trees. For (a), our proof is brute force.

\subsubsection{CCS as a playground}
We finally deal in Section~\ref{sec:ccs} with the last missing bit of
our proof: we show that $\Dccs$ forms a playground. This rests upon
the following two main ingredients.

First, we design a correctess criterion for plays, in a sense close to
correctness criteria in linear logic.  Namely, plays from some
position $X$ to position $Y$ are represented as particular cospans $Y
\xto{s} U \xot{t} X$ in some category.  Specifically, they are
obtained by closing a given set of cospans named \emph{moves} under
identities and composition.  We design a combinatorial criterion
for deciding when an arbitrary cospan is indeed a play.

The second main ingredient is a construction of the \emph{restriction}
of a play $U$ from some position $X$ to a subposition $X' \into X$.
Briefly, this means computing the part of $U$ which is relevant to
players in $X'$. This construction is almost easy: most of $U$ may be
`projected' back onto the initial position $X$, and then a mere
pullback
\begin{center}
  \Diag{%
    \pbk{m-2-1}{m-1-1}{m-1-2} %
  }{%
    \restr{U}{X'} \& U \\
    X' \& X %
  }{%
    (m-1-1) edge[into,labelu={}] (m-1-2) %
    edge[labell={}] (m-2-1) %
    (m-2-1) edge[into,labeld={}] (m-2-2) %
    (m-1-2) edge[labelr={}] (m-2-2) %
  }
\end{center}
of sets gives the needed restriction. The glitch is that in general
some parts of $U$ may not canonically be projected back onto $X$. The
principle for this projection is as simple as: project, e.g., input
moves to the inputting player. The problem arises for
synchronisations. Projecting them to the channel over which the
synchronisation occurs does not yield the desired result, and
similarly projecting to either of the involved players fails.  Our
solution is to ignore synchronisations at first, and later reintroduce
them automatically using a technique from algebraic
topology: factorisation systems~\cite{Joyal:ncatlab:facto}.

With both of these ingredients in place, the proof is relatively
straightforward.

Section~\ref{sec:conc} concludes and provides some perspectives for
future work.

\subsection{Related work}
Our bisimulation result relating terms to strategies for any
playground draws inspiration from \emph{Kleene
  coalgebra}~\cite{DBLP:conf/fossacs/BonsangueRS09,DBLP:conf/concur/BonchiBRS09}. There,
the main idea is that both the syntax and the semantics of various
kinds of automata should be \emph{derived} from more basic data
describing, roughly, the `rule of the game'. Formally, starting from a
well-behaved (\emph{polynomial}) endofunctor on sets, one constructs
both (1) an equational theory and (2) a sound and complete coalgebraic
semantics. This framework has been applied in standard automata
theory, as well as in quantitative settings.  Nevertheless, its
applicability to programming language theory is yet to be
established. E.g., the derived languages do not feature parallel
composition.  Our playgrounds may be seen as a first attempt to convey
such ideas to the area of programming language theory.  Technically,
our framework is rather different though, in that we replace the
equational theory by a transition system, and the coalgebraic
semantics by a game semantics.  To summarise, our approach is close in
spirit to Kleene coalgebra, albeit without quantitative
aspects. Conversely, Kleene coalgebra resembles our approach without
innocence.

Building upon previous
work~\cite{DBLP:conf/lics/AbramskyM99,Mellies04,DBLP:conf/concur/MelliesM07}
on \emph{asynchronous} games, a series of papers by Winskel and
collaborators (see, e.g., \citet{RideauW,DBLP:conf/fossacs/Winskel13})
attempt to define a notion of concurrent strategy encompassing both
innocent game semantics and presheaf models. Ongoing work evoked
above~\cite{CCWGalop14} shows that the model does contain innocent
game semantics, but presheaf models are yet to be investigated.
(Their notion of innocence, borrowed from Faggian and
Piccolo~\cite{DBLP:conf/tlca/FaggianP09}, is not intended to be
related to that of Hyland and Ong.)  In their framework, a game is an
event structure, whose events are thought of as moves, equipped with a
notion of polarity.  In one of the most recent papers in the
series~\cite{DBLP:conf/fossacs/Winskel13}, Winskel establishes a
strong relationship between his concurrent strategies and presheaves.
For a given event structure with polarity $A$, he considers the
so-called \emph{Scott order} on the set $\CCC(A)$ configurations of
$A$.  For two configurations $c$ and $d$, we have $c \sqsubseteq_A d$
iff $d$ may be obtained from $c$ by removing some negative moves and
then adding some positive ones, in a valid way.  Strategies are then
shown to coincide with presheaves on $(\CCC(A), \sqsubseteq_A)$.  This
is close in spirit to our use of presheaves, but let us mention a few
differences. First, our games do not directly deal with
polarity. Furthermore, in our setting, for any morphism $p \to q$ of
plays, $q$ is intuitively bigger than $p$ in some way, unlike what
happens with the Scott ordering. Finally, an important point in our
use of (pre)sheaves is that, unlike configuration posets, our plays
form proper \emph{categories}, i.e., homsets may contain more than one
element (intuitively, the same view may have several occurrences in a
given play). Thus, potential links between both approaches remain to
be further investigated.

To conclude this paragraph, let us mention a few, more remotely related
lines of work.  Melli\`es~\cite{DBLP:conf/lics/Mellies12}, although in
a deterministic and linear setting, incorporates some `concurrency'
into plays by presenting them as string diagrams. Our notion of
innocent strategy shares with Harmer et
al.'s~\cite{DBLP:conf/lics/HarmerHM07} presentation of innocence based
on a distributive law the goal of better understanding the original
notion of innocence. Finally, others have studied game semantics in
non-deterministic~\cite{DBLP:conf/lics/HarmerM99} or
concurrent~\cite{DBLP:conf/fossacs/GhicaM04,DBLP:conf/fsttcs/Laird06}
settings, using coarser, trace-based behavioural equivalences.

% \subsection*{Plan} 
% In Section~\ref{sec:prelim}, we introduce our notations and some 
% preliminaries.  Section~\ref{sec:HP} summarises from HP the game, the 
% notions of strategy and behaviour, and semantic fair testing 
% equivalence. Sections~\ref{sec:playgrounds} introduces playgrounds, up 
% to the functor from strategies to behaviours. Section~\ref{sec:strats} 
% then constructs both \ltss{} $\SSS_\D$ and $\TTT_\D$, and the strong, 
% functional bisimulation between them. Section~\ref{sec:graphs}  
% introduces graphs with complementarity and uses them to establish our 
% main result. Section~\ref{sec:ccs} wraps up by showing that 
% $\Dccs$ does form a playground. Section~\ref{sec:conc} concludes and 
% provides some perspectives for future work. 

\section{Prerequisites and preliminaries}\label{sec:prelim}
In this section, we recall some needed material and introduce our
notations.  We attempt to provide intuitive, yet concise explanations,
but these may not suffice to get the non-specialist reader up to
speed, so we also provide references when possible.

For the reader's convenience, we finally provide in
Figure~\ref{fig:cheat} (end of paper) a summary of notations, beyond
those introduced here.

\subsection{Sets, categories, presheaves}\label{subsec:prelim:cats}
We make intensive use of category theory, of which we assume prior
knowledge of categories, functors, natural transformations, limits and
colimits, adjoint functors, presheaves, bicategories, Kan extensions,
and pseudo double categories. All of this except pseudo double
categories is entirely covered in Mac Lane's standard
textbook~\cite{MacLane:cwm} and the beginning of Mac Lane and
Moerdijk~\cite{MM}. For a more leisurely introduction, one may consult
Lawvere and Schanuel~\cite{DBLP:books/daglib/0095291}, or
Leinster~\cite{LeinsterCats}.  The needed material on Kan extensions
roughly amounts to their expression as ends, which is recalled when
used (Section~\ref{subsubsec:strategies}).  The last bit, namely the
notion of \emph{pseudo double category} is briefly recalled below,
after fixing some notation.  Finally, there are very local uses of
locally presentable categories~\cite{Adamek} in the present section,
and of adhesive category theory~\cite{DBLP:conf/fossacs/LackS04} in
the proof of Lemma~\ref{lem:decompleft}.

Throughout the paper, any finite ordinal $n$ is seen as $\ens{1,
  \ldots, n}$ (rather than $\ens{0, \ldots, n-1}$).
In any category, for any object $C$ and set $X$, let $X \cdot C$
denote the $\card{X}$-fold coproduct of $C$ with itself, i.e., $C +
\cdots + C$, $\card{X}$ times.

$\Set$ is the category of sets; $\set$ is a skeleton of the category
of finite sets, e.g., the category of finite ordinals and arbitrary
maps between them; $\ford$ is the category of finite ordinals and
monotone maps between them.  For any category $\C$, $\Psh{\C} =
[\op\C,\Set]$ denotes the category of presheaves on $\C$, while
$\FPsh{\C} = [\op\C,\set]$ and $\OPsh{\C} = [\op\C,\ford]$
respectively denote the categories of presheaves of finite sets and of
finite ordinals.  One should distinguish, e.g., `presheaf of finite
sets' $\op\C \to \set$ from `finite presheaf of sets' $F \colon \op\C
\to \Set$. The category $\Pshf{\C}$ of \emph{finite} presheaves is the
full subcategory of $\Psh{\C}$ spanning presheaves $F$ which are
finitely presentable~\cite{Adamek}. In presheaf categories, finitely
presentable objects are the same as finite colimits of
representables. In the only case we will use ($\C$ below), because
representables have finite categories of elements, the latter in turn
coincide with presheaves $F$ such that the disjoint union $\sum_{c \in
  \ob (\C)} F (c)$ is finite.  For all presheaves $F$ of any such
kind, $x \in F(d)$, and $f \colon c \to d$, let $x \cdot f$ denote
$F(f)(x)$.
\begin{rem}
  This conflicts with the notation $X \cdot C$ above, but context
  should disambiguate, as in $X \cdot C$ a set $X$ acts on an object
  $C$, whereas in $x \cdot f$, a morphism $f$ acts on an object $x$.
\end{rem}
We denote the Yoneda embedding by
$\yoneda \colon \C \to \Psh{\C}$, and often abbreviate $\yoneda(c)$ to
just $c$.

For any functor $F \colon \C \to \D$ and object $D \in \D$, let 
$F_D$ denote the comma category on the left below, and 
$F(D)$ denote the pullback category on the right:
  \begin{equation}
  \Diag{%
    \twocellright{E}{EX}{un} %
  }{%
    |(EX)| F_D \& |(un)| 1 \\
    |(E)| \C \& |(Dh)| \D %
  }{%
    (EX) edge (un) edge (E) %
    (E) edge[labeld={F}] (Dh) %
    (un) edge[labelr={\name{D}}] (Dh) %
  }
\hspace*{.25\textwidth}
      \Diag{%
        \pbk{E}{EX}{un} %
      }{%
        |(EX)| F(D) \& |(un)| 1 \\
        |(E)| \C \& |(Dh)| \D. %
      }{%
        (EX) edge (un) edge (E) %
        (E) edge[labeld={F}] (Dh) %
        (un) edge[labelr={\name{D}}] (Dh) %
      }\label{eq:pbkcat}
    \end{equation}
When $F$ is clear from context, we simply write $\C_D$, resp.\
$\C(D)$. Also, as usual, when $F$ is the identity,
we use the standard slice notation $\D / D$.

Finally, we briefly recall pseudo double categories. They are a
weakening of Ehresmann's double
categories~\citep{Ehresmann:double,Ehresmann:double2}, notably studied
by \citet{GrandisPare,GrandisPareAdjoints}, \citet{LeinsterHC}, and
\citet{GarnerPhD}.  The weakening lies in the fact that one dimension
is strict and the other weak (i.e., bicategory-like).  We need to
consider proper pseudo double categories, notably we use cospans in
examples, but we often handle pseudoness a bit sloppily.  Indeed, the
proofs of Section~\ref{sec:playgrounds} quickly become unreadable when
accounting for pseudoness.

A pseudo double category $\D$ consists of a set $\ob (\D)$ of
\emph{objects}, shared by a `horizontal' category $\Dh$ and a
`vertical' bicategory $\Dv$. Following Paré~\cite{PareYoneda}, $\Dh$,
being a mere category, has standard notation (normal arrows, $\rond$
for composition, $\id$ for identities), while the bicategory $\Dv$
earns fancier notation ($\proto$ arrows, $\vrond$ for composition,
$\idv$ for identities). $\D$ is furthermore equipped with a set of
\emph{double cells} $\alpha$, which have vertical, resp.\ horizontal,
domain and codomain, denoted by $\domv (\alpha)$, $\codv (\alpha)$,
$\domh (\alpha)$, and $\codh (\alpha)$. 

\begin{wrapfigure}[6]{r}{0pt}
  \begin{minipage}[t][3em]{0.28\linewidth}
    \vspace*{-1.6em}%\hspace*{-1em}
  \Diag(.6,.6){%
    % \twocellbr{m-2-1}{m-1-1}{m-1-2}{\alpha} %\ 
    % \twocellbr{m-2-2}{m-1-2}{m-1-3}{\alpha'} %\ 
    % \twocellbr{m-3-1}{m-2-1}{m-2-2}{\beta} %\ 
    % \twocellbr{m-3-2}{m-2-2}{m-2-3}{\beta'} %\ 
  }{%
    X \& X' \& X'' \\
    Y \& Y' \& Y'' \\
    Z \& Z' \& Z''%
  }{%
    (m-1-1) edge[labelu={h}] (m-1-2) %
    edge[pro,labell={u},twoleft={ur}{}] (m-2-1) %
    (m-2-1) edge[labelu={h'}] (m-2-2) %
    (m-1-2) edge[pro,labell={u'},twoleft={u'r}{},tworight={u'l}{}] (m-2-2) %
    (m-1-2) edge[labelu={k}] (m-1-3) %
    (m-2-2) edge[labelu={k'}] (m-2-3) %
    (m-1-3) edge[pro,labelr={u''},tworight={u''l}{}] (m-2-3) %
    (m-2-1) 
    edge[pro,labell={v},twoleft={vr}{}] (m-3-1) %
    (m-3-1) edge[labela={h''}] (m-3-2) %
    (m-2-2) edge[pro,labell={v'},twoleft={v'r}{},tworight={v'l}{}] (m-3-2) %
    (m-3-2) edge[labela={k''}] (m-3-3) %
    (m-2-3) edge[pro,labelr={v''},tworight={v''l}{}] (m-3-3) %
% cells
    (ur) edge[cell={.2},labela={\alpha}] (u'l) %
    (u'r) edge[cell={.2},labela={\alpha'}] (u''l) %
    (vr) edge[cell={.2},labela={\beta}] (v'l) %
    (v'r) edge[cell={.2},labela={\beta'}] (v''l) %
  }
  \end{minipage}
\end{wrapfigure}
We picture this as, e.g., $\alpha$ on the right, where $u =
\domh (\alpha)$, $u' = \codh (\alpha)$, $h = \domv (\alpha)$, and $h'
= \codv (\alpha)$. Finally, there are operations for composing double
cells: \emph{horizontal} composition $\rond$ composes them along a
common vertical morphism, \emph{vertical} composition $\vrond$
composes along horizontal morphisms. Both vertical compositions (of
morphisms and of double cells) may be associative only up to coherent
isomorphism. The full axiomatisation is given by
Garner~\cite{GarnerPhD}, and we here only mention the
\emph{interchange law}, which says that the two ways of parsing the
above diagram coincide: $(\beta' \rond \beta) \vrond (\alpha' \rond
\alpha) = (\beta' \vrond \alpha') \rond (\beta \vrond \alpha)$.

For any (pseudo) double category $\D$, we denote by $\DH$ the category
with vertical morphisms as objects and double cells as morphisms, and
by $\DV$ the bicategory with horizontal morphisms as objects and
double cells as morphisms.  Domain and codomain maps arrange into
functors $\dom_v,\cod_v \colon \DH \to \D_h$ and $\dom_h,\cod_h \colon
\DV \to \D_v$. We will refer to $\domv$ and $\codv$ simply as $\dom$
and $\cod$, reserving subscripts for $\domh$ and $\codh$.

We introduce a bit more notation.
\begin{defi}
  A double cell 
%  \doublecell{Y}{X}{Y'}{X'}{h}{u}{u'}{k}{\alpha}
  is \emph{special} when its vertical domain and codomain are (horizontal) identities.
% $Y = Y'$, $X = X'$, $h = \id_Y$, and $k = \id_X$.
\end{defi}
For any object $X \in \ob (\D)$, $\DH (X)$ denotes the category with
\begin{itemize}
\item objects all vertical morphisms to $X$, and
\item morphisms $u \to v$ all double cells 
  \doublecellpro{Y}{Y'}{X}{X}{h}{u}{v}{k}{\alpha} with $\codv (\alpha) = k = \id_X$.
\end{itemize}
This complies with noting $\C(D)$ for the pullback category~\eqref{eq:pbkcat}, taking 
$\cod_v$ for $F$ and $X$ for $D$.

\subsection{Transition systems}\label{subsec:prelim:lts}
Beyond category theory, this paper also makes heavy use of the theory
of \ltss{} and associated techniques, especially bisimulation and
other behavioural equivalences.  The notion of \lts{} that we'll use
here is a little more general than usual. Indeed, usually, the
transitions of \anlts{} are labelled with letters in a given set
called the \emph{alphabet}, or the set of \emph{actions}. Here, we
consider the case where the vertices of \anlts{} may be typed, and
actions may change the type.  Extending the usual theory to this
setting is straightforward, so we only provide a brief overview. For
more on the usual theory, modern references are~\citet{Sangio}
and~\citet{SangioRutten}. Our setting is essentially a baby version of
Fiore's~\cite{DBLP:conf/ifipTCS/Fiore00} (see the references therein
for precursors).

Let $\Gph$ be the category of reflexive graphs, which has as objects
diagrams $s,t \colon E \rightrightarrows V$ in $\Set$, equipped with a
further arrow $e \colon V \to E$ such that $s \rond e = t \rond e =
\id_V$. We will as usual denote $e(v)$ by $\id_v$. Morphisms are those
morphisms between underlying graphs which preserve identity arrows.
$\Gph$ is thus the category of presheaves over the category \diaginline{\star
  \& {[1]}}{(m-1-2) edge[labelo={e}] (m-1-1) %
  (m-1-1) edge[bend left,labela={s}] (m-1-2) %
  edge[bend right,labelb={t}] (m-1-2) %
} with $e \rond s = \id_\star$ and $e \rond t = \id_\star$.
\begin{defi}
  For any $A \in \Gph$, let the \emph{category of \ltss{} over $A$} be
  just the slice category $\Gph / A$.
\end{defi}

\subsubsection{Basic notation}\label{subsubsec:notation:lts} $A$ is called the \emph{alphabet}, which goes
slightly beyond the usual notion of an alphabet. The latter would here
come in the form of the graph with one vertex, an identity edge, plus
an edge for each letter.  By convention, and mainly to ease graphical
intuitions in Sections~\ref{sec:playgrounds} and~\ref{sec:strats}, for
any \lts{} $p \colon G \to A$, we understand an edge $e \colon x' \to
x$ in $G$ as a transition from $x$ to $x'$. Of course, to recover a
more standard notation, one may replace all graphs with their
opposites.  When $e$ does not matter, but $p(e)$ does, we denote such
a transition by $x \mathrel{{}_{A}\!\xot{p(e)}} x'$, omitting the
subscript $A$ when clear from context.

For any reflexive graph $A$, we denote by $A^\star$ the graph with the
same vertices and arbitrary paths as edges. $A^\star$ is reflexive,
with identity edges given by empty paths. Similarly, $f^\star \colon
A^\star \to B^\star$ is the morphism induced by $f \colon A \to B$.
This defines a functor $\Gph \to \Cat$, which is \emph{not} left
adjoint to the forgetful functor $U \colon \Cat \to \Gph$.  There is a
left adjoint, though, which we denote by $\freecatfun$. It is given by
a quotient of $A^\star$, essentially equating $(\id)$ and $()$, i.e.,
the singleton, identity path and the empty one.
\begin{defi}
  Let $\freecat{A}$ denote the graph with the same vertices as $A$,
  whose edges $x \to x'$ are paths $x \tostar x'$ in $A$, considered
  equivalent modulo removal of identity edges.
\end{defi}
Any path $\rho$ has a normal form, obtained by removing all identity
edges and denoted by $\idfree{\rho}$.  We will deem such normal forms
\emph{identity-free}.  We denote by $x \mathrel{{}_{A}\!\xOt{a}} x'$
any path $\rho \colon x' \to^\star x$ in $G$, such that
$\idfree{p^\star(\rho)} = \idfree{(a)}$.  Concretely, if $a$ is an
identity, then $p^\star(\rho)$ only consists of identity edges;
otherwise, $p^\star(\rho)$ consists of $a$, possibly surrounded by
identity edges. In the former case, we further abbreviate the notation
to $x \xOt{} x'$ (observe that $\rho$ may well be empty).  Similarly,
for any path $r$ in $A^\star$, $x \xxOt{A}{r} x'$ denotes any path
$\rho \colon x' \to^\star x$ in $G$ such that $\idfree{p^\star (\rho)}
= \idfree{r}$.

\subsubsection{Bisimulation and change of base} 
In this section, we revisit the usual notion of (strong and weak) bisimulation in our graph-based setting,
and provide a few stability results under base change and cobase change.
Let us start with strong bisimulations.

\begin{defi} For any $G, G' \in \Gph$, a morphism $f \colon G \to G'$
  is a \emph{graph fibration} iff for all $x \in G$, $y \in G'$, and
  $e' \in G' (y,f (x))$, there exist $x' \in G$ and $e \in G (x',x)$
  such that $f (e) = e'$.
\end{defi}

Consider morphisms $p \colon G \to A$ and $p' \colon G' \to A$.  
A \emph{relation over $A$} is a subgraph of the pullback
\begin{center}
  \Diag{%
    \pbk{m-2-1}{m-1-1}{m-1-2} %
    }{%
    G \times_A G' \& G' \\
    G \& A. %
  }{%
(m-1-1) edge[labelu={}] (m-1-2) %
edge[labell={}] (m-2-1) %
(m-2-1) edge[labeld={p}] (m-2-2) %
(m-1-2) edge[labelr={p'}] (m-2-2) %
  }
\end{center}
In particular, if two edges $(e,e')$ are related by some $R \subseteq
G \times_A G'$, then so are their sources, resp.\ targets.
We denote such relations by $R \colon G \modto G'$.

We will most often deal with \emph{full} relations, i.e., such that
$R(e,e')$ iff both sources and targets are related. Of course, such
relations need only to be defined on vertices.
\begin{defi}
  A \emph{simulation} $G \modto G'$ is a relation $R$ over $A$ such that for all $e
  \in G(x',x)$, if $R(x,y)$ then there exist $y'$ and $e' \in
  G'(y',y)$ such that $R(e,e')$.  A \emph{bisimulation} is a
  simulation whose converse also is a simulation.
\end{defi}
When $R$ is full, $R$ is a simulation iff for all $e \in G(x',x)$, if
$R(x,y)$ then there exists $y'$ and $e' \in G'(y',y)$ such that
$R(x',y')$ and $e$ and $e'$ are mapped to the same edge in $A$.

\begin{prop}
  $R$ is a simulation iff its first projection $R \into {G \times_A
    G'} \to G$ is a  graph fibration.  Accordingly, $R$ is a
  bisimulation iff both projections are  graph fibrations.
\end{prop}
\begin{proof}
  Straightforward.
\end{proof}\enlargethispage{3\baselineskip}
\begin{rem}\label{rem:joyal}
  The characterisation of simulations in terms of  graph fibrations may
  be attributed to Joyal et al.~\cite{DBLP:conf/lics/JoyalNW93}, who first observed
   that a morphism $f \colon G \to G'$ in $\Gph/A$ is a functional bisimulation iff
   for any commuting square as the exterior of
   \begin{center}
    \diag{%
      \yoneda(\star) \& G \\
      \yoneda[1] \& G', %
    }{%
      (m-1-1) edge[labelu={}] (m-1-2) %
      edge[labell={\yoneda(t)}] (m-2-1) %
      (m-2-1) edge[labeld={}] (m-2-2) %
      (m-1-2) edge[labelr={f}] (m-2-2) %
      (m-2-1) edge[dashed] (m-1-2) %
    }
  \end{center}
  there exists a dashed arrow making both triangles commute.  Here,
  $\yoneda(t) \colon \yoneda(\star) \to \yoneda[1]$ maps the reflexive
  graph with a single vertex (and its identity edge) to the one with
  two vertices and just one non-identity edge $e$ between them, by
  picking out the target of $e$.
  This precisely says that $f$ is a  graph fibration.

  A peculiar aspect of this characterisation is that it may seem
  independent from $A$. Actually, $R$ is a relation over $G \times_A
  G'$, and $f$ is a morphism over $A$.
\end{rem}

As usual, fixing $G$ and $G'$ over $A$, we have:
\begin{prop}
  Bisimulations are closed under union, and the union of all
  bisimulations, called \emph{bisimilarity}, is again a bisimulation,
  the maximum one. 
\end{prop}
Considering \emph{endo}relations $G \modto G$, we talk about
bisimilarity \emph{in} $G$.
\begin{notation}
  Bisimilarity in $G$ over $A$ is denoted by $\bisim_A$. It may, upon
  a slight abuse of notation, be understood as an equivalence relation
  over all vertices of any two graphs over $A$. Namely, if $G$ and
  $G'$ are graphs over $A$, we may write $x \bisim_A y$ when $x \in G$
  and $y \in G'$ to mean bisimilarity in $G + G'$.
\end{notation}

Before treating weak bisimulations, we consider a first stability result,
which is all we need about strong bisimulations.

Any morphism $f \colon A \to B$ induces by pullback a change-of-base
functor $\cob{f} \colon \Gph / B \to \Gph / A$, which has a left
adjoint $\cocob{f}$ given by composition with $f$.

\begin{prop}\label{prop:change of base}
  For any morphism of graphs $f \colon A \to B$, both functors $\cob{f}
  \colon \Gph / B \to \Gph / A$ and $\cocob{f} \colon \Gph / A \to \Gph /
  B$, i.e., pullback along and post-composition with $f$,
  preserve functional bisimulations.
\end{prop}

\begin{proof}
  The case of $\cocob{f}$ is actually trivial.
  For $\cob{f}$, we use Remark~\ref{rem:joyal}. By the pullback lemma,
  the square on the right below
  % \begin{center} 
  %   \Diag{%\ 
  %     \pbk{m-2-1}{m-1-1}{m-1-2} %\ 
  %   }{%\ 
  %     \cob{f} (G) \& G \\ 
  %     \cob{f} (G') \& G', %\ 
  %   }{%\ 
  %     (m-1-1) edge[labelu={}] (m-1-2) %\ 
  %     edge[labell={}] (m-2-1) %\ 
  %     (m-2-1) edge[labeld={}] (m-2-2) %\ 
  %     (m-1-2) edge[labelr={}] (m-2-2) %\ 
  %   } 
  % \end{center} 
  is a pullback. We check that $\cob{f} (G) \to \cob{f} (G')$ is again a
  bisimulation. Indeed, consider any square as
  on the left below:
  % \begin{center} 
  %   \diag{%\ 
  %     \ens{b} \& \cob{f}(G) \\ 
  %     e \& \cob{f}(G'). %\ 
  %   }{%\ 
  %     (m-1-1) edge[labelu={}] (m-1-2) %\ 
  %     edge[labell={}] (m-2-1) %\ 
  %     (m-2-1) edge[labeld={}] (m-2-2) %\ 
  %     (m-1-2) edge[labelr={}] (m-2-2) %\ 
  %   } 
  % \end{center} 
  % Pasting this with the above pullback square, we obtain the solid 
  % part of the following diagram of graphs over $B$: 
  \begin{center}
    \Diag{%
      \pbk[10]{m-2-2}{m-1-2}{m-1-3} %
      \path[->,draw] (m-2-1) edge[fore,dashed,bend right=3] (m-1-3) %
      ; %
    }{%
      \yoneda (\star) \& \cob{f}(G) \& G \\
      \yoneda[1] \& \cob{f}(G') \& G'. %
    }{%
      (m-1-1) edge[labelu={}] (m-1-2) %
      edge[labell={\yoneda(t)}] (m-2-1) %
      (m-2-1) edge[labeld={}] (m-2-2) %
      (m-1-2) edge[labelr={}] (m-2-2) %
      (m-2-1) edge[dotted] (m-1-2) %
      (m-1-2) edge (m-1-3) %
      (m-2-2) edge (m-2-3) %
      (m-1-3) edge (m-2-3) %
    }
  \end{center}
  Because $G \to G'$ is a bisimulation, we obtain the dashed arrow
  making both triangles commute. But then by universal property of
  pullback, we obtain the dotted arrow, making the corresponding
  bottom triangle commute.  Finally, the top triangle commutes upon
  postcomposition with $\cob{f} (G) \to G$, and after composition with
  $\cob{f} (G) \to \cob{f} (G')$, hence commutes by uniqueness in the
  universal property of pullback.
\end{proof}
\begin{rem}
  This is an instance of the fact that \emph{right maps} are stable
  under pullback in any weak factorisation
  system~\cite{Joyal:ncatlab:facto}, here with the factorisation
  system cofibrantly generated by the sole map $\yoneda (t)$.
\end{rem}

Let us now treat weak bisimulations. We start with the functional case.
\begin{defi}
  A morphism $f \colon G \to G'$ in $\Gph/A$ is a \emph{functional,
    weak bisimulation} iff $\freecat{f} \colon \freecat{G} \to
  \freecat{G'}$ is a graph fibration.
\end{defi}
\begin{prop}
  This equivalent to the fact that, for any edge $e \colon y' \to
  f(x)$ in $G'$, there exists $x'$ in $G$ and a path $r \colon x'
  \tostar x$ such that $\idfree{f^\star(r)} = \idfree{(e)}$.
\end{prop}
\begin{proof}
  If $e$ is an identity, then taking the empty path for $r$ will do,
  so the condition really says something about non-identity edges $e$.
\end{proof}
\begin{rem}
  Remark~\ref{rem:joyal} adapts to weak, functional bisimulations,
  using $\freecat{f}$ instead of $f$.
\end{rem}

Let us now handle the relational case.  In the strong case, a relation
between graphs $G$ and $G'$ over $A$ was defined to be a subobject of
the pullback $G \times_A G'$, and simulation properties were related
to the projections being graph fibrations. In order to follow this
pattern here, we need to consider $\freecat{A}$ instead of
$A$. However, in general, $\freecat{G} \times_{\freecat{A}}
\freecat{G'}$ differs from $\freecat{G \times_A G'}$. We consider the
former:
\begin{defi}
  A \emph{weak simulation} $G \modto G'$ is a relation $R \subseteq
  \freecat{G} \times_{\freecat{A}} \freecat{G'}$ whose first projection $R
  \into \freecat{G} \times_{\freecat{A}} \freecat{G'} \to \freecat{G}$ is a
  graph fibration.
  
  $R$ is a \emph{weak bisimulation} iff both projections are graph
  fibrations.
\end{defi}
Explicitly, consider $p \colon G \to A$ and $p' \colon G' \to A$, and
$R$ as above a weak simulation. For any edge $r \colon x \ot x'$ in
$\freecat{G}$, i.e., identity-free path $r \colon x \xotstar{} x'$,
and $y \in G'$ such that $R(x,y)$, there should be an identity-free
path $r' \colon y \xotstar{} y'$ in $G'$ such that $(r,r') \in R$.  If
$R$ is full, this is equivalent to the existence, for each edge $e
\colon x \ot x'$ in $G$ and $y \in G'$ such that $R(x,y)$, of an
identity-free path $r' \colon y \xotstar{} y'$ such that $R(x',y')$
and $\idfree{(p(e))} = \idfree{(p')^\star (r')}$. We will only consider
full relations in this paper, hence only the last characterisation
will matter to us.

As in the strong case, we have for any fixed $G$ and $G'$ over $A$:
\begin{prop}
  Weak bisimulations are closed under union, and the union of all weak
  bisimulations, called \emph{weak bisimilarity}, is again a weak
  bisimulation, the maximum one.
\end{prop}
\begin{notation}
  Weak bisimilarity over $A$ is denoted by $\wbisim_A$. As for strong
  bisimilarity, we will abuse notation and consider $\wbisim_A$ as a
  relation between the vertices of any two graphs over $A$.
\end{notation}

\subsection{CCS}
The main subject of this paper is CCS~\cite{Milner89}, and fair
testing equivalence over it.  We work with a standard version, except
in two respects. First, we work with infinite terms, which spares us
the need for replication, recursion, or other possible mechanisms for
describing infinite processes in a finite way.  Second, we work with a
de Bruijn-like presentation: terms carry their (finite) sets of known
channels, in the form of a finite number. I.e., the number $n$
indicates that the considered process knows channels $1, \ldots, n$
(which complies with our notation for finite ordinals, introduced in
Section~\ref{subsec:prelim:cats}).

\begin{rem}
  While the de Bruijn-like presentation clearly is a matter of
  convenience, working with infinite terms does have an impact on our
  results.  Restricting ourselves to recursive processes (e.g., by
  introducing some recursion construct), we would still have that
  $\Transl{P} \faireq \Transl{Q}$ implies $P \faireqs Q$.  The
  converse is less obvious and may be stated in very simple terms:
  suppose you have two recursive CCS processes $P$ and $Q$ and a test
  process $T$, possibly non-recursive, distinguishing $P$ from $Q$; is
  there any recursive $T'$ also distinguishing $P$ from $Q$?  We leave
  this question open.
\end{rem}

Our (infinite) CCS terms are coinductively generated by the typed grammar
\begin{mathpar}
  \inferrule{\Gam \vdash P \\ \Gam \vdash Q}{\Gam \vdash P|Q} \and %
  \inferrule{\Gam,a \vdash P}{\Gam \vdash \nu a. P} \and % 
  \inferrule{\ldots \\ \Gam \vdash P_i \\ \ldots}{\Gam \vdash \sum_{i \in n} \alpha_i.P_i}~(n \in \Nat) \,.%
\end{mathpar}
Here, as announced, $\Gam$ ranges over $\Nat$, i.e., the free names of
a process always are $1 \ldots n$ for some $n$. Accordingly, $\Gam,a$
denotes just $n+1$ (and then $a = n+1$).  Furthermore, $\alpha_i$ is
either $a$, $\abar$, or $\tick$ (for $a \in \Gam$).  The latter is a
`tick' move used in the definition of fair testing equivalence.

\begin{defi}\label{def:A}
  Let $\A$ be the reflexive graph with vertices given by finite
  ordinals, edges $\Gam \to \Gam'$ given by $\emptyset$ if $\Gam \neq
  \Gam'$, and by $\Gam + \Gam + \ens{\id,\tick}$ otherwise,  $\id
  \colon \Gam \to \Gam$ being the identity edge on $\Gam$.
  Elements of the first summand are denoted by $a
  \in \Gam$, while elements of the second summand are denoted by
  $\abar$.
\end{defi}

\begin{figure}[t]
  \begin{mathpar}
    \inferrule{ }{(\Gam \vdash P) \xot{\id} (\Gam \vdash P)} %
    \and %
    \inferrule{ }{(\Gam \vdash \sum_{i \in n} \alpha_i.P_i) \xot{\alpha_i} %
      (\Gam \vdash P_i)} %
    \\ %
    \inferrule{(\Gam \vdash P_1) \xot{\alpha} (\Gam \vdash P'_1)}{(\Gam \vdash P_1 \para P_2) %
      \xot{\alpha} (\Gam \vdash P'_1 \para P_2)} %
    \and %
    \inferrule{(\Gam \vdash P_2) \xot{\alpha} (\Gam \vdash P'_2)}{(\Gam \vdash P_1 \para P_2) %
      \xot{\alpha} (\Gam \vdash P_1 \para P'_2)} %
    \and %
    \inferrule{(\Gam,a \vdash P) \xot{\alpha} (\Gam, a \vdash P')}{ %
      (\Gam \vdash \nu a.P) \xot{\alpha} (\Gam \vdash \nu a.P')}~{(\alpha \notin \ens{a, \abar})} %
    \and %
    \inferrule{ %
      (\Gam \vdash P_1) \xot{\alpha} (\Gam \vdash P'_1) \\
      (\Gam \vdash P_2) \xot{\overline{\alpha}} (\Gam \vdash P'_2) % %
    }{% %
      (\Gam \vdash P_1 \para P_2) \xot{\id} (\Gam \vdash P_1 \para P'_2)
    } %
  \end{mathpar}
  \caption{CCS transitions}
  \label{fig:ccs}
\end{figure}
We view terms as a graph $\ccs{}$ over $\A$ with the usual transition
rules, as recalled in Figure~\ref{fig:ccs} (which is an inductive
definition).  There, we let $\overline{\alpha}$ denote $\abar$ when
$\alpha = a$, or $a$ when $\alpha = \abar$.
\begin{rem}
  The graph $\A$ only has `endo'-edges, hence only relates terms with
  the same set of free channels. Some \ltss{} below do use more
  general graphs.
\end{rem}

Let us finally recall the definition of fair testing equivalence.
Let $\bot$ denote the set of processes $P$ such that 
for all paths $P \xxOt{\A}{} P'$, there exists a path 
$P' \xxOt{\A}{\tick} P''$.
\begin{defi}\label{def:ccsfair}
  A \emph{test} for $\Gam \vdash P$ is any process $\Gam \vdash Q$.  A
  test $Q$ is \emph{passed} by $P$ when $(\Gam \vdash P \para Q) \in
  \bot$.  Two processes $\Gam \vdash P$ and $\Gam' \vdash P'$ are
  \emph{fair testing equivalent}, notation $(\Gam \vdash P) \faireqs
  (\Gam' \vdash P')$, iff $\Gam = \Gam'$ and $P$ and $P'$ pass exactly
  the same tests.
\end{defi}

\section{Summary of previous work}\label{sec:HP}
In this section, we recall some material from \citetalias{2011arXiv1109.4356H}. Apart from the
admittedly numerous prerequisites mentioned in the previous section,
the paper should be self-contained, although the material in this
section would usefully be complemented by reading \citetalias{2011arXiv1109.4356H}. 

As sketched in the introduction, we construct a
multi-player game, consisting of positions and plays between
them. Positions are certain graph-like objects, where vertices
represent players and channels.  But what might be surprising is that
moves are not just a binary relation between positions, because we not
only want to say \emph{when} there is a move from one position to
another, but also \emph{how} one moves from one to the other. This
will be implemented by viewing moves from $X$ to $Y$ as \emph{cospans}
$Y \xto{s} M \xot{t} X$ in a certain category $\Chatf$ of
higher-dimensional graph-like objects, or `string diagrams', where $X$
and $Y$ respectively are the initial and final positions, and $M$
describes how one goes from $X$ to $Y$.  By composing such moves (by
pushout), we get a bicategory $\Dccsv$ of positions and plays. This is
described in Sections~\ref{subsec:diagrams}--\ref{subsec:plays}.  In
Section~\ref{sec:playgrounds}, we will equip this bicategory with more
structure, namely that of a pseudo double category, where one
direction models dynamics, and the other models space, e.g., the
inclusion of a position into another.  Section~\ref{subsec:strats:old}
further recalls our two notions of strategies derived from the game
(behaviours and innocent strategies, respectively), and
Section~\ref{subsec:fair} recalls our semantic variant of fair testing
equivalence.

\subsection{Diagrams}\label{subsec:diagrams}
In preparation for the definition of our base category $\C$, recall
that (directed, multi) graphs may be seen as presheaves over the
category freely generated by the graph with two objects $\star$ and
$[1]$, and two edges $s,t \colon \star \to [1]$. Any presheaf $G$
represents the graph with vertices in $G(\star)$ and edges in $G[1]$,
the source and target of any $e \in G[1]$ being respectively $e \cdot
s$ and $e \cdot t$. A way to visualise how such presheaves represent
graphs is to compute their \emph{categories of
  elements}~\cite{MM}. Recall that the category of elements $\elements G$
for a presheaf $G$ over $\C$ has as objects pairs $(c,x)$ with $c \in
\C$ and $x \in G(c)$, and as morphisms $(c,x) \to (d,y)$ all morphisms
$f \colon c \to d$ in $\C$ such that $y \cdot f = x$. This category
admits a canonical projection functor $\pi_G$ to $\C$, and $G$ is the colimit of
the composite $\elements G \xto{\pi_G} \C \xto{\yoneda} \Chat$ with the
Yoneda embedding. E.g., the category of elements for $\yoneda[1]$ is
the poset $(\star, s) \xto{s} ([1],\id_{[1]}) \xot{t} (\star, t)$,
which could be pictured as
\diagramme[stringdiag={0.1}{0.6}]{baseline=(A.south)}{%
  \path[-,draw] %
  (A) edge (E) %
  (B) edge (E) %
  ; %
  \node at ($(B.south east) + (.1,0)$) {,} ;%
}{%
  \joueur{A} \& \node[regular polygon,anchor=center,regular polygon
  sides=3,fill,minimum size=3pt,draw,rotate=-90] (E) {}; \&
  \joueur{B} %
}{%
} \hspace*{-.7em} where dots represent vertices, the triangle
represents the edge, and links materialise the graph of $G(s)$ and
$G(t)$, the convention being that $t$ goes from the apex of the
triangle.  We thus recover some graphical intuition.

Our string diagrams will also be defined as (finite) presheaves over
some base category $\C$. Let us give the formal definition of $\C$ for
reference.  We advise to skip it on first reading, as we then attempt
to provide some graphical intuition.
\begin{figure}[t]
    \begin{mathpar}
      {\begin{minipage}[t]{0.33\textwidth}
        \centering
        {\diag(.4,.4){%
            \&|(v)| v \\
            |(n)| [n] \& \& |(n')| [n'] \\
            \& |(star)| \star %
          }{%
            (star) edge[labelbl={s_i}] (n) %
            edge[labelbr={s_i}] (n') %
            (n) edge[labelal={t}] (v) %
            (n') edge[labelar={s}] (v) %
          }} \\
        ($\forall n \in \Nat, i \in n,  v \in {\cup_{a \in n}} \{\forkln,\linebreak \forkrn,
        \tickn, \inna, o_{n,a}, \nun\}$)
        % ($\forall$ $i,j \in n$, $v \in \{\forkln, \forkrn, 
        % \tickn, \linebreak\inni, \outnij, \nun\}$) 
      \end{minipage}}
  \and
  {\begin{minipage}[t]{0.2\textwidth}
        \centering
        {\diag(.4,.4){%
      \&|(v)| \forkn \\
     |(n)| \forkln \& \& |(n')| \forkrn \\
      \& |(star)| [n] %
    }{%
      (star) edge[labelbl={t}] (n) %
       edge[labelbr={t}] (n') %
       (n) edge[labelal={l}] (v) %
       (n') edge[labelar={r}] (v) %
    }} \\
  ($\forall$ $n$)
  \end{minipage}}
  \and 
  {  \begin{minipage}[t]{0.33\textwidth}
    \centering
    {\diag(.4,.3){%
        |(n)| [m]   \&  |(sender)| o_{m,c} \\
        |(star)| \star \& |(v)| \tau_{n,a,m,c} \\
        |(n')| [n] \& |(receiver)| \inna %
        \& %
      }{%
        (star) edge[labell={s_c}] (n) %
        edge[labell={s_a}] (n') %
        (n) edge[labela={t}] (sender) %
        (n') edge[labelb={t}] (receiver) %
        (sender) edge[labelr={\epsilon}] (v) %
        (receiver) edge[labelr={\rho}] (v) %
      }} \\
    ($\forall$ $n \in \Nat, a \in n$, and $c \in m$)
    \end{minipage}}%
  \end{mathpar}%
  \caption{Equations for $\C$}
  \label{fig:equationsC}
\end{figure}

\begin{defi}
  % i -> a
  % l -> b
  % j -> c 
  % k -> d
  Let $G_{\C}$ be the graph with, for all $n, m \in \Nat$, $a \in n$, and $c \in m$:
  \begin{itemize}
  \item vertices $\star$, $[n]$, $\forkln$, $\forkrn$, $\forkn$,
    $\nun$, $\tickn$, $\inna$, $o_{n,a}$, and $\tau_{n,a,m,c}$;
  \item edges $s_1,...,s_n : \star \to [n]$;
  \item for all $v \in \ens{\forkln,\forkrn,\tickn,\inna,o_{n,a }}$, edges
    $s,t : [n] \to v$;
  \item edges $[n] \xto{t} \nun \xot{s} [n+1]$;
  \item edges $\forkln \xto{l} \forkn \xot{r} \forkrn$;
  \item edges $\inna \xto{\rho} \tau_{n,a,m,c} \xot{\epsilon} o_{m,c}$.
  \end{itemize}

  Let $\C$ be the free category on $G_{\C}$, modulo the equations in
  Figure~\ref{fig:equationsC}, where, in the left-hand one, $n'$ is $n+1$
  when $v = \nun$, and $n$ otherwise.
% , with instances of the left-hand square 
  % for all $i \in n$, $j,k \in m$, and $v \in \{\forkln, \forkrn, 
  % \linebreak\tickn, \inni, \outmjk, \nun\}$, of the middle square for 
  % all $n$, and of the right-hand square for all $l \in n$, $j,k \in 
  % m$. (In the left-hand one, $n'$ is $n$ when $v = \nun$ and $n$ 
  % otherwise.) 
% %\  $t \rond s_i = s \rond s_i$ (in $\C(\star,v)$ ), 
% \item plus $l \rond t = 
%   r \rond t$ in $\C([n],\forkn)$, $\rho \rond t \rond s_i = \epsilon 
%   \rond t \rond s_j$ and $\rho \rond s \rond s_{n+1} = \epsilon \rond 
%   s \rond s_k$ . 
% \end{itemize} 
\end{defi}
Our category of string diagrams will be the category $\Chatf$ of finite
presheaves on $\C$.

\begin{wrapfigure}{r}{0pt}
  \begin{minipage}[t]{0.3\linewidth}
    \centering
    \diagramme[stringdiag={.8}{1.3}]{}{%
      % \path[-,draw] %\
      % (a) edge (j1) %\
      % (c) edge (j1) %\
      % (b) edge (j1) %\
      % ; %\
    }{%
      \node (s_1) {$(\star, s_1)$}; \& \node (s_2) {$(\star, s_2)$}; \& \node (s_3) {$(\star, s_3)$}; \\
      \& \node (id) {$([3], \id_{[3]})$}; % \\
      % \canal{a} \& \canal{b} \& \canal{c} \\
      % \& \joueur{j1}
    }{%
      (s_1) edge (id) %
      (s_2) edge (id) %
      (s_3) edge (id) %
    }
    \\
    \diagramme[stringdiag={.8}{1.3}]{}{%
      \path[-,draw] %
      (a) edge (j1) %
      (c) edge (j1) %
      (b) edge (j1) %
      ; %
    }{%
      \canal{a}     \& \canal{b} \&  \canal{c} \\
      \& \joueur{j1} }{%
    }
  \end{minipage}
\end{wrapfigure}
To explain this seemingly arbitrary definition, let us compute a few
categories of elements. Let us start with an easy one, that of $[3]
\in \C$ (we implicitly identify any $c \in \C$ with $\yoneda c$). An
easy computation shows that it is the poset pictured in the top part
on the right. We will think of it as a position with one player
$([3],\id_{[3]})$ connected to three channels, and draw it as in the
bottom part on the right, where the bullet represents the player, and
circles represent channels.  The \emph{positions} of our game are
finite presheaves empty except perhaps on $\star$ and $[n]$'s. Other
objects will represent moves.  The graphical representation is
slightly ambiguous, because the ordering of channels known to players
is implicit.  We will disambiguate in the text when necessary.  A
\emph{morphism of positions} is an injective morphism of presheaves.
The intuition for a morphism $X \to Y$ between positions is thus that
$X$ embeds into $Y$.
\begin{defi}\label{def:Dh}
  Positions and morphisms between them form a category $\Dccsh$.
\end{defi}

A more difficult category of elements is that of $\paraof{2}$. It is
the poset generated by the graph on the left (omitting base
objects for conciseness):
  \begin{mathpar}
    \diag (.4,.3) {%
      % |(lt1)| l s s_1 = r s s_1 %
      \& \& |(lt)| l s \& \& |(rt)| r s \& \& \\ % |(lt2)| l s s_2 = r s s_2 \\
      |(lt1)| l s s_1 \& \& |(l)| l \& |(para)| \id_{\paraof{2}} \& |(r)| r \& \& |(lt2)| l s s_2 \\ 
      % |(ls1)| l t s_1 = r t s_1 %
      \& \& \& |(ls)| l t = r t \&  \& \& % |(ls2)| l t s_2 = r t s_2 %
    }{%
      (lt1) %      (ls1) 
      edge (ls) %
      (lt2) %       (ls2) 
      edge (ls) %      
      (ls) edge (l) edge (r) %
      (lt) edge (l) %
      (rt) edge (r) %
      (l) edge (para) %
      (r) edge (para) %
      (lt1) edge[identity] (lt1) %(ls1) %
      edge (lt) %
      edge[fore,bend left=10] (rt) %
      (lt2) edge[identity] (lt2) % (ls2) %
      edge (rt) %
      edge[bend right=10,fore] (lt) %
    }
    \and
          \diagramme[stringdiag={.3}{.6}]{}{
%     \node[coordinate] (inter) at (intersection cs: %
%     first line={(t_1)-- (t1)}, %
%     second line={(para)-- (t_2)}) {} ; %
%     \path[-,dashed] (para) edge (inter) ; %
%     \path[-] (inter) edge (t_2) ; %
%    \node[diagnode,at= (s1.south east)] {\ \ \ .} ; %
    \node[diagnode,at= (t1.south east)] {\ \ \ .} ; %
  }{%
     \& \& \joueur{t_1} \&  \& \joueur{t_2} \\
    % \canal{t0} \& \& \& \& \& \& \& \canal{t1} \& \& \\ %\node[diagnode] (Y) {$Y$} ; \\
    \& \&   \&  \\
    \& \ \& \\
    \canal{t0} \& \& \& \couppara{para} \& \& \& \canal{t1} \\ % \node[diagnode] (M) {$M$} ; \\
    \& \ \& \\
    \& \&  \\
    % \canal{s0}
    \& \& \& \joueur{s} \& \& \& \& % \canal{s1} \& \& %\node[diagnode] (X) {$X$} ; 
  }{%
    (para) edge[-] (t_2) %
    (t1) edge[-,bend right=10] (t_2) %
    (t0) % (s0)
    edge[-] (s) %
    (t1) % (s1)
    edge[-] (s) %
    (s) edge[-] (para) %
    (para) edge[fore={.3}{.3},-] (t_1)
    (t0) edge[fore={.5}{.5},-,bend left=10] (t_2) %
    (t0) edge[-,bend left=15] (t_1) %
    (t1) edge[-,fore={1}{.5},bend right=10] (t_1) %
%    (i0) edge[-,gray,very thin] (para) %
%    (i1) edge[-,gray,very thin] (para) %
%    (Y) edge[into] (M) %
%    (X)  edge[linto] (M) %
  }  
\end{mathpar}
We think of it as a binary player ($l t$) forking into two players
($l s$ and $r s$), and draw it as on the right.
% The vertical edges 
% on the outside are actually identities: the reason we draw separate 
% vertices is to identify the top and bottom parts of the picture as 
% the respective images of both legs of the following cospan.  First, 
% consider the inclusion $[2] \para [2] \into \paraof{2}$: its domain 
% is any pushout of $[s_1,s_2] \colon (\star + \star) \to [2]$ with 
% itself, i.e., the position consisting of two binary players sharing 
% their channels; and the inclusion maps it to the top part of the 
% picture.  Similarly, we have a map $[2] \into \paraof{2}$ given by 
% the player $l t$ and its channels (the bottom part). The cospan 
% $[2]\para [2] \to \paraof{2} \ot [2]$ is called the \emph{local fork 
% move} of arity 2. 
\newcommand{\longueurfigun}{.6}\newcommand{\separation}{} The
graphical convention is that a black triangle stands for the
presence of $\id_{\forkof{2}}$, $l$, and $r$. Below, we represent
just $l$ as a white triangle with only a left-hand branch, and
symmetrically for $r$.  Furthermore, in all our pictures, time flows
`upwards'.

  Another category of elements, characteristic of CCS,
  is the one for synchronisation $\tau_{n,a,m,c}$. The case $(n,a,m,c) =
  (2,1,3,2)$ is the poset generated by the graph on the left of
  Figure~\ref{fig:tau}, which we will draw as on the right. %
  \begin{figure*}[t]
    \begin{mathpar}
      \diagramme[diag={.4}{.4}]{}{%
        \path[->] %
        (t1) edge (s) %
        (t2) edge (s') %
        (t0) edge (t) %
        (t0) edge[bend right=20] (s) %
        (t3) edge (t) %
        (t3) edge (s) %
        (t2) edge (t') %
        (t1) edge (s') %
%        (t') edge (t2) %
%        (t2) edge[->] (iota') %
        ; %
        \path[->,draw] %
        (t1) edge[fore={.3}{.5}] (iota) %
        (t1) edge[fore={0.3}{.5}] (iota') %
        ; %
        \path[draw,->] %(t2) edge[->,bend right=10] (iota) %
        (iota) edge[fore={.4}{0}] (tau) %
        (iota') edge[fore={.4}{0}] (tau) %
        ; %
        % \path (iota') -- (t2) node[coordinate,pos=0.55] (iotatip')
        % {}
        % ; %
        % \path[-] (iotatip) edge[bend right=7,-latex] (iota) ; %
        % \path[-] (iota') edge[-latex] (iotatip') ; %
        \path[->] (t1) edge[fore={.6}{.3},bend right=20] (t') %
        ; %
        \path[->] (t1) edge[fore={.6}{.5},bend left=20] (t) ; %
        \foreach \x/\y in {s/iota,t/iota,s'/iota',t'/iota'} \path[->]
        (\x) edge (\y) ; %
      }{%
        \& \&  |(t)| \epsilon s \& \&   \& \&  |(t')| \rho s   \\
        \ \& \\
        |(t0)[anchor=base west]| \epsilon t s_1 \& \&  |(iota)| \epsilon \& \& |(tau)| {\scriptscriptstyle \id_{\tau_{n,a,m,c}}} \& \& |(iota')| \rho  \& |(t2)| \rho t s_2 \\
        \& |(t3)| \epsilon t s_3 \&  \& \& |(t1)| \epsilon t s_2 \\
        \& \& |(s)| \epsilon t \& \& \& \& |(s')| \rho t % \canal{s'0}
      }{%
      }%
      \hfil
      \diagramme[stringdiag={.8}{.8}]{}{%
        \path[-] %
        (s) edge (t1) %
        % (t0) edge (inter) %
%        (t2) edge (s) %
        (t2) edge (s') %
        (t0) edge (t) %
        (t0) edge[bend right=20] (s) %
        (t3) edge (s) %
        edge (t) %
        (t2) edge (t') %
        (s') edge (t1) %
%        (t') edge (t2) %
%        (iota') edge[gray,very thin] (t2) %
        ; %
        \moveccsin{t1}{iota'}{t2}{1} %
        \moveccsout{t0}{iota}{t1}{1} %
        \twocell[.4][.4]{iota}{t1}{iota'}{}{
          decorate,decoration={snake,amplitude=.3mm,segment length=1mm},bend left=40}
        \path[-] %
        % (iota) edge[fore={.3}{.5},-latex,draw] (t1) %\ 
        % (t1) edge[fore={0.3}{.5},-latex,draw] (iota') %\ 
        ; %
        % \node[diagnode,at= (s'0.south east)] {\ \ \ .} ; %
        % on place le bout de la fleche du canal envoye en se placant
        % a
        % .8 entre le milieu de s1 et t1
        % \framenode{t1} %
        % \path[] (t2) edge[bend right=10] node[coordinate,pos=.6]
        % (iotatip) {} (iota) ; %
 %       \path[] (t2) edge[color=gray,very thin,bend right=10]
 %       node[coordinate,pos=.6] (iotatip) {} (iota) ; %
 %       \path (iota') -- (t2) node[coordinate,pos=0.55] (iotatip') {}
 %       ; %
%        \path[-] (iotatip) edge[bend right=7,-latex] (iota) ; %
%        \path[-] (iota') edge[-latex] (iotatip') ; %
        \path[-] (t1) edge[fore={.3}{.3}] (t') %
        ; %
        \path[-] (t) edge[fore={.1}{.5}] (t1) ; %
        \foreach \x/\y in {s/t,s'/t'} \path[-] (\x) edge (\y) ; %
%        \node[anchor=south] at (t2.north) {$\scriptstyle \beta$} ; %
        \node[anchor=north] at (t1.south) {$\scriptstyle \alpha$} ; %
        \node[anchor=south] at (t.north) {$\scriptstyle x'$} ; %
        \node[anchor=north] at (s.south) {$\scriptstyle x$} ; %
        \node[anchor=south] at (t'.north) {$\scriptstyle y'$} ; %
        \node[anchor=north] at (s'.south) {$\scriptstyle y$} ; %
%        \node[diagnode,at= (s'.south east)] {\ \ \ .} ; %
      }{%
        \& \& \joueur{t} \& \& \& \& %
        \joueur{t'}   \\
        \ \\
        \canal{t0}\& \&  \coupout{iota}{0} \& \& \canal{t1}  \& \&  \coupin{iota'}{0} \& \canal{t2} \\
        \& \canal{t3} \\
        \& \& \joueur{s} \& \& \& \& \joueur{s'} % \canal{s'0}
      }{%
      }%
    \end{mathpar}
    \caption{Category of elements for $\tau_{2,1,3,2}$ and graphical representation}
\label{fig:tau}
\end{figure*} %
The left-hand ternary player $x$ outputs on its $2$nd channel, here
$\alpha$. The right-hand unary player $y$ receives on its $1$st
channel, again $\alpha$. Both players have two occurrences, one before
and one after the move, respectively marked as $x / x'$ and $y / y'$.
Both $x$ and $x'$ have arity $3$ here, and both $y$ and $y'$ have
arity $1$. There are actually three moves, in the sense that
there are three higher-dimensional objects in the corresponding
category of elements.  The first is the output move from $x$ to $x'$,
graphically represented as the left-hand %
\raisebox{.25em}{
\diagramme[ampersand replacement=\&,column sep=.5cm,inner sep=0.1pt]{inner sep=0pt}{%%
          \path[-] %
          (a) edge[-latex] (b) %
          ; %
        }{%
          %\coupout{a}{0} \& \coupout{b}{0} \& \coupout{c}{0} %
          \node[coordinate] (a){}; \& \node[coordinate] (b){}; %
        }{%
        }%
}
(intended to evoke the `ping' sent by $x$ entering channel $\alpha$).
The second move is the input move from $y$ to $y'$, graphically represented
as the right-hand
\raisebox{.25em}{
\diagramme[ampersand replacement=\&,column sep=.5cm,inner sep=1pt]{inner sep=0pt}{%%
          \path[-] %
          (b) edge[-latex] (c) %
          ; %
        }{%
          %\coupout{a}{0} \& \coupout{b}{0} \& \coupout{c}{0} %
          \node[coordinate] (b){}; \& \node[coordinate] (c){}; %
        }{%
        }%
      } %
      (intended to evoke a `ping' exiting channel $\alpha$).  The
      third and final move is the synchronisation itself, which
      `glues' the other two together, as represented by the squiggly
      line.
% The carrier
% channel is marked with thick lines, while the transmitted channel is
% indicated with arrows.  
%Our graphical convention is that the
  %element $\id_{\taunimjk}$ is 

We leave the computation of other categories of elements as an
exercise to the reader. % and cospans.
The remaining diagrams % for $\paralp$, $\pararp$, $\outmcd$, $\inna$,
% $\tickp$, and $\nup$ 
are depicted % below
in the top row of Figure~\ref{fig:stringmoves}, for
$(n,a,m,c) = (2,1,3,2)$. %: \\ \noindent
%\begin{minipage}[t][5em]{\textwidth}
\begin{figure*}[t]
  \centering
  \begin{tabular}{*{8}{c}}
  \diagramme[stringdiag={.2}{.33}]{}{ }{%
    \& \joueur{t_1} \& \& \& \\ %\node[diagnode] (Y) {$Y$} ; \\
    \& \&   \&  \\
    \& \ \& \\
    \canal{t0} \& \& \coupparacreux{para} \& \& \canal{t1}
    \\ % \node[diagnode] (M) {$M$} ; \\
    \& \ \& \\
    \& \&  \\
    \& \& \joueur{s} \& \&
    % \node[diagnode] (X) {$X$} ;
  }{%
    (t0) edge[-] (t_1) %
    (t1) edge[-,bend right=20] (t_1) %
    (t0) edge[-] (s) %
    (t1) edge[-] (s) %
    (s) edge[-] (para) %
    (para) edge[-] (t_1) %
    % (i0) edge[-,gray,very thin] (para) %
    % (i1) edge[-,gray,very thin] (para) %
    % (Y) edge[into] (M) %
    % (X) edge[linto] (M) %
  }
  &
%
  %% Right-forking basic move
  \diagramme[stringdiag={.2}{.33}]{}{ }{%
    \& \& \& \joueur{t_2} \& \\ %\node[diagnode] (Y) {$Y$} ; \\
    \& \&   \\
    \& \ \& \\
    \canal{t0} \& \& \coupparacreux{para} \& \& \canal{t1}
    \\ % \node[diagnode] (M) {$M$} ; \\
    \& \ \& \\
    \& \&  \\
    \& \& \joueur{s} \& \&
    % \node[diagnode] (X) {$X$} ;
  }{%
    (t0) edge[-,bend left=20] (t_2) %
    (t1) edge[-] (t_2) %
    (t0) edge[-] (s) %
    (t1) edge[-] (s) %
    (s) edge[-] (para) %
    (para) edge[-] (t_2) %
    % (i0) edge[-,gray,very thin] (para) %
    % (i1) edge[-,gray,very thin] (para) %
  }
  &
%
  %% Output move
  \diagramme[stringdiag={.3}{.5}]{baseline=($(iota.south)$)}{%
    % \node[coordinate] (inter) at (intersection cs: %\
    % first line={(s)-- (t0)}, %\
    % second line={(s1)-- (t1)}) {} ; %\
    % \path[draw] (s) edge (inter) ; %\
    \path[-] %
    % (s0) edge (inter) %
    (t2) edge (s) %
    (t1) edge (t) %
    (t0) edge (t) %
    (t2) edge (t) %
    (t) edge (iota.west) %
    (s) edge (iota.west) %
    (t0) edge[bend right=20] (s) %
    (t1) edge (s) %
    ; %
    % \path (s1) -- (t1) node[coordinate,pos=.5] (st1) {} ; %
    \moveccsout[1]{t1}{iota}{t2}{1} %
    \foreach \x/\y in {s/t} \path[-] (\x) edge (\y) ; %
  }{%
    \& \& \joueur{t}  \\
    \&  \\
    \canal{t0} \& \& \coupout{iota}{0}  \& \canal{t2} \\
    \& \canal{t1} \\
    \& \& \joueur{s} }{%
  }%
  &
  %% Input move
  \diagramme[stringdiag={.6}{\longueurfigun}]{baseline=($(in.south)$)}{
    \path[-] (a) edge (p) %
    (in) edge (p) edge (p') %
    %edge[gray,very thin] (b) %
    (p') edge (a) edge (b) %
    % (b) edge (p) edge (b') %
    ; %
    % \node[diagnode,at= (b.south east)] {\ \ \ .} ; %
    \moveccsin[.5]{a}{in}{b}{1} %
    \foreach \x/\y in {p/p',a/a,p/b} \path[-] (\x) edge (\y) ; %
  }{ %
    \& \joueur{p'} \& \\ %
    \canal{a} \& \coupout{in}{0} \& \canal{b} \\ %
    \& \joueur{p} % \& \canal{b} %
  }{%
  } %
  &
%
  %% Tick move
  \diagramme[stringdiag={.6}{\longueurfigun}]{}{ \path[-] (a) edge
    (a) %
    edge (p) %
    (tick) edge[shorten <=-1pt] (p) edge[shorten <=-1pt] (p') %
    (p') edge (a) edge (b) %
    (b) edge (p) edge (b) %
    ; %
    % \node[diagnode,at= (b.south east)] {\ \ \ .} ; %
  }{ %
    \& \joueur{p'} \& \\ %
    \canal{a} \& \couptick{tick} \& \canal{b} \\ %
    \& \joueur{p} \& %
  }{%
  } &
  % name creation move
    %
  \diagramme[stringdiag={.3}{.5}]{baseline=($(nu.center)$)}{%
    % \node[coordinate] (inter) at (intersection cs: %\
    % first line={(s)-- (t0)}, %\
    % second line={(t1)-- (t1)}) {} ; %\
    % \path[draw] (s) edge (inter) ; %\
    \path[-,draw] %
    (t1) edge (s) %
    (t1) edge (t) %
    (t0) edge (t) %
    (t2) edge (t) %
    (t) edge (nu) %
    (s) edge (nu) %
    (nu) edge[gray,very thin] (t2) %
    (t0) edge[bend right=20] (s) %
    % (inter) edge (t0) %
    ; %
%    \node[diagnode,at= (s.base east)] {\ \ \ .} ; %
  }{%
      
    \& \& \joueur{t} \& \&  \& \\
    \&  \&    \\
    \canal{t0} \& \& \coupnu{nu} \& \& \canal{t2}  \\
    \& \canal{t1} \\
    \& \& \joueur{s} \& }{%
  }% 
  \\
  \diag(.6,.2){%
    {[n]} \\ {\forkln} \\ {[n]} %
  }{%
    (m-1-1) edge (m-2-1) %
    (m-3-1) edge (m-2-1) %
  }% 
  &
  \diag(.6,.2){%
    {[n]} \\ {\forkrn} \\ {[n]} %
  }{%
    (m-1-1) edge (m-2-1) %
    (m-3-1) edge (m-2-1) %
  }% 
  &
  \diag(.6,.2){%
    {[m]} \\ {o_{m,c}} \\ {[m]} %
  }{%
    (m-1-1) edge (m-2-1) %
    (m-3-1) edge (m-2-1) %
  }% 
  &
  \diag(.6,.2){%
    {[n]} \\ {\inna} \\ {[n]} %
  }{%
    (m-1-1) edge (m-2-1) %
    (m-3-1) edge (m-2-1) %
  }% 
  &
  \diag(.6,.2){%
    {[n]} \\ {\tickn} \\ {[n]} %
  }{%
    (m-1-1) edge (m-2-1) %
    (m-3-1) edge (m-2-1) %
  }% 
  &
  \diag(.6,.2){%
    {[n+1]} \\ {\nun} \\ {[n]} %
  }{%
    (m-1-1) edge (m-2-1) %
    (m-3-1) edge (m-2-1) %
  }% 
  &
\end{tabular}
  \caption{String diagrams and corresponding cospans 
    for $\paraln$, $\pararn$, $o_{m,c}$, $\inna$, $\tickn$, and $\nun$}
\label{fig:stringmoves}
\end{figure*}
%
%  \end{minipage}
  % By 
  % construction, we have morphisms $\paraln \xto{l} \paran 
  % \xot{r} \pararn$ and $\iotaposmj \xto{\sender} \taumjni 
  % \xot{\receiver} \iotanegni$.   
The first two are \emph{views}, in the game semantical sense, of the
fork move $\forkof{2}$ explained above. The next two, $o_{m,c}$ (for
`output') and $\inna$ (for `input'), respectively represent what the
sender and receiver can see of the above synchronisation move.  
% By convention, when an input or output move is not part of any 
% synchronisation, the corresponding arrow is shortened to not quite 
% reach the channel. 
The last two diagrams are a `tick' move, used for
defining fair testing equivalence, and a channel creation move.

\subsection{From diagrams to moves}\label{subsec:moves}
In the previous section, we have defined our category of diagrams as
$\Chatf$, and provided some graphical intuition on its objects.  The
next goal is to construct a bicategory whose objects are positions
(recall: presheaves empty except perhaps on $\star$ and $[n]$'s), and
whose morphisms represent plays in our game. We start in this section
by defining moves as cospans in $\Chatf$, and continue in the next one
by explaining how to compose moves to form plays. Moves are defined in
two stages: \emph{seeds}, first, give the local form for moves, which
are then defined by embedding seeds into bigger positions.

To start with, until now, our diagrams contain no information about
the `flow of time' (although it was mentioned informally for
pedagogical purposes). To add this information, for each diagram $M$
representing a move, we define its initial and final positions, say
$X$ and $Y$, and view the whole move as a cospan $Y \xto{s} M \xot{t}
X$. We have taken care, in drawing our diagrams before, of placing
initial positions at the bottom, and final positions at the top.  We
leave it to the reader to define, based on the above pictures, the
cospans
\begin{mathpar}
    \diag(.6,.2){%
      {[n] \para [n]} \\ {\forkn} \\ {[n]} %
    }{%
      (m-1-1) edge (m-2-1) %
      (m-3-1) edge (m-2-1) %
    }%
    \and
    \diag(.6,.2){%
      {[m] \paraofij{c}{a} [n] } \\ {\tau_{n,a,m,c}} \\ {
        [m] \paraofij{c}{a} [n] } %
    }{%
      (m-1-1) edge (m-2-1) %
      (m-3-1) edge (m-2-1) %
    }%
\end{mathpar}
for forking and synchronisation, plus the ones specified in the bottom
row of \figurename~\ref{fig:stringmoves}.  In these cospans, initial
positions are on the bottom row, and we denote by
$[m] \paraofij{c_1,\ldots,c_p}{a_1,\ldots,a_p} [n]$ the position
consisting of an $m$-ary player $x$ and an $n$-ary player $y$,
quotiented by the equations $x \cdot s_{c_k} = y \cdot s_{a_k}$ for
all $k \in p$. When both lists are empty, by convention, $m=n$ and the
players share all channels in order.
\begin{defi}
  These cospans are called \emph{seeds}. 
%  Their lower legs are called \emph{t-legs}.
\end{defi} 
\begin{rem}
  Such cospans will be used below as the morphisms of a bicategory
  $\Dccsv$, using their lower object as their \emph{target}. Thus, we
  often denote the corresponding leg by $t$ and the other by $s$.  The
  reason for this convention is that it emphasises below that the
  fibration axiom~\axref{fibration} is very close to a universal
  property of pullback~\cite{Jacobs}.
\end{rem}
\begin{rem}
  Both legs of each seed are monic, as will be below both legs of each
  move, and then of each play (because monics are stable under pushout
  in presheaf categories).
\end{rem}

As announced, the moves of our game are obtained by embedding seeds
into bigger positions. This means, e.g., allowing a fork move to occur
in a position with more than one player. We proceed as follows.
\begin{defi}\label{def:interface}
  Let the \emph{interface} of a seed $Y \xto{s} M \xot{t} X$ be $I_X =
  X(\star) \cdot \star$, i.e., the position consisting only of the
  channels of the initial position of the seed.  More generally, an
  \emph{interface} is a position consisting only of channels.
\end{defi}

\begin{wrapfigure}[3]{r}{0pt}
  \begin{minipage}[c][3em]{0.22\linewidth}
  \vspace*{-2em}
        \diag(.3,.6){%
    \&|(I)| I_X \\
   |(X)| Y \&|(M)| M \&|(Y)| X %
  }{%
    (I) edge (X) edge (M) edge (Y) %
    (X) edge (M) %
    (Y) edge (M) %
  }
  \end{minipage}
\end{wrapfigure}
Since channels present in the initial position remain in the final
one, we have for each seed a commuting diagram as on the right.  By gluing any
position $Z$ to the seed along its interface, we obtain a new cospan,
say $Y' \to M' \ot X'$.  I.e., for any injective morphism $I_X \to Z$, we push
$I_X \to X$, $I_X \to M$, and $I_X \to Y$ along $I_X \to Z$ and use
the universal property of pushout, as in:
  \begin{equation}
    \Diag(.02,.6){%
      \pbk[1.1em]{X}{X'}{Z} %
      \pullback[1.1em]{Z}{M'}{M}{draw,-} %
      \pullback[1.1em]{Z}{Y'}{Y}{draw,-} %
    }{%
      \& |(Y)| Y \& \& |(Y')| {Y'} \\
      \& \ \& \\
      \& |(M)| M \& \& |(M')| M'  \\
      |(I)| I_X \&\& |(Z)| Z \\
      \& |(X)| X \& \& |(X')| X'. %
    }{%
      (Z) edge[] (X') %
      edge (M') %
      edge (Y') %
      (I) edge[] (X) %
      edge (Z) %
      edge (M) %
      edge (Y) %
      (Y) edge[fore] (Y') %
      (M) edge[fore] (M') %
      (X) edge (X') %
      (X') edge[dashed] (M') %
      (Y') edge[dashed] (M') %
      (X) edge[fore] (M) %
      (Y) edge[fore] (M) %
    }
    \label{eq:extend}
  \end{equation}
  \begin{defi}
    Let \emph{moves} be all cospans obtained in this way.
  \end{defi}
  Recall that colimits in presheaf categories are pointwise. So, e.g.,
  taking pushouts along injective maps graphically corresponds to
  gluing diagrams together. 
  \begin{exa}\label{ex:forkmove}
    The cospan $[2]\para[2] \xto{[ls,rs]} \forkof{2} \xot{lt} [2]$ has
    as canonical interface the presheaf $I_{[2]} = 2 \cdot \star$,
    consisting of two channels, say $a$ and $b$.  Consider the
    position $[2] + \star$ consisting of a player $y$ with two
    channels $b'$ and $c$, plus an additional channel $a'$. Further
    consider the map $h \colon I_{[2]} \to [2]+\star$ defined by $a
    \mapsto a'$ and $b \mapsto b'$. The pushout
    \begin{mathpar}
      \Diag{%
        \pbk{pi}{M'}{star} %
        }{%
         |(I2)| {I_{[2]}} \& |(star)| {[2]+\star} \\
         |(pi)| {\forkof{2}} \&|(M')| {M'} %
        }{%
          (I2) edge (star) edge (pi) %
          (pi) edge (M') %
          (star) edge (M') %
        } %
        \and \mbox{is} \and
                \diagramme[stringdiag={.3}{.6}]{}{
    \node[diagnode,at= (c.south east)] {\ \ \ .} ; %
        \node[anchor=south] at (t_1.north) {$\scriptstyle x_1$} ; %
        \node[anchor=south] at (t_2.north) {$\scriptstyle x_2$} ; %
        \node[anchor=north] at (s.south) {$\scriptstyle x$} ; %
        \node[anchor=north] at (y.south) {$\scriptstyle y$} ; %
        \node[anchor=north] at (c.south) {$\scriptstyle c$} ; %
        \node[anchor=north] at (t0.south) {$\scriptstyle a=a'$} ; %
        \node[anchor=north] at (t1.south) {$\scriptstyle b=b'$} ; %
  }{%
     \& \& \joueur{t_1} \&  \& \joueur{t_2} \\
    \& \&   \&  \\
    \& \ \& \\
    \canal{t0} \& \& \& \couppara{para} \& \& \& \canal{t1} \& \& \joueur{y} \& \& \canal{c} \\ % \node[diagnode] (M) {$M$} ; \\
    \& \ \& \\
    \& \&  \\
    \& \& \& \joueur{s} \& \& \& \& % \canal{s1} \& \& %\node[diagnode] (X) {$X$} ; 
  }{%
    (para) edge[-] (t_2) %
    (t1) edge[-,bend right=10] (t_2) %
    (t0) % (s0)
    edge[-] (s) %
    (t1) % (s1)
    edge[-] (s) %
    (s) edge[-] (para) %
    (y) edge[-] (t1) %
     edge[-] (c) %
    (para) edge[fore={.3}{.3},-] (t_1)
    (t0) edge[fore={.5}{.5},-,bend left=10] (t_2) %
    (t0) edge[-,bend left=15] (t_1) %
    (t1) edge[-,fore={1}{.5},bend right=10] (t_1) %
  }  
    \end{mathpar}
    % A similar construction with the output seed would yield something like 
    % \begin{center} 
    %   \diagramme[stringdiag={.4}{.6}]{}{%\ 
    %     \path[-,draw] 
    %     (in) edge (s) edge (t) %\ 
    %     (t0) edge (t)  %\ 
    %     (t0) edge (s)  %\ 
    %     (t') edge (t0) %\ 
    %     ; %\ 
    %     \path (in) --  (t0) node[coordinate,pos=.5] (intip) {} ; %\ 
    %     \path[-] (in) edge[-latex] (intip) ; %\ 
    %     \node[diagnode, at = (t'.south east)] {\ \ \ .} ; %\ 
    %   }{%\ 
    %     \joueur{t} \&  \\ 
    %     \coupout{in}{0} \& \canal{t0} \& \joueur{t'} \\ 
    %     \joueur{s} %\ 
    %   }{%\ 
    %   } 
    % \end{center} 
  \end{exa}
We conclude with a useful classification of moves.
\begin{defi}
  A move is \emph{full} iff it is neither a left nor a right fork.  A
  seed is \emph{basic} iff it is neither a full fork nor a
  synchronisation.  We call $\F$ the identity-on-objects subgraph of
  $\Cospan{\Chatf}$ spanning full moves.
\end{defi}
Intuitively, a move is full when its final position contains all
possible avatars of involved players.

\subsection{From moves to plays}\label{subsec:plays}
\begin{wrapfigure}[5]{r}{0pt}
  \begin{minipage}[c]{0.25\linewidth}
    \vspace*{-1.5em}
    \diag(.3,1){%
      \&|(U)| U \\
      |(X)| X \& \&|(Y)| Y \\
      \&|(V)| V %
    }{%
      (X) edge (U) edge (V) %
      (Y) edge (U) edge (V) %
      (U) edge (V) %
    }
  \end{minipage}
\end{wrapfigure}
Having defined moves, we now define their composition to construct our
bicategory $\Dccsv$ of positions and plays.  $\Dccsv$ will be a
sub-bicategory of $\Cospan{\Chatf}$, the bicategory which has as
objects all finite presheaves on $\C$, as morphisms $X \to Y$ all
cospans $X \to U \ot Y$, and as 2-cells $U \to V$ all commuting
diagrams as on the right.  Composition is given by pushout, and hence
not strictly associative. 

% \begin{rem} We choose to view the initial 
%   position as the \emph{target} of the morphism in $\Cospan{\Chatf}$, 
%   in order to emphasise below that the fibration 
%   axiom~\axref{fibration} is very close to a universal property 
%   of pullback~\cite{Jacobs}. 
% \end{rem} 

\begin{defi}
  Let $\Dccsv$ denote the locally full subbicategory of
  $\Cospan{\Chatf}$ with positions as objects, whose morphisms,
  \emph{plays}, are either equivalences or isomorphic to some
  composite of moves. 
\end{defi}
We denote morphisms in $\Cospan{\Chatf}$ with special arrows $X \proto
Y$; composition and identities are denoted with $\vrond$ and $\idv$
(recalling the notation for vertical morphisms in a pseudo double
category in Section~\ref{subsec:prelim:cats}).

Again, composition by pushout glues diagrams on top of each other.
\begin{exa}\label{ex:forkplay}
Composition features some concurrency.
  Composing the move of Example~\ref{ex:forkmove} with a forking
  move by $y$ yields
    \begin{center}
                \diagramme[stringdiag={.3}{.6}]{}{
    \node[diagnode,at= (c.south east)] {\ \ \ .} ; %
        \node[anchor=south] at (t_1.north) {$\scriptstyle x_1$} ; %
        \node[anchor=south] at (t_2.north) {$\scriptstyle x_2$} ; %
        \node[anchor=south] at (y_1.north) {$\scriptstyle y_1$} ; %
        \node[anchor=south] at (y_2.north) {$\scriptstyle y_2$} ; %
        \node[anchor=north] at (s.south) {$\scriptstyle x$} ; %
        \node[anchor=north] at (y.south) {$\scriptstyle y$} ; %
        \node[anchor=north] at (c.south) {$\scriptstyle c$} ; %
        \node[anchor=north] at (t0.south) {$\scriptstyle a=a'$} ; %
        \node[anchor=north] at (t1.south) {$\scriptstyle b=b'$} ; %
  }{%
     \& \& \joueur{t_1} \&  \& \joueur{t_2} \& \& \& \& \joueur{y_1} \& \& \joueur{y_2} \\
    \& \&   \&  \\
    \& \ \& \\
    \canal{t0} \& \& \& \couppara{para} \& \& \& \canal{t1} \& \&  \& \couppara{para'} \&  \& \& \canal{c} \\
    \& \ \& \\
    \& \&  \\
    \& \& \& \joueur{s} \& \& \& \& \& \& \joueur{y} % \canal{s1} \& \& %\node[diagnode] (X) {$X$} ; 
  }{%
    (t1) edge[-,bend right=10] (t_2) %
    (t0) % (s0)
    edge[-] (s) %
    (t1) % (s1)
    edge[-] (s) %
    (para) edge[-] (t_2) %
    (s) edge[-] (para) %
    (y) edge[-] (t1) %
     edge[-] (c) %
    (para) edge[fore={.3}{.3},-] (t_1)
    (t0) edge[fore={.5}{.5},-,bend left=10] (t_2) %
    (t0) edge[-,bend left=15] (t_1) %
    (t1) edge[-,fore={1}{.5},bend right=10] (t_1) %
    (c) edge[-,bend right=10] (y_2) %
    (para') edge[-] (y_2) %
    (y) edge[-] (para') %
    (y) edge[-] (t1) %
     edge[-] (c) %
    (para') edge[fore={.3}{.3},-] (y_1) %
    (t1) edge[fore={.5}{.5},-,bend left=10] (y_2) %
    (t1) edge[-,bend left=15] (y_1) %
    (c) edge[-,fore={1}{.5},bend right=10] (y_1) %
  }  
    \end{center}
\end{exa}
\begin{exa}
  Composition retains causal dependencies between moves. To see this,
  consider the following diagram.  In the initial position, there are
  channels $a$ and $b$, plus three players $x(b), y(a,b)$, and
  $z(a)$ (we indicate the channels known to each player in
  parentheses). In a first move, $x$ outputs on $b$, while $y$ inputs. In
  a second move, $z$ outputs on $a$, while (the avatar $y'$ of) $y$
  inputs. The fact that $y$ first inputs on $b$ then on $a$ is encoded in 
  the corresponding diagram, which looks like the
  following:
\begin{center}
      \diagramme[stringdiag={.8}{1.3}]{}{%
%        \framenode{a} %
%        \framenode{a'} %
        % \circlenode{c} %\ 
        % \circlenode{c'} %\ 
        %\synchronisation{a}{o}{b}{i}{a'}{.4}{.4}{}{}{}{} %
        \moveccsin{b}{i}{a'}{1} %
        \moveccsout{a}{o}{b}{1} %
        \moveccsout{b}{o'}{a'}{1} % 
        \moveccsin[fore={.1}{.1}]{a'}{i'}{y''}{1} % 
        %\synchronisation{c}{o'}{a'}{i'}{c'}{.4}{.4}{}{}{}{} % 
       \node[diagnode,at= (z.south east)] {\ \ \ \ .} ; %
%        \node[below=1pt] at (a.south) {$\scriptstyle a$} ; %
        \node[below=1pt] at (b.south) {$\scriptstyle b$} ; %
%        \node[below=1pt] at (c.south) {$\scriptstyle c$} ; %
        \node[below=1pt] at (a'.south) {$\scriptstyle a$} ; %
%        \node[right=1pt] at (c') {$\scriptstyle c$} ; %
        \node[below=1pt] at (x.south) {$\scriptstyle x$} ; %
        \node[below=1pt] at (y.south) {$\scriptstyle y$} ; %
        \node[below=1pt] at (z.south) {$\scriptstyle z$} ; %
        \node[above left] at (y') {$\scriptstyle y'$} ; %
        \twocell{o}{b}{i}{}{
          decorate,decoration={snake,amplitude=.3mm,segment length=1mm},bend left=50}
        \twocell[.4][.25]{o'}{a'}{i'}{}{
          decorate,decoration={snake,amplitude=.3mm,segment length=1mm},bend right=30}
      }{%
         \& \& \joueur{y''} \\
         \& \& \coupin{i'}{0} \& \\ %\canal{c'} \\
         \joueur{x'} \& \& \joueur{y'} \& \& \joueur{z'} \\
         \coupout{o}{0} \& \canal{b} \& \coupin{i}{0} \& \canal{a'} \& \coupout{o'}{0} \\
         \joueur{x} \& \& \joueur{y} \& \& \joueur{z} %
  }{%
    % canaux -> joueurs
    % (a) edge[-] (x) %
    % edge[-] (x') %
    (b) edge[-] (x) %
    edge[-] (x') %
    edge[-] (y) %
    edge[-] (y') %
    edge[-] (y'') %
    (a') edge[-] (y') %
    edge[-] (y) %
    edge[-,bend right=10] (y'') %
    edge[-] (z) %
    edge[-] (z') %
    % (c) edge[-] (z) %
    % edge[-] (z') %
%    (c') edge[-] (y'') %
    % source -> coup -> but
    (i) edge[-] (y) %
    edge[-] (y') %
    (i') edge[-] (y') %
    edge[-] (y'') %
    (o) edge[-] (x) %
    edge[-] (x') %
    (o') edge[-] (z) %
    edge[-] (z') %
  }
\end{center}  
\end{exa}

\subsection{Behaviours and strategies}\label{subsec:strats:old}
\subsubsection{Behaviours}\label{subsubsec:behaviours}
Recall from \citetalias{2011arXiv1109.4356H} the category $\E$ 
\begin{center}
  \begin{tabular}[t]{p{.77\textwidth}p{.23\textwidth}}
    \vspace*{-2.3em}\begin{itemize}
  \item whose objects are maps $U \ot X$ in $\Chatf$, such that there
    exists a play $Y \to U \ot X$, i.e., objects are plays, where we
    forget the final position;
  \item and whose morphisms $(U \ot X) \to (U' \ot X')$ are commuting
    diagrams as on the right with all arrows monic.
  \end{itemize}
  &
    \diag|baseline= (m-1-1.base) |{%
      U \& U' \\
      X \& X' %
    }{%
      (m-1-1) edge[labelu={}] (m-1-2) %
      (m-2-1) edge[labell={}] (m-1-1) %
      (m-2-1) edge[labeld={}] (m-2-2) %
      (m-2-2) edge[labelr={}] (m-1-2) %
    }
  \end{tabular}
\end{center}
Morphisms $(U \ot X) \to (U' \ot X')$ in $\E$ represent extensions of $U$, both
spatially (i.e., embedding into a larger position) and dynamically
(i.e., adding more moves).

We may relativise this category $\E$ to a particular position $X$,
yielding a category $\E(X)$ of plays on $X$ as follows. Consider the
functor $\cod \colon \E \to \Dccsh$ mapping any play $U \ot X$ to its
initial position $X$, and consider the pullback category $\E(X)$ as
defined in Section~\ref{subsec:prelim:cats}. The objects of $\E(X)$
are just plays $(U \ot X)$ on $X$, and morphisms are morphisms of
plays whose lower border is $\id_X$.  This yields the definition of a
category of `naive' strategies, called behaviours.
\begin{defi}\label{def:behccs}
  The category $\Beh{X}$ of \emph{behaviours} on $X$ is the category
  $\FPsh{\E(X)}$ of presheaves of finite sets on $\E (X)$.
\end{defi}\enlargethispage{2\baselineskip}
Behaviours suffer from the deficiency of allowing unwanted cooperation
between players. 
\begin{exa}\label{ex:noninnocent}
  Consider a position $X$ with three players $x,y,z$ sharing a
  channel $a$, and the following plays on it: in $\trasse_{x,y}$, $x$
  outputs on $a$, and $y$ inputs; in $\trasse_{x,z}$, $x$ outputs
  on $a$, and $z$ inputs; in $i_z$, $z$ inputs on $a$. One may
  define a \stratglobale $S$ mapping $\trasse_{x,y}$ and $i_z$ to a
  singleton, and $\trasse_{x,z}$ to $\emptyset$. Because $\trasse_{x,y}$ is
  accepted, $x$ accepts to output on $a$; and because $i_z$ is
  accepted, $z$ accepts to input on $a$. The problem is that $S$
  rejecting $\trasse_{x,z}$ roughly amounts to $x$ refusing to synchronise
  with $z$, or conversely.
\end{exa}

\subsubsection{Strategies}\label{subsubsec:strategies}
To rectify this, we consider the following notion of view:
\begin{defi}
  Let $\EVi$ denote the full subcategory of $\E$ consisting of \emph{views}, i.e.,
  composites of basic seeds.
\end{defi}
We relativise views to a position $X$ by considering the comma
category $\EVi_X$ as defined in Section~\ref{subsec:prelim:cats}.  Its
objects are pairs of a view $V \ot [n]$ on a single $n$-ary player,
and an embedding $[n] \into X$, i.e., a player of $X$.

    \begin{defi}
     The category $\SS_X$ of \emph{strategies} on $X$ is the category
     $\OPsh{\EVi_X}$ of presheaves of finite ordinals on $\EVi_X$.
    \end{defi}
   \begin{rem}
     We could here replace finite ordinals with a wider category and
     still get a valid semantics. But then to show the correspondence
     with the syntax below we would work with the subcategory of
     presheaves of finite ordinals.
   \end{rem}
    
   This definition of strategies rules out undesired behaviours. We
   now sketch how to map strategies to behaviours (this is done in
   more detail for arbitrary playgrounds below): let first $\E_X$ be
   the category obtained by taking a comma category instead of a
   pullback in the definition of $\E (X)$.  Then, embedding
   $\OPsh{\EVi_X}$ into $\FPsh{\EVi_X}$ via $\ford \into \set$,
   followed by right Kan extension to $\op{\E_X}$ followed by
   restriction to $\op{\E(X)}$ yields a functor $\extafun{X} \colon
   \SS_X \to \Beh{X}$. The image of a strategy $S$ may be computed as in
    \begin{center}
      \diag{%
        \op{(\EVi_X)} \& \op{\E_X} \& \op{\E(X)} \\
        \ford \& \set, %
      }{%
        (m-1-1) edge[labell={S}] (m-2-1) %
        edge[into] (m-1-2) %
        (m-2-1) edge[into] (m-2-2) %
        (m-1-2) edge[labell={S'}] (m-2-2) %
        (m-1-3) edge[linto] (m-1-2) %
        edge[labelbr={\exta{X}{S}}] (m-2-2)
      }
    \end{center}
    where $S'$ is here obtained by right Kan extension (the embedding
    $\op{(\EVi_X)} \into \op{\E_X}$ being full and faithful, we may
    choose the diagram to strictly commute).
    By the standard formula for right Kan extensions as
    ends~\citep{MacLane:cwm} we have, for any $S \colon \op{(\EVi_X)} \to \ford$: 
    $$\exta{X}{S} (U) =
    \int_{v \in \EVi_X} S (v)^{\E_X (v,U)}.$$ If $S$ is boolean, i.e.,
    takes values in $\ens{\emptyset,1}$, then the involved end may be
    viewed as a conjunction, saying that $U$ is accepted by $\exta{X}{S}$
    whenever all its views are accepted by $S$.  Equivalently,
    $\exta{X}{S} (U)$ is a limit of $\op{(\EVi_X / U)} \xto{\dom}
    \op{(\EVi_X)} \xto{S} \ford \into \set.$

% Not sure because of the ford/set distinction
    % Finally, $\extafun{X}$ admits a left adjoint, which we might
    % call `innocentisation', because it maps `naive strategies'
    % (our behaviours) to `innocent ones' (our strategies). 

\subsubsection{Decomposition: a syntax for strategies}\label{subsubsec:syntax}
Our definition of strategies is rather semantic in flavour. Indeed,
presheaves are akin to domain theory. However, they also lend
themselves well to a syntactic description (unlike behaviours). Again,
this is treated at length in the abstract setting below, so we here
only sketch the construction.

First, it is shown in \citetalias{2011arXiv1109.4356H} that strategies on an arbitrary position $X$
are in 1-1 correspondence with families of strategies indexed by the
players of $X$. Recall that $[n]$ is the position consisting of one
$n$-ary player. A player of $X$ is the same as a morphism $[n] \to X$
(for some $n$) in $\Dccsh$. Thus, we define the set $\Pl (X) = \sum_{n
  \in \Nat} \Dccsh ([n],X)$ of players of $X$.
\begin{prop}
We have  $\SSX \iso \prod_{(n,x) \in \Pl (X)} \SSn$. % and $\EVi_{[n]} \iso \EVi [n]$.
For any $S \in \SSX$, we denote by $S \cdot x$ the component corresponding
to $x \in \Pl (X)$ under this isomorphism.
\end{prop}
So, strategies on arbitrary positions may be entirely described by
strategies on `typical' players $[n]$. As an important particular
case, we may let two strategies interact along an interface (recall
from Definition~\ref{def:interface} that this means a position
consisting only of channels), which will be the basis of our semantic
definition of fair testing equivalence. We proceed as follows.
Consider any pushout $Z$ of $X \ot I \to Y$ where $I$ is an
interface. We have
\begin{cor}
  $\SS_Z \iso \SSX \times \SSY$.
\end{cor}
\begin{proof}
  We have $\EVi_Z \iso \EVi_X + \EVi_Y$, and conclude by universal
  property of coproduct.
\end{proof}
We denote by $[S,T]$ the image of $(S,T) \in \SSX \times \SSY$ under
this isomorphism.

Having shown how strategies may be decomposed into strategies on
`typical' players $[n]$, we now explain that strategies on such
players may be further decomposed. First, we observe that $\EVi_{[n]}$
is isomorphic to the full subcategory $\EVi([n])$ of $\E([n])$
spanning views.  For any strategy $S$ on $[n]$ and seed $b \colon [n']
\proto [n]$, let the \emph{residual} $S \cdot b$ of $S$ after $b$ be
the strategy playing like $S$ after $b$, i.e., for all $v \in
\EVi_{[n']}$, $(S \cdot b)(v) = S (b \vrond v)$. $S$ is almost
determined by its residuals. The only information missing from the
$S\cdot b$'s to reconstruct $S$ is the set of initial states and how
they relate to the initial states of each $(S \cdot b)$. This may be
taken into account as follows.

\begin{defi}\label{def:restriction}
  For any $S \in \SS_{[n]}$ and initial state $\state \in S(\idv)$, let $\restr{S}{\state}$,
   the \emph{restriction} of $S$ to $\state$, be determined by
$$\restr{S}{\state} (v) = \ens{\state' \in S(v) \aalt S(!_v)(\state') = \state},$$
where $!_v$ denotes the unique morphism $!  \colon \idv \to v$.
\end{defi}
$S$ is determined by its set $S(\idv)$ of initial
states, plus the map $(\state, b) \mapsto (\restr{S}{\state}
\cdot b)$ sending any $\state \in S(\idv)$ and isomorphism class $b$
of seeds to $\restr{S}{\state} \cdot b$. In other words, we have
for all $n$:
\begin{thm}
   $\SS_{[n]} \iso (\prod_{n' \in \Nat, b \colon [n'] \proto [n]} \SS_{[n']})^\star$.
\end{thm}
Given an element $(D_1,\ldots,D_m)$ of the right-hand side, the
corresponding strategy maps the identity view $\idv$ to $m$, and any
non-identity view $b \vrond v$ on $[n]$ to the sum $\sum_{i \in m}
D_i(b)(v)$.

A closely related result is that strategies on a player $[n]$ are in
bijection with infinite terms in the following typed grammar, with
judgements $n \vdashdefinite D$ and $n \vdash S$, where $D$ is called a
\emph{definite strategy} and $S$ is a \emph{strategy}:
\begin{mathpar}
\inferrule{\ldots \ n_b \vdash S_b \ \ldots \ {(\forall b \colon [n_b] \proto [n] \in \MMMB_n)}
}{
n \vdashdefinite \langle (S_b)_{b \in \MMMB_n} \rangle
}
\and
\inferrule{\ldots \  n \vdashdefinite D_i \  \ldots \ (\forall i \in m)}{
n \vdash \oplus_{i \in m} D_i}~(m \in \Nat).
\end{mathpar}
Here, $\MMMB_n$ denotes the set of all isomorphism classes of seeds
from $[n]$.  This achieves the promised syntactic description of
strategies. We may readily define the translation of CCS processes,
coinductively, as follows. For processes with channels in $\Gam$, we
define
\begin{equation}
\begin{array}[t]{rcl}
  \Transl{\sum_{i \in n} \alpha_i.P_i} & = & \langle b \mapsto
      \oplus_{\ens{i \in n \aalt b = \Transl{\alpha_i}}} \Transl{P_i}
      \rangle \\
      \Transl{\nu a.P} & = & 
    %   \left \langle 
    %     {\begin{array}[c]{lcl} 
    %         \nugam & \mapsto & \Transl{P}  \\ 
    %         \_ & \mapsto & \emptyset  
    %     \end{array}} 
    % \right \rangle \\ 
      \langle
            \nugam  \mapsto  \Transl{P},
            \_  \mapsto  \emptyset 
            \rangle \\
      \Transl{P\para Q} & = & 
    %   \left \langle 
    %     {\begin{array}[c]{lcl} 
    %         \paralgam & \mapsto & \Transl{P}  \\ 
    %         \parargam & \mapsto & \Transl{Q}  \\ 
    %         \_ & \mapsto & \emptyset  
    %     \end{array}} 
    % \right \rangle  
 \langle
            \paralgam  \mapsto  \Transl{P}  ,
            \parargam  \mapsto  \Transl{Q}  ,
            \_  \mapsto  \emptyset 

 \rangle 
\end{array}
\hspace*{2cm} \begin{array}[t]{rcl}
    \Transl{a} & = & \iotaneg{\Gam,a} \\
    \Transl{\abar} & = & \iotapos{\Gam,a} \\
    \Transl{\tick} & = & \tickgam.
\end{array}
\label{eq:traduc}
\end{equation}
For example, $a.P + a.Q + \bar{b}.R$ is mapped to
$$\langle \iotaneg{\Gam,a} \mapsto (\Transl{P} \oplus \Transl{Q}),
\iotapos{\Gam,b} \mapsto \Transl{R}, \_ \mapsto \emptyset \rangle.$$

\subsection{Semantic fair testing}\label{subsec:fair}
The tools developed in the previous section yield the following
semantic analogue of fair testing equivalence.
\begin{defi}\label{def:cw:successful}
  \emph{Closed-world} moves are those generated by some seed among
  $\nun$,$\tickn$,$\paran$, and $\taunimj$. A play is
  \emph{closed-world} when it is a composite of closed-world moves.
  Let a closed-world play be \emph{successful} when it contains a
  $\tick$ move, and \emph{unsuccessful} otherwise. A state $\state \in
  B(U)$ of a behaviour $B \in \Beh{Z}$ over a closed-world play
  $U \ot Z$ is successful when the play $U$ is, and unsuccessful
  otherwise.
\end{defi}

Let then $\bbot_Z$ denote the set of behaviours $B \in
\Beh{Z}$ such that any unsuccessful, closed-world state admits a
successful extension.  Formally:
\begin{defi}\label{def:bbotccs}
  Let $B \in \bbot_Z$ iff, for any unsuccessful, closed-world play $U
  \ot Z$ and $\state \in B (U)$, there exists a successful,
  closed-world $U'$, a morphism $f \colon U \to U'$ in $\E(Z)$, and a state
  $\state' \in B (U')$ such that $\state' \cdot f = \state$.
\end{defi}
Finally, let us say
that a triple $(I,h,S)$, for any $h \colon I \to X$ (where $I$ is an
interface) and $S \in \SSX$, \emph{passes} the test
consisting of a morphism $k \colon I \to Y$ of positions and a
strategy $T \in \SSY$ iff $\exta{Z}{[S,T]} \in \bbot{}_Z$, where $Z$
is the pushout of $h$ and $k$.  Let $(I,h,S)^{\bbot{}}$ denote the set of
all such $(k,T)$.
\begin{defi}
  For any $h \colon I \to X$, $h' \colon I' \to X'$, $S \in \SSX$, and
  $S' \in \SS_{X'}$, $(I,h,S) \faireq (I',h',S')$ iff $I = I'$ and
  $(I,h,S)^{\bbot{}} = (I,h',S')^{\bbot{}}$.
\end{defi}
Obviously, $\faireq$ is an equivalence relation, analogous to standard fair
testing equivalence, which we hence also call (semantic) fair testing
equivalence.

This raises the question of whether the translation $\Translfun$
preserves or reflects fair testing equivalence. The rest of the paper
is devoted to proving that it does both. As announced in the
introduction, this is done by organising the game into a
\emph{playground}, as defined in the next section.

\section{Playgrounds: from behaviours to strategies}\label{sec:playgrounds}

\subsection{Motivation: a pseudo double category}\label{subsec:pseudodouble}
\begin{wrapfigure}[6]{r}{0pt}
  \begin{minipage}[c]{0.25\linewidth}
    \vspace*{-1.5em}
    \diag{%
      Y' \& Y \\
      U' \& U \\
      X' \& X %
    }{%
      (m-1-1) edge[labelu={h}] (m-1-2) %
      (m-2-1) edge[labelo={k}] (m-2-2) %
      (m-3-1) edge[labeld={l}] (m-3-2) %
      (m-1-1) edge[labell={s'}] (m-2-1) %
      (m-1-2) edge[labelr={s}] (m-2-2) %
      (m-3-1) edge[labell={t'}] (m-2-1) %c`
      (m-3-2) edge[labelr={t}] (m-2-2) %
    }
  \end{minipage}
\end{wrapfigure}
We start by organising the game described above into a (pseudo) double
category.  We have seen that positions are the objects of the category
$\Dccsh$, whose morphisms are embeddings of positions.  We have also
seen that positions are the objects of the bicategory $\Dccsv$, whose
morphisms are plays.  It should seem natural to define a pseudo double
category structure with
\begin{itemize}
\item $\Dccsh$ as horizontal category,
\item $\Dccsv$ as vertical bicategory,
\item commuting diagrams as on the right as double cells.
\end{itemize}
Here, $X$ is the initial position and $Y$ is the final one;
all arrows are mono.
  This forms a pseudo double category $\Dccs$, and we have:
\begin{prop}\label{prop:pseudodouble}
  The functor $\codv \colon \DccsH \to \Dccsh$ is a Grothendieck
  fibration~\citep{Jacobs}.
\end{prop}
Intuitively, $\codv$ being a fibration demands some canonical way of
\emph{restricting} a given play on some position $X$ to some
`subposition' $X' \to X$.  More technically, it amounts to the
existence, for all plays $Y \xproto{u} X$ and horizontal morphisms $X'
\xto{l} X$, of a universal ($\approx$ maximal) way of restricting $u$
to $X'$, as on the left below:
\begin{mathpar}
    \diag(1,1){%
      |(X)| {Y'} \& |(Y)| {Y} \\ %
      |(U)| {X'} \& |(V)| {X} %
    }{%
      (X) edge[dashed,labela={h}] (Y) %
      edge[pro,dashed,twol={u'}] (U) %
      (Y) edge[pro,twor={u}] (V) %
      (U) edge[labela={l}] (V) %
      (l) edge[dashed,cell=.3,labela={\alpha}] (r) %
    }
\and     \Diag(.25,.8){% 
    }{% 
      |(U'')| E'' \\ \\ 
      \& |(U')| E' \& \& |(U)| E \\ 
      |(Y'')| p(E'') \\ \\ 
      \& |(Y')| p(E') \& \& |(Y)| p(E).  % 
    }{% 
      (U') edge[labelb={r}] (U) % 
      (Y') edge[labelb={p(r)}] (Y) % 
      (U) edge[serif cm-to,fore,shorten <=.3cm,shorten >=.3cm] (Y) % 
%U'' -> U 
      (U'') edge[bend left=10,labelar={t}] (U) % 
      (Y'') edge[bend left=10,labelar={p(t)}] (Y) % 
      (U'') edge[serif cm-to,fore,shorten <=.3cm,shorten >=.3cm] (Y'') %c` 
%U'' -> U' 
      (U'') edge[dashed,labelbl={s}] (U') % 
      (Y'') edge[labelbl={k}] (Y') % 
% lifting 
      (U') edge[serif cm-to,fore,shorten <=.3cm,shorten >=.3cm] (Y') %c` 
      (U') edge[serif cm-to,shorten <=.3cm,shorten >=.3cm] (Y') %c` 
      (U'') edge[serif cm-to,shorten <=.3cm,shorten >=.3cm] (Y'') %c` 
      (U) edge[serif cm-to,shorten <=.3cm,shorten >=.3cm] (Y) % 
    } % 
\end{mathpar}
Formally, consider any functor $p \colon \E \to \B$. A morphism $r
\colon E' \to E$ in $\E$ is \emph{cartesian} when, as on the right
above, for all $t \colon E'' \to E$ and $k \colon p(E'') \to p(E')$,
if $p(r) \rond k = p(t)$ then there exists a unique $s \colon E'' \to
E'$ such that $p(s) = k$ and $r \rond s = t$.
\begin{defi} 
  A functor $p \colon \E \to \B$ is a \emph{fibration} iff for all $E 
  \in \E$, any $h \colon B' \to p(E)$ has a cartesian lifting, i.e., a 
  cartesian antecedent by $p$. 
\end{defi} 

Proposition~\ref{prop:pseudodouble} is proved among other facts in
Section~\ref{sec:ccs}.  This was the starting point of the
notion of playground: which axioms should we demand of a pseudo double
category in order to enable the constructions of
\citetalias{2011arXiv1109.4356H}?  We follow the constructions in this
section, considering an arbitrary pseudo double category $\D$, on
which we impose axioms along the way. Objects and vertical
morphisms will respectively be called \emph{positions} and
\emph{plays}. The pseudo double category $\Dccs$ does satisfy the
axioms, albeit in a non-trivial way. This is stated and proved in
Section~\ref{sec:ccs}, but we use the result in advance in examples to
illustrate our constructions. 

Let us record the axioms imposed on $\D$ in the next sections to
obtain our bisimulation result (Theorem~\ref{thm:bisim}):
\begin{center}
  \begin{minipage}[t]{0.48\linewidth}
    \begin{itemize}
    \item \axref{fibration}, page \pageref{fibration},
    \item \axref{discreteness}---\axref{fibration:continued}, page
      \pageref{discreteness},
    \item \axref{ax:views}, page \pageref{ax:views},
    \item \axref{leftdecomposition}, page \pageref{leftdecomposition},
    \end{itemize}
  \end{minipage}
  \hfil
  \begin{minipage}[t]{0.48\linewidth}
    \begin{itemize}
    \item \axref{views:decomp}, page \pageref{views:decomp},
    \item \axref{finiteness}, page \pageref{finiteness},
    \item \axref{basic:full}, page \pageref{basic:full}.
    \end{itemize}
  \end{minipage}
\end{center}
\subsection{Behaviours}\label{subsec:beh}
The easiest construction of \citetalias{2011arXiv1109.4356H} to carry over to the abstract setting
of playgrounds is that of behaviours.  First, let us stress that, in
the case of $\Dccs$, $\DccsH$ is very different from the category of
plays called $\E$ recalled in Section~\ref{subsubsec:behaviours}.
Indeed, any morphism $\alpha \colon u \to u'$ in $\DccsH$ in
particular induces an embedding of the final position $\dom (u)$ of
$u$ into that of $u'$. In $\E$, instead, a morphism $U \to U'$ may
involve extending $U$ with more moves.
\begin{exa}
  The move of Example~\ref{ex:forkmove} embeds into the
  play of Example~\ref{ex:forkplay} in the sense of $\E$, but not in
  the sense of $\DccsH$.  Indeed, the passive player $y$ of
  Example~\ref{ex:forkmove} does belong to the final position, but its
  image in Example~\ref{ex:forkplay} does not.
\end{exa}
\begin{wrapfigure}[6]{r}{0pt}
  \begin{minipage}{0.3\linewidth}
    \centering
    \vspace*{-2em}
    \begin{equation}
      \diag{%
        |(Z)| Z \& |(Y')| Y' \\
        |(Y)| Y \& \\
        |(X)| X \& |(X')| X' %
      }{%
        (Z) edge[labelu={h}] (Y') %
        edge[pro,labell={w}] (Y) %
        (Y) edge[pro,labell={u}] (X) %
        (Y')  edge[pro,twor={u'}] (X') %
        (X) edge[labela={k}] (X') %
        (Y) edge[cell=.3,labelu={\alpha}] (r) %
      }\label{eq:alpha}
    \end{equation}
  \end{minipage}
\end{wrapfigure}
So our first step is to construct an analogue of $\E$ from any
playground $\D$. Intuitively, it should have as objects all plays, and
as morphisms $u \to u'$ all pairs $(w,\alpha)$ as on the right.
However, this definition is slightly wrong on morphisms, in that
$\alpha$ carries some information about how $w$ embeds into $u'$,
while we are only interested in how $u$ does.  Thus, we instead define
morphisms $u \to u'$ to be pairs $(w,\alpha)$ as in~\eqref{eq:alpha},
quotiented by the equivalence relation generated by pairs
$((w,\alpha),(w',\beta))$ such that there exists morphisms $i$ and
$\gamma$ satisfying $\alpha = \beta \circ (u \vrond \gamma)$, as in
   \begin{equation}
     \Diag(.15,.6){%
       \path (Y') -- node[pos=.22] (right){} (X') ; %
       \path (Z') edge[tworight={Z'Y}{}] (Y2) %
       ; %
       \path[->] (ZY) edge[cell={.8},fore,labelu={\gamma}] (Z'Y) ; %
       \path (Y) edge[cell={1}] node[pos=.8,above=6pt] {$\scriptstyle \alpha$} (r) ; %
       \path (Y2) edge[cell={.4},labelbr={\beta}] (r) ; %
       \path[->] (Z') edge[twobr={}] (Y')  %
       edge[pro,fore,tworight={Z'Y}{}] (Y2) %
       %(al) edge[cell={.2},labelbl={\id}] (br)  %
       ; %
       }{%
       \& \& \& \& |(Y')| Y' \\
       |(Z)| Z \\
       \& \& |(Z')| Z' \\ \\
        |(Y)| Y  \\
        \& \& |(Y2)| Y \& \\ 
       \& \& \& \& |(X')| X'. \\
       \& |(X)| X  
       }{%
         (Z) edge[pro,twoleft={ZY}{}] (Y) %
         edge[twoal={}] (Y') %
         (Y) edge[pro,twoleft={y}{u}] (X) %
         (Y2) edge[pro,tworight={y2}{u}] (X) %
         (X) edge (X') %
         (Y') edge[pro,twor={u'}] (X') %
         (Z) edge[labelar={i}] (Z') %
         (Y) edge[identity] (Y2) %
         (y) edge[cell={.1},labelu={\id}] (y2) %
       }
       \label{eq:quotient:E}
   \end{equation}

   In order to define composition in this category, we state the
   following axiom (cf.\ Proposition~\ref{prop:pseudodouble}).
\begin{ax}
  \begin{myaxioms}[series=myaxiomsseries]
  \item (Fibration) The vertical codomain functor $\cod : \DH \to \Dh$ is a fibration.\label{fibration}
  \end{myaxioms}
\end{ax}

Composition may now be defined by pullback (i.e., cartesian lifting in
the fibration $\cod \colon \DH \to \Dh$) and pasting:
   \begin{center}
     \Diag{%
       \path (Z') -- node[pos=.5] (mid) {} (X'); %
       \path (Y) edge[cell={.5},labelu={\alpha}] (r) ; %
       \path (mid) edge[cell={.5},labelu={\beta}] (right) ; %
       \pbkk{Z}{Z''}{Z'} %
     }{%
       |(Z'')| Z'' \& |(Z')| Z' \& |(V)| V \\
       |(Z)| Z \& |(Y')| Y' \&  \\
       |(Y)| Y \& \& \\
       |(X)| X  \& |(X')| X' \& |(U)| U.
       }{%
         (Z) edge[pro,labell={w}] (Y) %
         edge (Y') %
         (Y) edge[pro,labell={u}] (X) %
         (X) edge (X') %
         (Y') edge[pro,twor={u'}] (X') %
         (Z') edge[pro,labell={w'}] (Y') %
         edge (V) %
         (X') edge (U) %
         (V) edge[pro,tworight={right}{u''}] (U) %
         (Z'') edge[pro,labell={w''}] (Z) %
         edge (Z') %
       }
   \end{center}
   (We use `double pullback' marks to denote cartesian double cells.)
   Quotienting makes composition functional and associative, and
   furthermore it is compatible with the above equivalence. Identities are obvious.

   \begin{prop}
     This forms a category $\E$.
   \end{prop}
   \begin{exa}
     Consider the move $M'$ from Example~\ref{ex:forkmove}, and let us
     name its initial and final positions as in $M' \colon Y' \proto
     X'$.  Let us further call $U \colon Y'' \proto X'$ the play from
     Example~\ref{ex:forkplay}, obtained by composing $M'$ with a
     forking move by $y \in Y'[2]$.  In order to obtain a double cell
     $M' \to U$, we need to provide an extension of $M'$ with some
     move by $y$, and there are actually three ways of doing this.
     One is with a left forking move, another is with a right forking
     move, and the last is with a full forking move. In this example,
     the last possibility actually yields an identity double cell $U
     \to U$, and may be obtained using~\axref{fibration} in the
     \begin{minipage}[t]{\linewidth}
       \vspace*{-.7em}
       \begin{wrapfigure}{r}{0pt}
       \begin{minipage}[t]{.18\linewidth}
       \vspace*{-1.5em}
         \Diag(.6,1){%
           \pbkk{Y}{Z}{Z'} %
           \node[at=(c.center),anchor=base west] {$\scriptstyle
             \beta$} ; %
         }{%
           |(Z)| Z \&|(Z')| Z' \\
           |(Y)| Y \&|(Y')| Y' \\
           |(X)| X \&|(X')| X' %
         }{%
           (Z) edge[pro,twoleft={left}{w}] (Y) %
           edge (Z') (Y) edge[pro,twol={u}] (X) %
           edge (Y') %
           (X) edge (X') %
           (Z') edge[pro,tworight={right}{w'}] (Y') %
           (Y') edge[pro,twor={u'}] (X') %
           (l) edge[cell={.3},labelu={\alpha}] (r) %
           % (left) edge[cell={.5},labelu={\beta}] (right) %
         }
       \end{minipage}
       \end{wrapfigure}
       \noindent following general way.  Consider any double cell
       $\alpha \colon u \to u'$ in $\DH$, and play $w'$ such that $u'
       \vrond w'$ is well-defined. Then, letting $\beta \colon w \to
       w'$ be the cartesian lifting of $w'$ along $\dom(\alpha)$, we
       obtain a morphism $u \to u' \vrond w'$ in $\E$, as in on the
       right.  The universal property of $\beta$ here amounts to the
       fact that left and right forking moves both embed uniquely into
       full forking, which makes our three candidate morphisms $u \to
       u' \vrond w'$ equal in~$\E$.
     \end{minipage}
   \end{exa}

   Recalling notation from Section~\ref{subsec:prelim:cats}, consider
   now the pullback category $\E(X)$, where $X$ is any position.
   Following Definition~\ref{def:behccs}, we state: %\vspace*{-2ex} {}
\begin{defi}\label{def:beh}
  The category $\Beh{X}$ of \emph{behaviours} on $X$ is
  $\FPsh{\E(X)}$, i.e., the category of presheaves of finite sets on
  $\E (X)$.
\end{defi}
This construction has a bit of structure. Indeed, the
map $X \mapsto \E (X)$ extends to a pseudo functor $\E (-) \colon \Dv \to
\Cat$ by vertical post-composition.  Post-composing the opposite of
this pseudo functor by $\FPsh{(-)} \colon \op\Cat \to \Cat$, we obtain a
pseudo functor $\Beh{-} \colon \op\Dv \to \Cat$, satisfying $\Beh{u} (B) (u')
= B (u \vrond u')$.

\subsection{More axioms}
We now turn to generalising further constructions
of \citetalias{2011arXiv1109.4356H} to the general setting of
playgrounds.  We mentioned in Section~\ref{sec:HP} that strategies on
a position $X$ should be defined as presheaves on the category of
views on $X$.  We will further want to generalise the decomposition
theorems for strategies of \citetalias{2011arXiv1109.4356H}, which
crucially rely on a property of views stated (in
Section~\ref{subsec:views} below) as Proposition~\ref{prop:decompV}.

In order for this to work, we need to state more axioms for $\D$.  In
particular, the axioms equip $\D$ with a notion of \emph{player} for a
position $X$.  Each position has a set of players, each player having
a certain `type'.  Furthermore, in Section~\ref{subsec:views}, $\D$ is
equipped with a notion of view; and views have a type, too.
Proposition~\ref{prop:decompV}, e.g., states that views on a position
$X$ form a coproduct, over all players $x$ in $X$, of views over the
type of $x$.

We first state a series of simple axioms, and then, building on these,
two more complicated axioms.
\newcounter{axiomcounter}
\begin{ax}\label{ax:indivmoves}
  $\D$ is equipped with
  \begin{itemize}
  \item a full subcategory $\DI \into \Dh$ of objects called
    \emph{individuals},
  \item a replete class $\DM$ of vertical morphisms called \emph{moves},
    with replete subclasses $\DB$ and $\DF$, respectively called
    \emph{basic} and \emph{full} moves,
  \item a map $\length{-} \colon \ob(\DH) \to \Nat$ called the \emph{length},
  \end{itemize}
  satisfying the following conditions:
  \begin{axioms}
  \item $\DI$ is discrete. Basic moves have no non-trivial
    automorphisms in $\DH$.  Vertical identities on individuals have
    no non-trivial endomorphisms.
    \label{discreteness}
  \item (Individuality) Basic moves have individuals as both domain
    and codomain. \label{individuality}
  \item
    \begin{minipage}[t]{.97\linewidth}
      \vspace*{-.82em}
      \begin{wrapfigure}{r}{0pt}
        \begin{minipage}[t]{.18\linewidth}
          \vspace*{-1.5em} \diag{%
            |(X0)| X \&|(X)| X \\
            |(Xi)| X \&|(Y)| Y %
          }{%
            (X0) edge[identity,pro,twol={}] (Xi) %
            edge[identity] (X) %
            (X) edge[pro,twor={u}] (Y) %
            (Xi) edge[labelb={\bar{u}}] (Y) %
            (l) edge[cell=0.3,labelb={\alpha^u}] (r) %
          }
        \end{minipage}
      \end{wrapfigure}
      (Atomicity) For any cell $\alpha \colon v \to u$, if $\length{u}
      = 0$ then also $\length{v} = 0$.  Up to a special isomorphism in
      $\DH$, all plays $u$ of length $n > 0$ admit decompositions into
      $n$ moves.  For any $u \colon X \proto Y$ of length 0, there is
      an isomorphism $\idv_X \to u$ as on the right in $\DH$.
    \end{minipage}
\label{atomicity}
\item (Fibration, continued) Restrictions of moves (resp.\ full moves)
  to individuals either are moves (resp.\ full moves), or have length
  0. \label{fibration:continued}
  \end{axioms}
\end{ax}
Replete means stable under isomorphism (here in $\DH$).  In
\axref{fibration:continued}, \emph{restriction} is w.r.t.\ the
fibration $\cod \colon \DH \to \Dh$, as explained below
Proposition~\ref{prop:pseudodouble}. % \axref{fibration:continued} thus 
% amounts to requiring that plays of length $0$ or $1$ 
% form a subfibration of $\cod$, and similarly with full moves. 

\begin{defi}
  A \emph{player} in a position (i.e., object) $X$, is a pair $(d,x)$,
  where $d \in \DI$ and $x : d \to X$. Let $\Pl (X) = \sum_{d \in \DI}
  \Dh(d,X)$ be the set of players of $X$.
\end{defi}
\begin{exa}
  In $\Dccs$, individuals are representable positions $[n]$, which
  consist for some $n$ of a single $n$-ary player, connected to $n$
  distinct channels. Importantly, for each isomorphism class of such
  positions we pick one representative: this makes $\DI$ discrete by
  Yoneda. Furthermore, basic moves are basic seeds.
\end{exa}

Here is a further, crucial axiom.
\begin{defi}\label{def:DB0}
  Let $\DB_{0}$ be the full subcategory of $\DH$ having as objects
  basic moves and morphisms of length 0 between individuals.
\end{defi}

\begin{ax}
  \begin{axioms}
  \item (Views) For any move $M \colon Y \proto X$ in $\Dv$, the domain
    functor $\dom \colon \DB_{0} / M \to \DI / Y$ is an equivalence of categories. \label{ax:views}
  \end{axioms}
\end{ax}
In elementary terms, for any $y \colon d \to Y$ in $\Dh$ with $d\in
\DI$, there exists a cell
\begin{center}
  \diag{%
    |(d)| d \& |(Y)| Y \\
    |(dMy)| d^{y,M} \& |(X)| X, %
  }{%
    (d) edge[labelu={y}] (Y) %
    edge[pro,dashed,twol={v^{y,M}}] (dMy) %
    (Y) edge[pro,twor={M}] (X) %
    (dMy) edge[dashed,labeld={y^M}] (X) %
    (l) edge[cell=.3,dashed,labelu={\scriptstyle \alpha^{y,M}}] (r) %
  }
\end{center}
with $v^{y,M} \in \DB_0$, which is unique up to canonical isomorphism of
such. An isomorphism between two such tuples, say $(d',v',y',\alpha')$ and $(d'',v'',y'',\alpha'')$
is a diagram
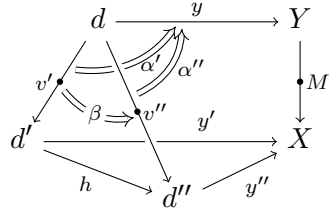
\begin{center}
  \Diag (.3,.5) {%
    \twocellr{dMy}{d}{Y}{\alpha'} %
    \twocellr{d''}{d}{Y}{\alpha''} %
    \twocellro[.5]{dMy}{d}{d''}{\beta} %
     \path[->,draw] %
    (dMy) edge node[pos=.7,above] {$\scriptstyle y'$} (X) %
    edge[labelbl={h}] (d'')
    (d'') edge[labelbr={y''}] (X) %
    (d) edge[pro,fore,labelr={v''}] (d'') %
    ; %
  }{%
    \& |(d)| d \& \& \& |(Y)| Y \\
    \& {} \\
    |(dMy)| d' \& \& \& \& |(X)| X \\
    \& \& |(d'')| d'' \& %
  }{%
    (d) edge[labelu={y}] (Y) %
    edge[pro,twol={v'}] (dMy) %
    (Y) edge[pro,twor={M}] (X) %
%    (l) edge[cell=.3,labelu={\scriptstyle \alpha'}] (r) %
  }    
\end{center}
such that $\alpha'' \rond \beta = \alpha'$ (where necessarily
$d'=d''$, $h = id$, and $y' = y''$).

% Consider the pseudo-double category $\Dseq$ with $\Dseqh = \Dh$, but 
% with as vertical morphisms $X \to Y$ finite, possibly empty sequences 
% of moves $X = X_0 \xto{M_1} X_1 \ldots X_{n-1} \xto{M_n} X_n = Y$. 
% Let, for any such sequence $M$, $\vcompose{M} = M_n \vrond \ldots 
% \vrond M_1$ denote its composition in $\Dv$.  Double cells $M \to M'$ 
% in $\DseqH$ are double cells $\vcompose{M} \to \vcompose{{M'}}$ in 
% $\DH$. This obviously forms a pseudo-double category, with a 
% pseudo-double functor $\vcomposefun \colon \Dseq \to \D$. 

% \begin{ax} 
%   \begin{axioms} 
%   \item (Atomicity) The functor $\vcomposeHfun \colon \DseqH \to \DH$ 
%     induced by $\vcomposefun$ is surjective on objects, up to special 
%     isomorphism. \label{atomicity} 
%   \end{axioms} 
% \end{ax} 

\begin{exa}
  This axiom is obviously satisfied by $\Dccs$.
\end{exa}

We then have two decomposition axioms.
% Consider first the category $\Dtwo$ obtained by pullback 
% \begin{center} 
%   \Diag{%\ 
%     \pbk{m-2-1}{m-1-1}{m-1-2} %\ 
%   }{%\ 
%     \Dtwo \& \DH \\ 
%     \DH \& \Dh. %\ 
%   }{%\ 
% (m-1-1) edge[labelu={q}] (m-1-2) %\ 
% edge[labell={p}] (m-2-1) %\ 
% (m-2-1) edge[labeld={\dom}] (m-2-2) %\ 
% (m-1-2) edge[labelr={\cod}] (m-2-2) %\ 
%   } 
% \end{center} 
% For any composable $X \xproto{u_1} Y \xproto{u_2} Z$ in $\Dv$, we may consider  
% the slice category $\Dtwo / (u_2,u_1)$, which admits a functor to 
% $\DH / (u_2 \vrond u_1)$ given by vertical composition. 
\begin{ax}
  \begin{axioms}
  \item (Left decomposition) 
    % For any such composable $u_1$ and $u_2$, 
    % the functor $\Dtwo / (u_2,u_1) \to \DH / (u_2 \vrond u_1)$ is an 
    % equivalence. 
    Any double cell
    \begin{center}
      \Diag{%
        \twocellbr{B}{A}{X}{\alpha} %
      }{%
        |(A)| A \& |(X)| X \\
        \& |(Y)| Y \\
        |(B)| B \& |(Z)| Z %
      }{%
        (A) edge[labelu={h}] (X) %
        edge[pro,labell={u}] (B) %
        (X) edge[pro,labelr={w_1}] (Y) %
        (Y) edge[pro,labelr={w_2}] (Z) %
        (B) edge[labeld={k}] (Z) %
      }
      \hfil decomposes as \hfil
      \Diag(1,2){%
        \twocellbr[.3]{C}{A}{X}{\alpha_1} %
        \twocellbr[.3]{B}{C}{Y}{\alpha_2} %
        \twocell[.3]{L}{A}{C}{}{celllr={0}{0},bend
          right,fore,labelbl={\scriptscriptstyle \alpha_3}} %
      }{%
        \& |(A)| A \& |(X)| X \\
        |(L)| \& |(C)| C \& |(Y)| Y \\
        \& |(B)| B \& |(Z)| Z %
      }{%
        (A) edge[labelu={h}] (X) %
        edge[pro] node[pos=.5,anchor=north west] {$\scriptscriptstyle {u_1}$} (C) %
        edge[pro,bend right=70,labell={u}] (B) %
        (C) edge[pro] node[pos=.5,anchor=north west] {$\scriptscriptstyle {u_2}$} (B) %    
        edge[labelu={l}] (Y) %
        (X) edge[pro,labelr={w_1}] (Y) %
        (Y) edge[pro,labelr={w_2}] (Z) %
        (B) edge[labeld={k}] (Z) %
      }
    \end{center}
    with $\alpha_3$ an isomorphism, in an essentially unique way.
    \label{leftdecomposition}
  \end{axioms}
\end{ax}

Here is our second decomposition axiom. 

\begin{ax}
  \begin{axioms}
  \item (Right decomposition) Any double cell as in the center below,
    where $b$ is a basic move and $M$ is a move, decomposes in exactly
    one of the forms on the left and right:
  
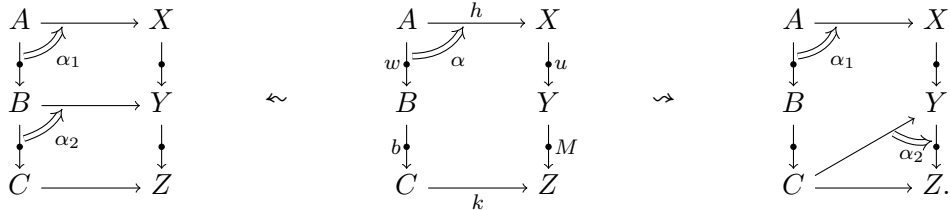
\begin{mathpar}
    \begin{minipage}[t]{0.18\linewidth}
      \centering \Diag {%
        \twocell[.4][.3]{B}{A}{X}{}{celllr={0.0}{0.0},bend
          right=30,labelbr={\alpha_1}} %
        \twocell[.4][.3]{C}{B}{Y}{}{celllr={0.0}{0.0},bend
          right=20,labelbr={\alpha_2}} %
      }{%
    |(A)| A \& |(X)| X \\
     |(B)| B \& |(Y)| Y \\
    |(C)| C \& |(Z)| Z %
      }{%
        (A) edge[pro] (B) %
        edge (X) %
        (B) edge (Y) %
        edge[pro] (C) %
        (C) edge (Z) %
        (X) edge[pro] (Y) %
        (Y) edge[pro] %node[pos=.5,inner sep=0pt,anchor=south east]
        (Z) %
      } 
    \end{minipage}
    \and \leftlsquigarrow \and
    \begin{minipage}[t]{0.18\linewidth}
      \centering      \Diag{%
    \twocellbr{B}{A}{X}{\alpha} %
  }{%
    |(A)| A \& |(X)| X \\
     |(B)| B \& |(Y)| Y \\
    |(C)| C \& |(Z)| Z %
  }{%
    (A) edge[labelu={h}] (X) %
    edge[pro,labell={w}] (B) %
    (B) edge[pro,labell={b}] (C) %
    (X) edge[pro,labelr={u}] (Y) %
    (Y) edge[pro,labelr={M}] (Z) %
    (C) edge[labeld={k}] (Z) %
  }
    \end{minipage}
\and \rightrsquigarrow \and
     \begin{minipage}[t]{0.18\linewidth}
      \centering \Diag {%
        \twocell[.4][.3]{B}{A}{X}{}{celllr={0.0}{0.0},bend
          right=30,labelbr={\alpha_1}} %
        \twocell[.3][.4]{C}{Y}{Z}{}{celllr={0.0}{0.0},bend
          right=20,labeld={\alpha_2}} %
      }{%
    |(A)| A \& |(X)| X \\
     |(B)| B \& |(Y)| Y \\
    |(C)| C \& |(Z)| Z. %
      }{%
        (A) edge[pro] (B) %
        edge (X) %
        (B) % edge (Y') %
        edge[pro] (C) %
        (C) edge (Z) %
        edge (Y) %
        (X) edge[pro] (Y) %
        (Y) % edge node[pos=.5,inner sep=0pt,anchor=south east]
 %       {$\scriptstyle u_1$} (Y') %
        edge[pro] (Z) %
%        (Y') edge (Z) %
      } 
    \end{minipage}
  \end{mathpar}
    \label{views:decomp}
  \end{axioms}
\end{ax}

\begin{rem}
  This axiom takes pseudoness rather sloppily. Indeed, the domain of
  the right-hand composite is not really $b \vrond w$, but rather
  $\idv_C \vrond (b \vrond w)$. So we actually mean $\alpha =
  (\alpha_2 \vrond \alpha_1) \rond \inv{\lambda_{b \vrond w}}$, where
  $\lambda$ cancels identities on the left.
\end{rem}

\begin{exa} \hfill \  \linebreak
\begin{minipage}{\textwidth}
\begin{wrapfigure}[6]{r}{0pt}
  \begin{minipage}[c]{0.18\linewidth}
    \vspace*{-1.5em}
    \Diag{%
      \twocellbr{B}{A}{X}{\alpha} %
    }{%
      |(A)| X \& |(X)| X \\
      |(B)| X \&  |(Y)| X \\
      |(C)| X \& |(Z)| X %
    }{%
      (A) edge[identity] (X) %
      edge[pro,labell={i_y}] (B) %
      (X) edge[identity] (Y) %
      (Y) edge[pro,labelr={S}] (Z) %
      (B) edge[pro,labell={o_x}] (C) %
      (C) edge[identity] (Z) %
    }
  \end{minipage}
\end{wrapfigure}
That this axiom is satisfied by $\Dccs$ is not obvious and is proved
in Section~\ref{sec:ccs}. However, let us disprove the more general
version where $b$ is not required to be basic. Let $X$ consist of two
players $x$ and $y$ sharing a channel $a$. Let $i_y \colon X \proto X$
be the play where $y$ inputs on $a$, $o_x \colon X \proto X$ be the
play where $x$ outputs on $a$, and let $S \colon X \proto X$ be the
play where both players synchronise on $a$.  We obtain a double cell
as on the right, which does not decompose as
in~\axref{views:decomp}. The problem here is that, on the left-hand
side, the upper input by $y$ has to be mapped to $S$, which prevents any
suitable decomposition.
\end{minipage}
\end{exa}

We now define and study views.

\subsection{Views}\label{subsec:views}
\begin{defi}
  A \emph{view} in $\D$ is a play which is specially isomorphic in
  $\DH$ to a possibly empty (vertical) composite of basic moves. I.e.,
  if
$d_n \xproto{b_{n}} d_{n-1} \ldots d_1 \xproto{b_1} d_0$
are all basic moves, then the composite is a view. Let $\DVi$ be the
full subcategory of $\DH$ consisting of views.
\end{defi}
The definition includes the `identity' view $\idv_{d}$. In $\Dccs$,
this of course coincides with views as defined in
\citetalias{2011arXiv1109.4356H}.

Here is an important consequence of our axioms. It is a bit
complicated to state, but very useful in the (more intelligible)
developments on views below.

\begin{lem}\label{lem:big}
    For all plays $w \colon Y \proto d_n$ and $u \colon X_p \proto X_0$, 
    views $v \colon d_n \proto d_0$, and double cells $\alpha \colon v \vrond w \to u$,
    for all special isomorphisms 
    $\gamma \colon (b_1 \vrond (\ldots (b_{n-1} \vrond b_n)\ldots)) \to v$
    and 
    $\gamma' \colon u \to (M_1 \vrond (\ldots (M_{p-1} \vrond M_p)\ldots))$
    decomposing $v$ and $u$ into moves, there exists a unique, strictly monotone map
    $f \colon n \cup \ens{0} \to p \cup \ens{0}$ with $f (0) = 0$
    and double cells $\beta \colon w \to (M_{f(n)+1} \vrond (\ldots \vrond M_p))$ and  
    $\alpha_k \colon \bar{b}_k \to M_k$ for $1 \leq k \leq f(n)$, where
    $$\bar{b}_k = \left \{
      \begin{array}{ll}
        b_i & \mbox{(if $k \in \im(f)$, with $f(i) = k$)} \\
        \idv_{\cod(b_{\min \ens{i \in n \aalt f(i) > k}})} & \mbox{(otherwise),}
      \end{array}
\right .$$
such that $\alpha_1 \vrond (\ldots \vrond (\alpha_{f(n)} \vrond \beta))
= \gamma' \rond \alpha \rond (\gamma \vrond w)$,
as in 
    \begin{mathpar}
  \Diag(.6,4){%
    \path (vdotsi) -- node (hurm) {$\vdots$} (d'n1) ; %
    \deuxcellule{1.5}{1}{dn1}{dn}{l}{}{cell=.3,labelb={\gamma}} %
    \deuxcellule{1.1}{.9}{l}{dn}{r}{}{celllr={0}{.4},labela={\alpha\ \ \ \ }} %
    \path[->] (r) edge[cell={.2},labela={\gamma'}] (hurm) ; %
  }{%
   |(A)| Y \&|(X)| X_p \\
   |(dn)| d_n \&|(d'n)| X_{p-1} \\
   |(dn1)| d_{n-1} \&|(d'n1)|  \\
   |(vdots)| \vdots \&|(vdotsi)| \\
   |(d1)| d_1 \&|(d'1)| X_1 \\
   |(d0)| d_0 \&|(d'0)| X_0 %
  }{%
    (A) edge[labell={w}] (dn) %
    edge (X) %
    (dn) edge[labell={b_n}] (dn1) %
    edge[bend left=60] node[pos=.3,right] (l) {$\scriptstyle v$} (d0) %
    (X) edge[labelr={M_{p}}] (d'n) %
    edge[bend right=50,twor={}] node[pos=.3,left] {$\scriptstyle u$} (d'0) %
    (d0) edge (d'0) %
    (d1) edge[labell={b_1}] (d0) %
    (d'1) edge[labelr={M_1}] (d'0) %
%    (l) edge[cell={.4},labela={\alpha}] (r) %
%    (vdots) edge[cell={.2},labela={\inv\gamma}] (l) %
  }
\and  = \and    
    \diagramme[diag={0.6}{1.5}]{}{%
        \twocellbr{B0}{A0}{Bf1minus1}{} %
        \twocellbr{A0}{A1}{Bf1}{\alpha_{f(1)}} %
        \twocellr{Bf1}{A1}{Bf2minus1}{} %
        \twocellbr{Amminus1}{Am}{Bfm}{\alpha_{f(n)}} %
        \twocellbr{Am}{A}{Bn}{\beta} %
      }{%
       |(A)| Y \&|(Bn)| X_p \\
       |(Am)| d_n \&|(Bfm)| X_{f (n)} \\
       |(Amminus1)| d_{n-1} \&|(Bfmminus1)| X_{f(n)-1} \\
       \&|(Bf2minus1)| X_{f(2)-1} \\
       |(A1)| d_1 \&|(Bf1)| X_{f (1)} \\
       |(A0)| d = d_0 \&|(Bf1minus1)| X_{f (1) - 1} \\
         \&|(B0)| X_{0}. %        
      }{%
        (A) edge (Bn) %
        edge[pro,labell={w}] (Am) %
        (Am) edge (Bfm) %
        edge[pro,labell={b_n}] (Amminus1) %
        (Amminus1) edge (Bfmminus1) %
        (A1) edge (Bf2minus1) %
        edge (Bf1) %
        edge[pro,labell={b_1}] (A0) %
        (A0) edge (Bf1minus1) %
        edge (B0) %
        (Bn) edge[pro,labelr={M_{> f (n)}}] (Bfm) %
        (Bfm) edge[pro,labelr={M_{f (n)}}] (Bfmminus1) %
        (Bf2minus1) edge[pro,labelr={M_{] f (1), f (2)[}}] (Bf1) %
        (Bf1) edge[pro,labelr={M_{f (1)}}] (Bf1minus1) %
        (Bf1minus1) edge[pro,labelr={M_{< f (1)}}] (B0) %
        (Bfmminus1) edge[dotted] (Bf2minus1) %
        (Amminus1) edge[dotted] (A1) %
      }
  \end{mathpar}
%Figure~\ref{fig:big}.
\end{lem}
In the case $p=0$, also $n = \length{w} = 0$, and the decomposition of $u$
  should be understood as $M_1 \vrond \ldots \vrond M_{f(n)}$ being an
  \emph{identity}, with $M_{f(n)+1} \vrond \ldots \vrond M_p$ being
  $u$. 

%\begin{figure}[ht]
%   \caption{Picture for Lemma~\ref{lem:big}}
%   \label{fig:big}
% \end{figure}
\begin{rem}
  Only $f$ is claimed to be unique here.  Furthermore, as
  in~\axref{views:decomp}, we are a bit sloppy regarding pseudoness.
  Also, in the following, we consider only the underlying map $f
  \colon n \to p$, implicitly extended with $f(0) = 0$.  Finally, for
  all $\alpha \colon v \vrond w \to u$, there exist $\gamma$ and
  $\gamma'$ as in the lemma. This is obvious when $n$ and $p \neq 0$;
  we just explained it for the case $p=0$; and when $n=0$ it follows
  from Lemma~\ref{lem:hv} below.
\end{rem}
\begin{proof}
  We proceed by lexicographic induction on the pair $(n,p)$.  

  If $n = 0$ then our map $f \colon n \to p$ is the unique map $0 \to
  p$, $f(n) = 0$, and we take $\beta = \gamma' \rond \alpha \rond
  \gamma$.  Otherwise, we apply~\axref{views:decomp} with $b = b_1$,
  $w = (b_{2} \vrond \ldots \vrond b_{n-1} \vrond w)$, $M = M_1$ and
  $u = (M_2 \vrond \ldots \vrond M_{p-1})$.
    \begin{itemize}
    \item If we are in the left-hand case, $\alpha$ decomposes as
      $\alpha_1 \vrond \alpha_2$, with $\alpha_1 \colon b_1 \to M_1$
      and $\alpha_2 \colon (b_2 \vrond \ldots \vrond b_{n} \vrond w)
      \to (M_2 \vrond \ldots \vrond M_{p})$.  By induction hypothesis,
      we obtain a map $f' \colon n-1 \to p-1$ and a
      corresponding decomposition of $\alpha_2$.  We then let
      $f \colon n \to p$ map $1$ to $1$, and $k + 1$ to
      $f'(k) + 1$ for any $k \in (n-1)$.
    \item If we are in the right-hand case, we obtain a map
      $f' \colon n \to p-1$, and return the map $k
      \mapsto f'(k) + 1$.
    \end{itemize}
    This shows existence of the desired decomposition. For uniqueness,
    consider any map $g \colon n \to p$ and corresponding
    decomposition.  Axiom~\axref{leftdecomposition} entails that at
    each stage, $\inv{f}\ens{1,\ldots,k}$ and
    $\inv{g}\ens{1,\ldots,k}$ have the same cardinality. Indeed,
    otherwise, we would find isomorphic decompositions of $b_1 \vrond
    \ldots \vrond (b_n \vrond w)$ with incompatible lengths.  Thus, $f
    = g$.
\end{proof}

We continue with a few easy results.  Recall the family of isomorphisms
$\alpha^u$ from Axiom~\axref{atomicity}, indexed by vertical morphisms
of length $0$.  Furthermore, let us denote by $\rho_u \colon u \vrond
\idv_X \to u$ and $\lambda_u \colon \idv_Y \vrond u \to u$ the
coherence isomorphisms from $\Dv$ for cancelling vertical identities.
\begin{lem}\label{lem:alphau}
  For any $u \colon X \proto Y$ of length $0$, there is an isomorphism
  \begin{mathpar}
          \diag{%
         |(X)| X \&|(Y)| Y \\
         |(Y0)| Y \&|(Yi)| Y %
       }{%
         (Y) edge[identity,pro,twor={}] (Yi) %
         (X) edge[labela={\bar{u}}] (Y) %
         edge[pro,twol={u\ \ }] (Y0) %
         (Y0) edge[identity] (Yi) %
         (l) edge[cell=0.3,labela={\alpha_u}] (r) %
       }         
  \end{mathpar}
  in $\DH$, such that $\alpha_u \vrond \alpha^u = \inv{\lambda_u}
  \rond \rho_u$ and $\alpha_u \rond \alpha^u = \idv_{\bar{u}}$.
\end{lem}
\begin{proof}
  Pose $\alpha_u = \idv_{\bar{u}} \rond \inv{(\alpha^u)}$.
\end{proof}
\begin{lem}\label{lem:hv}
  If $b \colon d \proto d'$ has length $0$, then $d = d'$, $\bar{b} =
  \id_d$, and $\alpha_b$ and $\alpha^b$ are horizontal inverses.
\end{lem}
\begin{proof}
  By~\axref{discreteness}.
\end{proof}

\begin{lem}\label{lem:B0isos}
  $\DB_0$ (Definition~\ref{def:DB0}) is a groupoid.
\end{lem}
\begin{proof}
  This means that any $\alpha \colon b \to b'$ in $\DB_0$ is an
  isomorphism.  Let $b \colon d_1 \proto d_2$ and $b' \colon d'_1
  \proto d'_2$. Existence of $\alpha$ entails $d_1 = d'_1$ and $d_2 =
  d'_2$, by~\axref{discreteness}.

  If $b' \in \DB$, then $\alpha$ and $\id_{b'}$ are both
  mapped by $\dom \colon \DB_0 / b' \to \DI / d'_1$ to $\dom(\alpha) =
  \id_{d'_1}$. By~\axref{ax:views}, there is thus a unique
  isomorphism $\gamma \colon b \to b'$ in $\DH$ such that $\id_{b'}
  \rond \gamma = \alpha$, i.e., $\gamma = \alpha$. This shows that
  $\alpha$ is an iso.
  
  If $b'$ has length $0$, then by~\axref{atomicity} we furthermore
  have $\length{b} = 0$ and $d_1 = d_2 = d'_1 = d'_2$. Moreover, the
  composite $\alpha_{b'} \rond \alpha \rond \alpha^b$ (with
  $\alpha_{b'}$ and $\alpha^{b}$ as in Lemma~\ref{lem:alphau}
  and~\axref{atomicity}) is an endomorphism of $\idv_{d_1}$, hence
  $\id_{\idv_{d_1}}$ by~\axref{discreteness}.  It is thus an
  isomorphism, hence so is $\alpha^{b'} \rond \alpha_{b'} \rond \alpha
  \rond \alpha^b \rond \alpha_b$, which is equal to $\alpha$
  by two applications of Lemma~\ref{lem:hv}.
\end{proof}

\begin{lem}\label{lem:auto}
  In any category $\C$, for any object $c$ isomorphic to an object $d$ such
  that $d$ has no non-trivial endomorphisms, $c$ does not have any
  non-trivial endomorphisms either.
\end{lem}
\begin{proof}
  By the Yoneda lemma, we have $\C(c,c) \iso \C(c,d) \iso \C(d,d) \iso
  1$.
\end{proof}

\begin{lem}
   Any groupoid $\C$ whose objects have no non-trivial endomorphisms
   is an equivalence relation.
\end{lem}
\begin{proof}
  For any objects $c$ and $d$, we have that if $\C(c,d)$ is non-empty
  then $c$ and $d$ are isomorphic, so by Yoneda $\C(c,d) \iso \C(c,c)
  \iso 1$.
\end{proof}

\begin{cor}\label{cor:B0eqrel}
$\DB_0$ is an equivalence relation.  
\end{cor}
This adds to Lemma~\ref{lem:B0isos} that there is at most one morphism
between any two objects.
  \begin{proof}
    By Lemma~\ref{lem:auto} and~\axref{atomicity}, its objects have no
    non-trivial automorphisms, which in a groupoid is the same as
    having no non-trivial endomorphisms. By the last result, $\DB_0$
    is an equivalence relation.
  \end{proof}

This leads to a better understanding of $\DVi$.
\begin{lem}\label{lem:preorder}
  Consider any morphism of views $\alpha \colon v \to v'$, with
  isomorphisms $\gamma \colon (b_1 \vrond \ldots \vrond b_n) \to v$
  and $\gamma' \colon v' \to (b'_1 \vrond \ldots \vrond b'_{n'})$, for
  basic moves $b_i \colon d_i \proto d_{i-1}$ and $b'_j \colon d'_j
  \proto d'_{j-1}$ for all $i \in n$ and $j \in n'$. We have $n = n'$,
  $d_{i-1} = d'_{i-1}$ for all $i \in n+1$, and there exist unique
  isomorphisms $\alpha_i \colon b_i \to b'_i$ such that $\gamma' \rond
  \alpha \rond \gamma = (\alpha_n \vrond \ldots \vrond
  \alpha_1)$, as in
\begin{mathpar}
  \diag(.6,4){%
   |(dn)| d_n \&|(d'n)| d'_n \\
   |(dn1)| d_{n-1} \&|(d'n1)| d'_{n-1} \\
   |(vdots)| \vdots \&|(vdotsi)| \vdots \\
   |(d1)| d_1 \&|(d'1)| d'_1 \\
   |(d0)| d_0 \&|(d'0)| d'_0 %
  }{%
    (dn) edge[labell={b_n}] (dn1) %
    edge (d'n) %
    edge[bend left=50,twol={}] node[pos=.3,right] {$\scriptstyle v$} (d0) %
    (d'n) edge[labelr={b'_n}] (d'n1) %
    edge[bend right=50,twor={}] node[pos=.3,left] {$\scriptstyle v'$} (d'0) %
    (d0) edge (d'0) %
    (d1) edge[labell={b_1}] (d0) %
    (d'1) edge[labelr={b'_1}] (d'0) %
    (l) edge[cell={.4},labela={\alpha}] (r) %
    (vdots) edge[cell={.2},labela={\gamma}] (l) %
    (r) edge[cell={.2},labela={\gamma'}] (vdotsi) %
  }
\and  = \and
  \diag(.6,2){%
   |(dn)| d_n \&|(d'n)| d'_n \\
   |(dn1)| d_{n-1} \&|(d'n1)| d'_{n-1} \\
   |(vdots)| \vdots \&|(vdotsi)| \vdots \\
   |(d1)| d_1 \&|(d'1)| d'_1 \\
   |(d0)| d_0 \&|(d'0)| d'_0. %
  }{%
    (dn) edge[twol={b_n}] (dn1) %
    edge[identity] (d'n) %
    (d'n) edge[twor={b'_n}] (d'n1) %
    (d0) edge[identity] (d'0) %
    (d1) edge[twoleft={b1}{b_1}] (d0) %
    (d'1) edge[tworight={b'1}{b'_1}] (d'0) %
    (dn1) edge[identity] (d'n1) %
    (l) edge[cell={1},labela={\alpha_n}] (r) %
    (b1) edge[cell={1},labela={\alpha_1}] (b'1) %
    (d1) edge[identity] (d'1) %
  }
\end{mathpar}
\end{lem}
\begin{proof}
Applying Lemma~\ref{lem:big} with $w = \idv_{d_n}$ yields $f \colon n \to n'$ which by
Corollary~\ref{cor:B0eqrel} and~\axref{atomicity} has to be a bijection. This yields
the desired $\alpha_i$'s, which are unique by Corollary~\ref{cor:B0eqrel} again.
\end{proof}

This entails:
\begin{cor}\label{cor:Vpreordergroupoid}
  $\DVi$ is an equivalence relation, compatible with length.
\end{cor}

Here is an analogue of~\axref{ax:views} for general plays and views
instead of just moves and basic moves.
\begin{prop}\label{prop:views}
For any $y \colon d \to Y$ in $\Dh$ with $d\in \DI$,
    and any $u \colon Y \proto X$ in $\Dv$, there exists a cell
\begin{center}
  \diag{%
    |(d)| d \& |(Y)| Y \\
    |(dMy)| d^{y,u} \& |(X)| X, %
  }{%
    (d) edge[labelu={y}] (Y) %
    edge[pro,dashed,twol={v^{y,u}}] (dMy) %
    (Y) edge[pro,twor={u}] (X) %
    (dMy) edge[dashed,labeld={y^u}] (X) %
    (l) edge[cell=.3,dashed,labelu={\scriptstyle \alpha^{y,u}}] (r) %
  }
\end{center}
with $v^{y,u}$ a view, which is unique up to canonical isomorphism of
such.   
\end{prop}

\begin{proof}
  We find $v^{y,u}$ by repeated application of~\axref{ax:views}.  For
  essential uniqueness, by repeated application of~\axref{ax:views}, 
  we find an isomorphism between any two such views, which 
  by Corollary~\ref{cor:Vpreordergroupoid} is unique.
\end{proof}

We continue with an analogue of~\axref{views:decomp}:
\begin{prop}\label{prop:views:decomp}
  Any double cell
  \begin{center}
      \Diag{%
    \twocellbr{B}{A}{X}{\alpha} %
  }{%
    |(A)| A \& |(X)| X \\
     |(B)| B \& |(Y)| Y \\
    |(C)| C \& |(Z)| Z, %
  }{%
    (A) edge[labelu={h}] (X) %
    edge[labell={w}] (B) %
    (B) edge[labell={v}] (C) %
    (X) edge[labelr={u}] (Y) %
    (Y) edge[labelr={u'}] (Z) %
    (C) edge[labeld={k}] (Z) %
  }
  \end{center}
where $v$ is a view, decomposes in exactly one of the following forms:
  
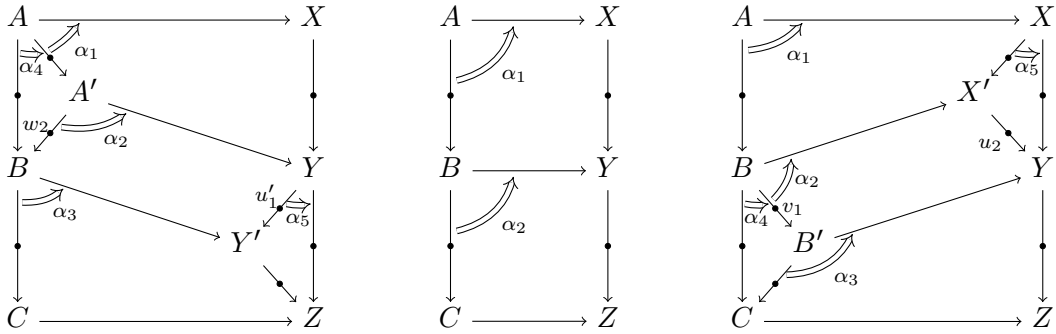
\begin{center}
    \begin{minipage}[t]{0.35\linewidth}
      \centering \Diag (.5,.25) {%
        \twocell[.4][.2]{A'}{A}{X}{}{celllr={0.0}{0.0},bend
          right=10,labelbr={\alpha_1}} %
        \twocell[.4][.2]{B}{A'}{Y}{}{celllr={0.0}{0.0},bend
          right=20,labelbr={\alpha_2}} %
        \twocell[.2][.2]{C}{B}{Y'}{}{celllr={0.0}{0.0},bend
          right=20,labelbr={\alpha_3}} %
        \twocell[.2][.4]{B}{A}{A'}{}{celllr={0.0}{0.0},bend
          right=10,labeld={\alpha_4}} %
        \twocell[.4][.2]{Y'}{Y}{Z}{}{celllr={0.0}{0.0},bend
          right=10,labeld={\alpha_5}} %
      }{%
        |(A)| A \& \&\&\&\&\&\& \& |(X)| X \\
        \& |(A')| A' \&\&\&\&\&\& \&  \\
        |(B)| B \& \&\&\&\&\&\& \& |(Y)| Y \\
        \& \&\&\&\&\&\& |(Y')| Y' \&  \\
        |(C)| C \& \&\&\&\&\&\& \& |(Z)| Z %
      }{%
        (A) edge[pro] (B) %
        edge[pro] (A') %
        edge (X) %
        (A') edge[pro] node[pos=.5,inner sep=0pt,anchor=south east]
        {$\scriptstyle w_2$} (B) %
        edge (Y) %
        (B) edge (Y') %
        edge[pro] (C) %
        (C) edge (Z) %
        (X) edge[pro] (Y) %
        (Y) edge[pro] node[pos=.5,inner sep=0pt,anchor=south east]
        {$\scriptstyle u'_1$} (Y') %
        edge[pro] (Z) %
        (Y') edge[pro] (Z) %
      } 
    \end{minipage}
    \hfil
    \begin{minipage}[t]{0.25\linewidth}
      \centering
      \Diag (1.5,1.5) {%
        \twocellbr{B}{A}{X}{\alpha_1} %
        \twocellbr{C}{B}{Y}{\alpha_2} %
      }{%
        |(A)| A \& |(X)| X \\
        |(B)| B \& |(Y)| Y \\
        |(C)| C \& |(Z)| Z %
      }{%
        (A) edge (X) %
        edge[pro] (B) %
        (B) edge[pro] (C) %
        edge (Y) %
        (X) edge[pro] (Y) %
        (Y) edge[pro] (Z) %
        (C) edge (Z) %
      }
    \end{minipage}
    \hfil
    \begin{minipage}[t]{0.35\linewidth}
      \centering \Diag (.5,.25) {%
        \twocell[.2][.2]{B}{A}{X}{}{celllr={0.0}{0.0},bend
          right=20,labelbr={\alpha_1}} %
        \twocell[.4][.2]{B'}{B}{X'}{}{celllr={0.0}{0.0},bend
          right=20,labelr={\alpha_2}} %
        \twocell[.4][.2]{C}{B'}{Y}{}{celllr={0.0}{0.0},bend
          right=30,labelbr={\alpha_3}} %
        \twocell[.2][.4]{C}{B}{B'}{}{celllr={0.0}{0.0},bend
          right=10,labeld={\alpha_4}} %
        \twocell[.4][.2]{X'}{X}{Y}{}{celllr={0.0}{0.0},bend
          right=10,labeld={\alpha_5}} %
      }{%
        |(A)| A \& \&\&\&\&\&\& \& |(X)| X \\
        \&  \&\&\&\&\&\& |(X')| X' \&  \\
        |(B)| B \& \&\&\&\&\&\& \& |(Y)| Y \\
        \& |(B')| B' \&\&\&\&\&\&  \&  \\
        |(C)| C \& \&\&\&\&\&\& \& |(Z)| Z %
      }{%
        (A) edge[pro] (B) %
        edge (X) %
        (B) edge (X') %
        edge[pro] node[pos=.7,inner sep=0pt,anchor=south west]
        {$\scriptstyle v_1$} (B') %
        edge[pro] (C) %
        (B') edge (Y) %
        edge[pro] (C) %
        (C) edge (Z) %
        (X) edge[pro] (X') %
        edge[pro] (Y) %
        (X') edge[pro,labelbl={u_2}] (Y) %
        (Y) edge[pro] (Z) %
      }
    \end{minipage}
  \end{center}
  with $\length{w_2} > 0$, $\length{v_1} > 0$, and $\alpha_4$ and
  $\alpha_5$ iso in $\DH$.
\end{prop}
A possible reading of this is that in the left and middle cases, the
whole of $v$ embeds into $u'$. In the left case, a non-trivial part of
$w$ embeds into $u'$.  In the right case, a
non-trivial part of $v$ embeds into $u$.
\begin{proof}
  Choose decompositions of $u'$ and $u$ as $M_1 \vrond \ldots \vrond
  M_p$ and $M_{p+1} \vrond \ldots \vrond M_{p+q}$, respectively, and
  of $v$ as $b_1 \vrond \ldots \vrond b_n$. Apply Lemma~\ref{lem:big}
  to obtain $f \colon n \to p+q$. If $f(n) > p$, we are in the
  right-hand case.  If $f(n) = p$, we are in the middle case. If $f(n)
  = r < p$, let $u'_2 = M_1 \vrond \ldots M_r$ and $u'_1 =
  M_{r+1}\vrond \ldots \vrond M_p$.  Lemma~\ref{lem:big} provides
  $\beta \colon w \to u'_1 \vrond u$ and $\gamma \colon v \to u'_2$
  such that $\gamma \vrond \beta = \alpha$.
  Applying~\axref{leftdecomposition} to $\beta$ gives a decomposition
  of $\alpha$ as on the left below
  \begin{center}
    \Diag (.5,.25) {%
        \twocell[.4][.2]{A'}{A}{X}{}{celllr={0.0}{0.0},bend
          right=10,labelbr={\beta_1}} %
        \twocell[.4][.2]{B}{A'}{Y}{}{celllr={0.0}{0.0},bend
          right=20,labelbr={\beta_2}} %
        \twocell[.2][.2]{C}{B}{Y'}{}{celllr={0.0}{0.0},bend
          right=20,labelbr={\gamma}} %
        \twocell[.2][.4]{B}{A}{A'}{}{celllr={0.0}{0.0},bend
          right=10,labeld={\alpha_4}} %
        \twocell[.4][.2]{Y'}{Y}{Z}{}{celllr={0.0}{0.0},bend
          right=10,labeld={\alpha_5}} %
      }{%
        |(A)| A \& \&\&\&\&\&\& \& |(X)| X \\
        \& |(A')| A' \&\&\&\&\&\& \&  \\
        |(B)| B \& \&\&\&\&\&\& \& |(Y)| Y \\
        \& \&\&\&\&\&\& |(Y')| T \&  \\
        |(C)| C \& \&\&\&\&\&\& \& |(Z)| Z %
      }{%
        (A) edge[pro] (B) %
        edge[pro] (A') %
        edge (X) %
        (A') edge[pro] node[pos=.5,inner sep=0pt,anchor=south east]
        {$\scriptstyle w_2$} (B) %
        edge (Y) %
        (B) edge (Y') %
        edge[pro] (C) %
        (C) edge (Z) %
        (X) edge[pro] (Y) %
        (Y) edge[pro] node[pos=.5,inner sep=0pt,anchor=south east]
        {$\scriptstyle u'_1$} (Y') %
        edge[pro] (Z) %
        (Y') edge[pro,labell={u'_2}] (Z) %
      }
\hfil
  \Diag{%
        \twocell[.3][.4]{Bhaut}{A}{A'haut}{}{celllr={0.0}{0.0},bend
          right=10,labelbr={\alpha_4},shorten <=-10} %
  }{%
   |(A)| A \\
   \& |(A'haut)| A' \\
   |(Bhaut)| B \& \& |(A')| A' \\
   |(B)| B  %
  }{%
    (A) edge[pro] (A'haut) %
    edge[pro,bend left=40] (A') %
    edge[pro,bend right=30] (Bhaut) %
    (A'haut) edge[pro,identity] (A') %
    edge[pro,labelal={w_2}] (Bhaut) %
    (Bhaut) edge (A') %
    edge[pro,identity] (B) %
    (A') edge[pro,labelbr={w_2}] (B) %
  }
\end{center}
with $\alpha_4$ and $\alpha_5$ isos. If $\length{w_2} \neq 0$, then we
are in the left-hand case of the proposition, and the middle case is
impossible by essential uniqueness
in~\axref{leftdecomposition}. Otherwise, we may decompose $\alpha_4$
as on the right by atomicity (empty cells are given by coherence or
\axref{atomicity}), so we are in the middle case of the proposition.
\end{proof}

Lastly, we need a few more definitions before
Proposition~\ref{prop:decompV}.
\begin{defi}
  Let $\EVi$ be the full subcategory of $\E$ consisting of views.
\end{defi}

Consider, for any $X$, the comma category $\E_X$ induced by the
vertical codomain functor $\cod \colon \E \to \Dh$
mapping~\eqref{eq:alpha} to $k$ (following notation from
Section~\ref{subsec:prelim:cats}).  Similarly, consider $\EVi_X$.
Concretely, an object of $\EVi_X$ is a pair of a view $v \colon d'
\proto d$, and a player $x \colon d \to X$ of $X$.  A morphism
$(v_1,x_1) \to (v_2,x_2)$ is a morphism $(w,\alpha) \colon v_1 \to
v_2$ in $\EVi$, such that $x_2 \rond \cod (\alpha) = x_1$.  

Recall now from above Definition~\ref{def:beh} the pullback category
$\E (X)$. It is isomorphic to the full subcategory of $\E_X$
consisting of pairs $(u,x)$ where $x = \id_X$.  Similarly, we have
$\EVi(X)$, which is empty unless $X$ is an individual.
   \begin{prop}\label{prop:decompV} We have
     \begin{enumerate}[label=(\roman*)]
     \item The inclusion $\EVi(d) \into \EVi_d$ mapping $v$ to 
       $(v, \id_d)$ is an isomorphism of categories. \label{prop:decompV:i}
     \item The inclusion $\sum_{(d,x) \in \Pl (X)} \EVi(d) \into
       \EVi_X$ mapping $((d,x),v)$ to $(v,x)$ is an isomorphism of
       categories.
       \label{prop:decompV:ii}
     \item $\EVi(d)$ is a preorder.\label{prop:decompV:iii}
     \end{enumerate}
     % $\EVi_d \iso \EVi (d)$, $\EVi_X \iso \sum_{(d,x) \in \Pl (X)}
     % \EVi_d$, and $\EVi (d)$ is a preorder.
   \end{prop}
   \begin{proof}
     First, because $\DI$ is discrete, $\Dh (d,d) = \ens{\id_d}$,
     hence~\ref{prop:decompV:i}. For~\ref{prop:decompV:ii}, the
     functor $\EVi_X \to \sum_{(d,x) \in \Pl (X)} \EVi(d)$ mapping any
     $(v,x)$ to $((d,x), v)$, with $v \colon d' \proto d$ a view and
     $x \colon d \to X$ a player, is inverse to the given functor.
     Finally, consider any two morphisms $v_1 \to v_2$ in $\EVi (d)$,
     say
\begin{center}
  \diag{%
    |(d'_1)| X_1 \& |(d_2)| d_2 \\
    |(d_1)| d_1 \& \\
    |(d)| d \& |(d')| d %
  }{%
    (d'_1) edge[labelu={h_1}] (d_2) %
    edge[pro,labell={w_1}] (d_1) %
    (d_1) edge[pro,labell={v_1}] (d) %
    (d_2)  edge[pro,twor={v_2}] (d') %
    (d) edge[identity] (d') %
    (d_1) edge[cell=.3,labelu={\alpha_1}] (r) %
  }
\hfil and \hfil
   \diag{%
    |(d'_1)| X_2 \& |(d_2)| d_2 \\
    |(d_1)| d_1 \& \\
    |(d)| d \& |(d')| d. %
  }{%
    (d'_1) edge[labelu={h_2}] (d_2) %
    edge[pro,labell={w_2}] (d_1) %
    (d_1) edge[pro,labell={v_1}] (d) %
    (d_2)  edge[pro,twor={v_2}] (d') %
    (d) edge[identity] (d') %
    (d_1) edge[cell=.3,labelu={\alpha_2}] (r) %
  }
\end{center}
Fixing decompositions of $v_1$ and $v_2$ into basic moves, we obtain
by Lemmas~\ref{lem:big} and~\ref{lem:preorder} that $\alpha_1$ and
$\alpha_2$ respectively decompose as
\begin{center}
  \diag{%
    |(d'_1)| X_1 \& |(d_2)| d_2 \\
    |(d_1)| d_1 \& |(d'')| d' \\
    |(d)| d \& |(d')| d %
  }{%
    (d'_1) edge[labelu={h_1}] (d_2) %
    edge[pro,twoleft={w1}{w_1}] (d_1) %
    (d_1) edge[pro,twoleft={v1}{v_1}] (d) %
    edge (d'') %
    (d_2)  edge[pro,twor={v^1_2}] (d'') %
    (d'')  edge[pro,tworight={r'}{v^2_2}] (d') %
    (d) edge[identity] (d') %
    (w1) edge[cell=.3,labelu={\alpha^1_1}] (r) %
    (v1) edge[cell=.3,labelu={\alpha^2_1}] (r') %
  }
\hfil and \hfil
  \diag{%
    |(d'_1)| X_2 \& |(d_2)| d_2 \\
    |(d_1)| d_1 \& |(d'')| d' \\
    |(d)| d \& |(d')| d. %
  }{%
    (d'_1) edge[labelu={h_2}] (d_2) %
    edge[pro,twoleft={w1}{w_2}] (d_1) %
    (d_1) edge[pro,twoleft={v1}{v_1}] (d) %
    edge (d'') %
    (d_2)  edge[pro,twor={v^1_2}] (d'') %
    (d'')  edge[pro,tworight={r'}{v^2_2}] (d') %
    (d) edge[identity] (d') %
    (w1) edge[cell=.3,labelu={\alpha^1_2}] (r) %
    (v1) edge[cell=.3,labelu={\alpha^2_2}] (r') %
  }
\end{center}
By Corollary~\ref{cor:Vpreordergroupoid}, $\alpha^2_1 =
\alpha^2_2$. Furthermore, we conclude by~\axref{fibration} and the
quotienting~\eqref{eq:quotient:E} in the definition of $\E$ that both
morphisms are equal in $\EVi(d)$ to $\alpha^2_1 \vrond \id_{v^1_2}$.
\end{proof}

\subsection{From behaviours to strategies}
   \begin{defi}
     The category $\SS_X$ of \emph{strategies} on $X$ is the category
     $\OPsh{\EVi_X}$ of presheaves of finite ordinals on $\EVi_X$.
   \end{defi}
   \begin{exa}
     On $\Dccs$, $\EVi_X$ as defined here yields a category equivalent
     to the definition in \citetalias{2011arXiv1109.4356H}, so the
     categories of strategies are also equivalent (even isomorphic
     because $\ford$ contains no non-trivial automorphism).
   \end{exa}
   The rest of this section develops some structure on strategies,
   which is needed for constructing the \lts{} in
   Section~\ref{subsec:lts:strats}. We start by extending the
   assignment $X \mapsto \SS_X$ to a pseudo double functor $\op\D \to
   \QCat$, where $\QCat$ is Ehresmann's double category of
   \emph{quintets} on the 2-category $\Cat$:
   \begin{defi}
     $\Quintets\Cat$ has small categories as objects, functors as both
     horizontal and vertical morphisms, and natural transformations as
     double cells.
   \end{defi}

   Actually, our first step is to extend the assignment $X \mapsto \EVi_X$
   to pseudo double functor $\D \to \Quintets\Cat$.  Define the action
   of a horizontal map $h \colon X \to X'$ to map any object $(v,x)$
   of $\EVi_X$ to $(v, h \rond x)$, and any morphism to itself viewed
   as a morphism in $\EVi_{X'}$. (This functor is induced by universal
   property of $\EVi_X$ as a comma category.) This defines a functor
   $\EVi_{-} \colon \Dh \to \Cat$.  The pseudo functor $\Dv \to \Cat$
   is a bit harder to construct. For any $u \colon Y \proto X$ in
   $\Dv$ and $y \colon d \to Y$, the cell $\alpha^{y,u}$ from
   Proposition~\ref{prop:views} induces a functor $\cocob{v^{y,u}}
   \colon \EVi (d) \to \EVi (d^{y,u})$ mapping any $v \colon d' \proto
   d$ to $v^{y,u} \vrond v$.  Composing with the coproduct injection
   $\inj_{d^{y,u}, y^u} \colon \EVi (d^{y,u}) \into \sum_{(d'',x) \in
     \Pl (X)} \EVi (d'')$, because $\EVi_X \iso \sum_{(d'',x) \in \Pl
     (X)} \EVi (d'')$, we obtain functors $$\EVi (d)
   \xto{\cocob{v^{y,u}}} \EVi (d^{y,u}) \xinto{\inj_{d^{y,u},y^u}}
   \EVi_X,$$ whose copairing defines a functor $\EVi_{u} \colon
   \EVi_Y \to \EVi_X$.
   % As above, we thus obtain a functor 
   % $$\Cat (\op{u_!}, \ford) \colon \Cat (\op{(\EVi_{X})},\ford) 
   % \to \Cat (\op{(\EVi_{Y})},\ford),$$ 
   % which we denote by $\SS_u \colon \SS_{X} \to \SS_{Y}$. 

   Now, for any cell as on the left below,
   we obtain by Proposition~\ref{prop:views} a canonical natural isomorphism as on the right
     \begin{mathpar}
       \begin{minipage}[t]{0.3\linewidth}
         \doublecellpro{Y}{Y'}{X}{X'}{k}{u}{u'}{h}{\alpha}
       \end{minipage} \and
       \begin{minipage}[t]{0.3\linewidth}
         \doublecell{\EVi_Y}{\EVi_{Y'}}{\EVi_X}{\EVi_{X'}.}{\EVi_k}{\EVi_{u}}{\EVi_{u'}}{\EVi_h}{\iso}
       \end{minipage}
     \end{mathpar}
     By canonicity of the above double cell, we have
   \begin{prop}\label{prop:pseudodoublefunctor}
     This assignment defines a pseudo double functor $\EVi_{-} \colon
     \D \to \Quintets\Cat$.
   \end{prop}

   \begin{defi}
     Let the \emph{opposite} $\op{\D}$ of a pseudo double category
     $\D$ be obtained by reversing both vertical and horizontal
     arrows, and hence double cells.
   \end{defi}
   % (In the literature, this would probably rather be denoted by
   % $\D^{\mathit{co}\mathit{op}}$. We stick to our notation for
   % conciseness.)  

   We obtain:
   \begin{defi}\label{def:SSfunctor}
     Let $\SS \colon \op\D \to \Quintets\Cat$ be the composite
     $\op\D \xto{\op{(\EVi_{-})}} \op{\Quintets\Cat} \xto{\OPsh{-}} \Quintets\Cat.$
   \end{defi}
   As a shorthand, we denote $\SS (f) (S)$ by $S \cdot f$ for $f$
   horizontal or vertical.  Concretely, for any horizontal $h \colon Z
   \to X$, $S \cdot h$ satisfies $$(S \cdot h) (v,z) = S(v,h
   \rond z),$$ whereas for any vertical $u \colon Y \proto X$, $S
   \cdot u$ satisfies
   $$(S \cdot u) (v,y) = S (v^{y,u} \vrond v, y^u).$$

   % \begin{prop} 
   %   The map $X \mapsto \SS_X$ extends to a double functor $\SS 
   %   \colon \op{\D} \to \Cat$. 
   % \end{prop} 

%\subsection{From behaviours to strategies}
We conclude this section by constructing the \emph{extension}
functor from strategies to behaviours, in arbitrary playgrounds.

Recall that strategies on a position $X$ are presheaves of finite
ordinals on $\EVi_X$, and that behaviours are presheaves of finite
sets on $\E(X)$.  To go from the former to the latter, we use $\E_X$
as a bridge.  Recall from Section~\ref{subsec:prelim:cats} that
objects of $\EVi_X$ are diagrams of the shape $d' \xproto{v} d \xto{x}
X$, with $v$ a view, and that objects of $\E(X)$ are just plays $Y
\xproto{u} X$. The idea here is that on the one hand $\EVi_X$ is richer
than $\E(X)$, in that its objects may be plays on \emph{subpositions}
of $X$, whereas objects of $\E(X)$ are plays on the whole of $X$.  But
on the other hand, $\E(X)$ is richer than $\EVi_X$ because its objects
may be arbitrary plays, whereas objects of $\EVi_X$ have to be
views. $\E_X$ contains both $\EVi_X$ and $\E(X)$, its objects
being diagrams $Y \xproto{u} Z \xto{h} X$, for arbitrary plays $u$.

First, let $k_X \colon \OPsh{\EVi_X} \to \FPsh{\EVi_X}$ denote
postcomposition with $\ford \into \set$.  Because views form a full
subcategory of $\DH$, all embeddings $i_X \colon \EVi_X \into \E_X$
are also full.  This entails:
\begin{lem}\label{lem:kanff}
  For all $X$, right Kan extension $(\op{i_X})_\star \colon \FPsh{\EVi_X}
  \into \FPsh{\E_X}$ along $\op {i_X}$ is well-defined, full, and faithful.
\end{lem}

\begin{proof}
  One easily shows that, when defined, right extension along a full
  and faithful functor is full and faithful. 

  It remains to show that the considered right extensions exist.  It
  is well-known~\citep{MacLane:cwm} that the right Kan extension of
  any $S \in \FPsh{\EVi_X}$ maps any $(u,h)$ to the limit of the
  functor $\op{(\EVi_X / (u,h))} \to \op{(\EVi_X)} \xto{S} \set$, if
  the latter exists.  Since finite limits exist in $\set$ (though not
  in $\ford$, which explains why we use $\set$ instead of $\ford$ for
  extending strategies), it is enough to prove that each $\EVi_X /
  (u,h)$ is essentially finite, i.e., equivalent to a finite
  category. This is proved in the next lemma.
\end{proof}

\begin{lem}\label{lem:kanexists}
  For any play $u \colon Z \proto Y$ and horizontal $h \colon Y \to X$, the category $\EVi_X / (u,h)$ is
  essentially finite.
\end{lem}

For this lemma to hold, we need more axioms.
\begin{ax}
  \begin{axioms}
  \item (Finiteness) For any position $X$, there are only finitely
    many players, i.e., the category $\DI / X$ is finite. \label{finiteness}
    % Furthermore, up to isomorphism in $\DH (X)$, there are only 
    % finitely many moves with initial position $X$.\label{finiteness} 
    % 
  % \item For any basic move $b \colon d' \to d$ and any move $M \colon 
  %   X \to d$, there is at most one double cell $b \to M$ in $\DH$. \label{charac} 
  \end{axioms}
\end{ax}

\proof[Proof of Lemma~\ref{lem:kanexists}] Let us fix a pair $(u,h)$. By
Proposition~\ref{prop:decompV}, $\EVi_X / (u,h)$ is a preorder, so we just
need to prove that its object set is essentially finite.  Now, letting
$n = \length{u}$, we fix a decomposition of $u$ into moves, say $Z = Y_n \xproto{M_n} Y_{n-1}
\ldots Y_1 \xproto{M_1} Y_0$. For any morphism $\alpha \colon (v,x) \to (u,h)$ in
$\E_X$, by Lemma~\ref{lem:big}, $m = \length{v}$ may not exceed $n$.  Furthermore, by
Lemma~\ref{lem:big}, Proposition~\ref{prop:views}, and our
quotienting~\eqref{eq:quotient:E}, any such $\alpha$ is determined up
to isomorphism by $m$, a strictly monotone map $f \colon m \to n$, and
a player $y$ of $Y_{f(m)}$. Because such triples $(m,f,y)$ are in finite number,
$\EVi_X / (u,h)$ is essentially finite. \qed

  This concludes the proof of Lemma~\ref{lem:kanff}: right
  Kan extension along $\op{i_X} \colon \op{(\EVi_X)} \into \op{\E_X}$
  yields a full and faithful functor. We now design the second half of
  our bridge from $\EVi_X$ to $\E(X)$ via $\E_X$. Consider the
  embedding $j_X \colon \E (X) \into \E_X$ mapping any $u$ to $(u,
  \id_X)$. Restriction along $\op{(j_X)}$ defines a functor
  $\cob{\op{(j_X)}} \colon \FPsh{\E_X} \to \FPsh{\E (X)}$.

Recall from Definition~\ref{def:beh} the notion of behaviour.
\begin{defi}
  For any $X$, let the \emph{extension} functor $\extfun{X} \colon
  \SSX \to \Beh{X}$ be the composite 
  $$\OPsh{\EVi_X} \xto{k_X} \FPsh{\EVi_X} \xto{(\op{i_X})_\star} \FPsh{\E_X} \xto{\cob{\op{(j_X)}}} \FPsh{\E(X)}.$$
  We call a behaviour on $X$ \emph{innocent} when it is in the
  essential image of $\extfun{X}$.
\end{defi}
Notation: when $X$ is clear from context, we abbreviate $\ext{X}{S}$ as $\exta{X}{S}$.

\begin{rem}
  The calculations of Section~\ref{subsubsec:strategies} carry over
  unchanged to the new setting.
\end{rem}
Finally, the definitions of Section~\ref{subsec:fair} apply more or
less verbatim to the playground $\Dccs$, yielding a semantic fair
testing equivalence which coincides with that of \citetalias{2011arXiv1109.4356H}.

\section{Playgrounds: transition systems}\label{sec:strats}
In the previous section, we have defined behaviours and strategies,
and constructed the extension functor from the former to the latter.
In this section, we first build on this to state decomposition
theorems, which lead to a syntax and \anlts{} for strategies. Then, we
define our second \lts{}, and relate the two by a strong, functional
bisimulation.

\subsection{A syntax for strategies}\label{subsec:syntax:strats}
Let us begin by proving in the abstract setting of playgrounds
analogues of the decomposition results of
\citetalias{2011arXiv1109.4356H}, in particular that strategies form a
terminal coalgebra for a certain polynomial functor. This is
equivalent to saying that they are essentially infinite terms in a
typed grammar. We use this in the next section to define our \lts{}
$\SSS_\D$, and study transitions therein.

First, we have \emph{spatial} decomposition:
   \begin{prop}\label{prop:spatial} 
     The functor $\SSX \to \prod_{(d,x) \in \Pl(X)} \SS_d$ given at
     $(d,x)$ by $\SS(x) \colon \SSX \to \SS_d$ is an isomorphism of
     categories.
   \end{prop}
   \proof
     We have:
     \begin{center}
       $
       \begin{array}[b]{rcll}
         \SS_X & = & \Cat (\op{(\EVi_{X})}, \ford) \\
         & \iso & \Cat (\sum_{(d,x) \in \Pl (X)} \op{\EVi (d)}, \ford) & \mbox{(by Proposition~\ref{prop:decompV})} \\
         & \iso & \prod_{(d,x) \in \Pl (X)} \Cat (\op{\EVi (d)}, \ford) \\
         & = & \prod_{(d,x) \in \Pl (X)} \SS_d. 
       \end{array}
       $\qed
     \end{center}
     % ICI je desobeis!
     For any $S \in \SSX$, let $S \cdot x$ denote the strategy on $d$
     corresponding to $(d,x)$ accross the isomorphism.

     The second decomposition result is less straightforward, but goes
     through essentially as in the concrete case. Let us be a bit more
     formal here than in Section~\ref{subsubsec:syntax}, by showing
     that strategies form a terminal coalgebra for some endofunctor on
     $\Set^\DI$.  We start by defining the relevant endofunctor.

     % Recall the familiar equivalence of categories, for any set $X$,
     % between sets over $X$ and $X$-indexed families of sets: $\Set / X
     % \equi \Set^X$. 
     \begin{defi}
       Let $\MMMB_d$ denotes the set of all isomorphism classes of
       basic moves from $d$ (i.e., with vertical codomain $d$).
     \end{defi}
     \begin{defi}
       Let $G \colon \Set^\DI \to \Set^\DI$ be the functor
       mapping any family $U$ to $$(G (U))_d = \left (\prod_{b \in
           \MMMB_d} U_{\dom (b)} \right )^\star,$$ where $(-)^\star$
       denotes finite sequences.
     \end{defi}

     \begin{rem}\label{rem:lists}
       This functor is polynomial in the sense of
       Kock~\citep{Kock01012011}, as $$(G (U))_d = \sum_{n \in \Nat}
       \left (\prod_{i \in n, b \in \MMMB_d} U_{\dom (b)} \right ).$$
     \end{rem}
     
     We now show that strategies, viewed as the $\DI$-indexed family
     $(\ob(\SS_d))_{d \in \DI}$, form a terminal $G$-coalgebra. We drop 
     the $\ob(-)$ for readability.
     
     \begin{defi}\label{defi:restr}
       For any $S \in \SS_d$ and $\state \in S (\idv_d)$, let the
       \emph{restriction} $\restr{S}{\state} \in \SS_d$ of $S$ to
       $\state$ be defined by the fact that $\restr{S}{\state} (v) =
       \ens{\state' \in S (v) \aalt S(!_v)(\state') = \state}$.
     \end{defi}
     (Here, we freely use the isomorphism $\EVi_d \iso \EVi (d)$ from
     Proposition~\ref{prop:decompV}, and let $!_v$ denote the unique
     morphism $\idv_d \to v$ in $\EVi(d)$.)

     In view of Remark~\ref{rem:lists}, $(G(\SS))_d = \sum_{n
       \in \Nat} \left (\prod_{b \in \MMMB_d} \SS_{\dom (b)}
     \right )^n$.  We thus may define the $G$-coalgebra structure
     $\deriv \colon \SS \to G (\SS)$ in $\Set^\DI$ of
     strategies as follows.
     \begin{defi}
       Let, for all $d \in \DI$, $\deriv_d \colon \SS_d \to
       \sum_n (\prod_{b \in \MMMB_d} \SS_{\dom (b)})^n$ send any $S
       \in \SS_d$ to $n = S (\idv_d)$ and the map
       $$
       \begin{array}{rcll}
         S(\idv_d) & \to & \prod_{b \in \MMMB_d} \SS_{\dom (b)} \\
         \state & \mapsto & b \mapsto (\restr{S}{\state}) \cdot b.
       \end{array}
       $$
     \end{defi}
     Here, we view the ordinal $S (\idv_d)$ as a natural number, and
     the given map $S(\idv_d) \to \prod_{b \in \MMMB_d} \SS_{\dom
       (b)}$ as a list of elements of $\prod_{b \in \MMMB_d} \SS_{\dom
       (b)}$.  We further use the action of $b$ on $S$, as below
     Definition~\ref{def:SSfunctor}.
We have:
\begin{thm}\label{thm:stratcoalg}
  The map $\deriv \colon \SS \to G(\SS)$ makes $\SS$ into a terminal
  $G$-coalgebra.
\end{thm}
This intuitively means that strategies, on individuals, are
infinite terms for the following typed grammar with
judgements $d \vdashdefinite D$ and $d \vdash S$, where $D$ is a
\emph{definite strategy} and $S$ is a \emph{strategy}
\begin{mathpar}
\inferrule{\ldots \ d' \vdash S_b \ \ldots \ {(\forall b \colon d' \proto d \in \MMMB_d)}
}{
d \vdashdefinite \langle (S_b)_{b \in \MMMB_d} \rangle
}
\and
\inferrule{\ldots \  d \vdashdefinite D_i \  \ldots \ (\forall i \in n)}{
d \vdash \bigoplus_{i \in n} D_i}~(n \in \Nat). 
\end{mathpar}
Semantically, definite strategies correspond to strategies $S$ such
that $S (\idv_d) = 1$, which will play a crucial role in the \lts{}
below.

The rest of this section is a proof of Theorem~\ref{thm:stratcoalg}. 

First of all, we construct an inverse to $\deriv$.  
\begin{defi}
  Consider $B = (B_1, \ldots, B_n) \in (G(\SS))_d$. For any view $v
  \colon d' \proto d$, define $\deriv'(B) \in \SS_d$ by
$$\deriv'(B) (v) =   \left \{ \begin{array}{ll}
n & \mbox{if $v =
  \idv_d$} \\
\sum_{i \in n} B_i(b)(v') & \mbox{if $v = b \vrond v'$},
\end{array} \right .$$
and on morphisms $$\deriv'(B) (v \xto{(w,\alpha)} v')(\state) =
\left \{ \begin{array}{ll}
i & \mbox{if $v' = \idv_d$ and $\state = i \in n$} \\
& \mbox{or if $v = \idv_d$ and $\state = (i,x)$} \\
(i, B_i(b)(w,\alpha_1)(x)) & \mbox{if $v = b \vrond v_1$ and $\state = (i,x)$},
\end{array} \right .$$
where in the last clause necessarily $v' \iso b' \vrond v'_1$ and
Lemma~\ref{lem:big} yields $\alpha_b \colon b \xto{\iso} b'$ and $\alpha_1 \colon v_1 \vrond w \to v'_1$
such that $\alpha_b \vrond \alpha_1 = \alpha$.
\end{defi}
\begin{lem}
  We have $\deriv' = \inv{\deriv}$.
\end{lem}
\begin{proof}
  Starting from a strategy $S \in \SS_d$, let $n = S(\idv_d)$, and
  $B_i(b) = (\restr{S}{i}) \cdot b$, for any $d' \xproto{b} d$.
  % $B_i(b)(v') = (\restr{S}{i})(b \vrond v')$, for any $d''
  % \xproto{v'} d' \xproto{b} d$.
  We have $\deriv S = (B_1, \ldots, B_n)$, and thus
  $\deriv' (\deriv S) (v) = n$ if $v = \idv_d$, and $\deriv' (\deriv
  S) (v) = \sum_{i \in n} (\restr{S}{i})(b \vrond v') = S(v)$ if $v =
  b \vrond v'$, as desired.

  Conversely, starting from $B = (B_1, \ldots, B_n) \in
  (G(\SS))_d$, let $S = \deriv' B$.  We have that $\deriv S$ has
  length $n$, and its $i$th component maps any $b \colon d' \proto d$
  to the strategy mapping any $v' \colon d'' \proto d'$ to the
  strategy $(\restr{S}{i}) \cdot b$. Thus, $(\deriv S)_i (b) (v') =
  ((\restr{S}{i}) \cdot b) (v') = (\restr{S}{i}) (b \vrond v')$.  But
  by definition, this is equal to $B_i (b) (v')$, as desired.  
\end{proof}

Consider any $G$-coalgebra $a \colon U \to G U$.

We define by induction on $N$ a sequence of maps $f_N \colon U \to
\SS$, such that for any $d$ and $u \in U_d$ the $f_{n+N}(u)$'s agree
on views of length $\leq n$. I.e., for any $d \in \DI$, $u \in U_d$,
view $v$ of length less than $n$, and any $N$, $f_{n+N} (u) (v) =
f_{n} (u) (v)$, and similarly the action of $f_{n+N} (u)$ on morphisms
between such views is the same as that of $f_{n} (u)$.

To start the induction, take $f_0 (u)$ to be the strategy mapping
$\idv_{d}$ to $\length{a (u)}$, i.e., the length of $a (u) \in
(\prod_{b} U_{\dom (b)})^\star$, and all other views to $0$.

Furthermore, given $f_N$, define $f_{N+1}$ to be
$$U \xto{a} G U \xto{G  (f_N)} G (\SS) \xto{\inv\deriv} \SS.$$

In other words, $f_N$ is 
$$U \xto{a} G U \xto{G (a)} \ldots G^{N-1} U \xto{G^{N-1} a} G^{N} U \xto{G^N f_0} 
G^N \SS \xto{G^{N-1}(\inv\deriv)} G^{N-1} (\SS) \ldots G(\SS) \xto{\inv\deriv} \SS.$$

Unfolding the definitions yields:
\begin{lem}\label{lem:unfold}
  Consider any $u \in U_d$, and let $a (u) = (z_1, \ldots, z_k)$.  For
 any $f \colon U \to
  \SS$, we have
  \begin{itemize}
  \item  $\inv\deriv(G(f)(a(u)))(\idv_d) = k$, and
  \item $\inv\deriv(G(f)(a(u)))(b \vrond v) = \sum_{i \in k} f
    (z_i (b)) (v)$ for any composable basic move $b$ and view $v$.
  \end{itemize}
\end{lem}

\begin{cor}\label{cor:unfold}
  We have, for any $N \in \Nat$, $f_N(u)(\idv_d) = k$.  
  Furthermore, for any basic move $b \colon d' \proto d$, and view $v \colon d''
  \proto d'$, we have for any $N \in \Nat$:
  $$f_{N+1} (u) (b \vrond v) = \sum_{i
    \in k} f_N (z_i (b)) (v).$$
\end{cor}

As announced, we have:
\begin{lem}\label{lem:station}
  For any view $v \colon d' \proto d$ and $n \in \Nat$,
  $f_{\length{v}+n}(u)(v) = f_{\length{v}} (u) (v)$.
\end{lem}
\proof
  We proceed by well-founded induction on $(\length{v},n)$, for the
  lexical ordering. 
  Let again $a(u) = (z_1, \ldots, z_k)$.
  First, we have $f_{\length{\idv}}(u) (\idv) = k$, and for any $n$,
  $f_{\length{\idv}+n+1}(u)(\idv) = k$ by Corollary~\ref{cor:unfold}.
  Now, if $v = b \vrond v'$, then by Corollary~\ref{cor:unfold} again:
  \begin{center}
    $\begin{array}[b]{rcll} f_{\length{v}+n+1} (u) (b \vrond v') &=&
      \sum_{i \in k} f_{\length{v}+n} (z_i(b)) (v') \\
      &=&
      \sum_{i \in k} f_{\length{v'}+n+1} (z_i(b)) (v') & \mbox{(by $\length{v} = \length{v'}+1$)}\\
      &=&
      \sum_{i \in k} f_{\length{v'}} (z_i(b)) (v') & \mbox{(by induction hypothesis)} \\
      &=& f_{\length{v}} (u) (b \vrond v') & \mbox{(by Corollary~\ref{cor:unfold} again).}
    \end{array}$ \qed
  \end{center}

  The sequence $(f_n(u))$ thus has a colimit in $\SS_d = \EVihato{d}$:
  the presheaf mapping any view $v$ to $f_{\length{v}} (u)(v)$. This
  allows us to define:
\begin{defi}
  Let $f \colon U \to \SS$ map any $u \in U_d$ to the colimit of
  the $f_N (u)$'s.
\end{defi}

\begin{lem}
  The following diagram commutes:
  \begin{center}
    \diag{%
      |(U)| U \& |(FU)| G U \\
      |(CVhatf)| \SS \& |(FCVhatf)| G  (\SS). %
    }{%
      (U) edge[labelu={a}] (FU) %
      edge[labell={f}] (CVhatf) %
      (FU) edge[labelr={G (f)}] (FCVhatf) %
      (FCVhatf) edge[labelu={\inv\deriv}] (CVhatf) %
    }
  \end{center}
\end{lem}
\proof Consider any $u \in U_d$ and view $v$, and let $a(u) = (z_1,
\ldots, z_k)$.  Let also $n = f(u)(v) = f_{\length{v}} (u) (v)$ and
$n' = \inv\deriv (G(f)(a(u)))(v)$.
  \begin{itemize}
  \item If $\length{v} = 0$, then by Lemma~\ref{lem:unfold} $n = n' =
    k$.
  \item If $v = b \vrond v'$, then by Lemma~\ref{lem:unfold} again we
    have $n' = \sum_{i \in k} f (z_i (b)) (v')$. But by definition of
    $f$, we obtain $n' = \sum_{i \in k} f_{\length{v'}} (z_i (b))
    (v')$, which is in turn equal to $f_{\length{v}} (u) (v) = n$ by
    Corollary~\ref{cor:unfold}. \qed
  \end{itemize}

\begin{cor}
  The map $f$ is a map $U \to \SS$ of $G $-coalgebras.
\end{cor}
% \begin{proof}
%   Let, for any strategy $S \in \EVihato{d}$ and $i \in S
%   (\id_{d})$, $\restr{S}{i}$ be the strategy mapping any view $v$ to
%   the fibre over $i$ of $S (v) \to S (\id_{d})$. Using the notations
%   of Lemma~\ref{lem:unfold}, we must show that for any $i \in k$, we
%   have $\restr{(f (u))}{i} (v \rond M) = f (z_i (M)) (v)$.  But
%   Lemma~\ref{lem:unfold} entails that $f (u) (v \rond M) \to f (u)
%   (\id_{d})$ is actually the coproduct over $i' \in k$ of all $f
%   (z_{i'} (M)) (v) \xto{!} 1 \xto{i'} \pi(a (u))$, so its fibre over
%   $i$ is indeed $f (z_{i} (M)) (v)$.
% \end{proof}

\begin{lem}
  The map $f$ is the unique map $U \to \SS$ of $G $-coalgebras.
\end{lem}
\begin{proof}
  Consider any such map $g$ of coalgebras, and let $a(u) = (z_1,\ldots,z_k)$. The map $g$ must be such that 
  $$g (u) (\idv_{d}) = \inv\deriv (G (g) (a (u))) (\idv_d) = k,$$ by
  Lemma~\ref{lem:unfold}.
  Furthermore, by the same lemma, it must satisfy:
  $$g (u) (b \vrond v) = \inv\deriv (G (g) (a (u))) (b \vrond v) = \sum_{i \in k}
  g (z_i (b)) (v),$$ which imposes by induction that $f = g$.
\end{proof}

The last two results directly entail Theorem~\ref{thm:stratcoalg}.

\subsection{The labelled transition system for strategies}\label{subsec:lts:strats}
In this section, we go beyond \citetalias{2011arXiv1109.4356H}, and
define \anlts{} for strategies, for an arbitrary playground $\D$. 

First, the alphabet for our \lts{} will constist of quasi-moves, in the following sense.
\begin{notation}\label{not:cartesian}
  We use the following notation for cartesian lifting
  (by~\axref{fibration}) of a play $u$ along a horizontal morphism
  $k$ (fixing a global choice of liftings):
  \begin{center}
    \doublecellpro{D_{k,u}}{X'}{Y}{X.}{h_{k,u}}{\restr{u}{k}}{u}{k}{\alpha_{k,u}}
  \end{center}
\end{notation}
\begin{defi}\label{def:quasi-move}
  A \emph{quasi-move} is a vertical morphism which locally either is a
  move or has length 0.  More precisely, a play $u \colon Y \proto X$ is
  a quasi-move iff for all players $x \colon d \to X$, $\restr{u}{x}$
  either is a move or has length 0.

  A quasi-move is \emph{full} when it locally either is a full move
  or has length 0. Let $\QF$ denote the subgraph of $\Dv$ consisting
  of full quasi-moves.
\end{defi}
Observe that a quasi-move on an individual either is a move or has
length 0.

States in our \lts{} will be the following special kind of strategies:
\begin{defi}
  A strategy $S \in \SSX$ is \emph{definite} when $\exta{X}{S}(\idv_X)
  = 1$, or equivalently when for all players $(d,x) \in \Pl (X)$,
  we have $S(\idv_d,x) = 1$.
\end{defi}

Intuitively, for any quasi-move $X' \xproto{M} X$, we would like
transitions $(X',S') \xto{M} (X,S)$ in our \lts{} to occur when $S'$
is a definite restriction of $S \cdot M$ to some state of
$\exta{X}{S}(M)$.  I.e., a transition roughly corresponds to a way for
$\exta{X}{S}$ to accept $M$.  However, $S \cdot M$ is not quite
$\exta{X}{S}(M)$ so the right notion of restriction may not be
obvious. But we have defined a notion of restriction in
Definition~\ref{defi:restr}, for strategies on individuals.  We now
define restriction for general strategies, and use this to define our
\lts{}. Finally, we elucidate the connection with $\exta{X}{S}(M)$.

Consider, for any $S \in \SSX$ and $\state \in \prod_{(d,x) \in \Pl
  (X)} S (\idv_d, x)$, and recall from below
Definition~\ref{def:SSfunctor} that $S \cdot h$ is shorthand for the
image of $S \in \SSX$ under the action of a horizontal morphism $h
\colon Y \to X$ for the horizontal part of our pseudo double functor
$\SS$.

\begin{defi}
  Let the \emph{restriction} $\restr{S}{\state} \in \SSX$ of $S$ to
  $\state$ be defined by the fact that for any player $x \colon d \to
  X$, $$(\restr{S}{\state}) \cdot x = \restr{(S \cdot x)}{\state (d,x)}.$$
\end{defi}
Concretely, we have, for any $v$, $\restr{S}{\state} (v,x) =
\ens{\state' \in S (v,x) \aalt S(!_v)(\state') = \state(d, x)}$, where
$!_v$ is the unique morphism $(\idv_d,x) \to (v,x)$ in $\EVi_X$.

We now define our \lts{} for strategies over $\QF$.
\begin{defi}
  The underlying graph $\SSS_\D$ for our \lts{} is the graph with as
  vertices all pairs $(X, S)$ where $X$ is a position and $S \in \SSX$
  is a definite strategy, and whose edges $(X',S') \to (X,S)$ are all
  full quasi-moves $M \colon X' \proto X$ such that there exists a
  state $\state \in \prod_{(d',x') \in \Pl (X')} (S \cdot M)
  (\idv_{d'}, x')$ with $S' = \restr{(S \cdot M)}{\state}.$

The assignment $(X,S) \mapsto X$ defines a morphism $p_{\SS} \colon
\SSS_\D \to \QF$ of reflexive graphs, which is our \lts{}.
\end{defi}
An alternative characterisation of transitions $(X,S) \xot{M} (X',S')$ is the existence of $\state$ such that
$$S' \cdot x' = \restr{(S \cdot M \cdot x')}{\state(d',x')} = \restr{(S \cdot (x')^M \cdot v^{x',M})}{\state(d',x')}$$
for all $(d',x') \in \Pl (X')$.

Let us now return to the connection between $\prod_{(d',x') \in \Pl (X')} (S \cdot M) (\idv_{d'}, x')$
and $\exta{X}{S}(M)$.
First, we have
by definition $(S \cdot M) (\idv_{d'}, x') = S (v^{x',M}, (x')^M)$,
for any player $x' \colon d' \to X'$.  
Now, as recalled above, $\exta{X}{S}(M)$ may be characterised as a
limit of
  $$\op{(\EVi_X / M)} \xto{\dom} \op{(\EVi_X)} \xto{S} \ford \into \set.$$ 
  Since $\alpha^{x',M} \colon v^{x',M} \to M$ is an object in
  $\op{(\EVi_X / M)}$, we obtain by projection a map $\exta{X}{S} (M)
  \to S(v^{x',M},(x')^M)$.

\begin{defi}
  For any $S \in \SSX$, let $\psi_M \colon \exta{X}{S}(M) \to
  \prod_{(d',x') \in \Pl (X')} S(v^{x',M},(x')^M)$ denote the
  corresponding tupling map.
\end{defi}

\begin{prop}\label{prop:SM:SdotM}
  For any definite $S \in \SSX$, the map $\psi_M$ is a bijection.
\end{prop}
We prove this through the following lemma.
For any full quasi-move $M \colon X' \proto X$, observe that for any
player $x' \colon d' \to X'$, $v^{x',M}$ has length at most $1$
(consider $\restr{M}{(x')^M}$), and let $$\Pl_M (X') = \ens{(d',x')
  \in \Pl (X') \aalt \length{v^{x',M}} \neq 0}.$$
\begin{lem}\label{lem:psibij}
  For any definite $S \in \SSX$, and full quasi-move $M \colon X' \proto X$,
  the map
  $$\exta{X}{S}(M) \xto{\psi_M} \prod_{(d',x') \in \Pl (X')} S (v^{x',M}, (x')^M) 
  \to \prod_{(d',x') \in \Pl_M (X')} S(v^{x',M}, (x')^M),$$ where the
  second map is by projection, is bijective.
\end{lem}
\begin{proof}
  Recall that $\exta{X}{S}(M)$ is a limit of
  $$\op{(\EVi_X / M)} \xto{\dom} \op{(\EVi_X)} \xto{S} \ford \into \set,$$ 
  and consider the poset $P$ with underlying set $\Pl (X) + \Pl_M
  (X')$ and ordering given by $(d,x) < (d',x')$ iff $x = (x')^M$.
  Consider the functor $p \colon P \to \EVi_X / M$ mapping any $(d,x)
  \in \Pl(X)$ to the unique morphism $\idv_d \to M$ with lower border
  $x$, and any $(d',x') \in \Pl_M (X')$ to $\alpha^{x',M}$. Since $P$
  is a poset, $p$ is faithful.  It is furthermore full by
  Proposition~\ref{prop:decompV}.  Finally,
  for any $(w,\alpha) \colon v \to M$ in $\EVi_X / M$,
  \begin{itemize}
  \item either $\length{v} = 0$ and there is a unique player $x \colon
    d \to X$ such that $(w,\alpha)$ is the (unique) morphism $\idv_d \to M$
    with lower border $x$,
  \item or $\length{v} = 1$ and there exists a unique player $(d',x')
    \in X'$ such that $(w,\alpha) = (\id,\alpha^{x',M})$ (let $x =
    \cod (\alpha)$; $\length{\restr{M}{x}} = 1$, so by
    Proposition~\ref{prop:views:decomp} $\length{w} = 0$).
  \end{itemize}
  This entails that $p$ is essentially surjective on objects, hence an equivalence.
  Thus, $\exta{X}{S}(M)$ is also a limit of 

  $$\op{P} \equi \op{(\EVi_X / M)} \xto{\dom} \op{(\EVi_X)} \xto{S} \ford \into \set.$$ 

  But now, because $S$ is definite, this functor maps any $(d,x) \in
  \Pl (X)$ to a singleton, hence $\exta{X}{S}(M)$ is also a limit of
$$\Pl_M (X')  \into \op{P} \equi \op{(\EVi_X / M)} \xto{\dom} \op{(\EVi_X)} \xto{S} \ford \into \set,$$ 
i.e., isomorphic to $\prod_{(d',x') \in \Pl_M (X')} S(v^{x',M}, (x')^M)$,
as desired.
\end{proof}

\begin{proof}[Proof of Proposition~\ref{prop:SM:SdotM}]
  If $\length{v^{x',M}} = 0$, then $S(v^{x',M}, (x')^M)$ is a
  singleton. Thus, the second map of Lemma~\ref{lem:psibij} is bijective, hence so is $\psi_M$.
\end{proof}

The moral of Proposition~\ref{prop:SM:SdotM} is that transitions
$(X,S) \xot{M} (X',S')$ in $\SSS_\D$ are precisely given by full
quasi-moves $M \colon X' \proto X$ such that there exists a state
$\state \in \exta{X}{S} (M)$ with
$$S' = \restr{(S \cdot M)}{\psi_M(\state)},$$
for all $(d',x') \in \Pl (X')$.

We now give more syntactic characterisations of transitions, starting 
with transitions from states of the shape $(d,S)$.
Recall the syntax for strategies below Theorem~\ref{thm:stratcoalg}.
\begin{prop}
  If $S = \with{(S_b)_{b \in \MMMB_d}}$ is a definite strategy on $d
  \in \DI$, and if for all $b \in \MMMB_d$, $S_b = \bigoplus_{i \in
    n_b} D^b_i$ for definite $D^b_i$, then for any $M \colon X' \proto
  d$ we have $(d,S) \xot{M} (X',S')$ iff
  \begin{itemize}
  \item for all $(d',x') \in \Pl_M (X')$, there exists $i_{x'} \in
    n_{v^{x',M}}$ such that $S' \cdot x' = D^{v^{x',M}}_{i_{x'}}$,
  \item and for
    all $(d',x') \in \Pl (X') \setminus \Pl_M (X')$, $S' \cdot x' = S$.
  \end{itemize}
\end{prop}
Let us now characterise transitions from arbitrary positions in terms
of their restrictions to individuals.  Recalling
Notation~\ref{not:cartesian}, we have:
\begin{prop}\label{prop:localtrans}
  We have $(X,S) \xot{M} (X',S')$ iff for all $(d,x) \in \Pl (X)$, 
  $$(d,S \cdot x) \xot{\restr{M}{x}} (D_{x,M}, S' \cdot h_{x,M}).$$
\end{prop}

Putting both previous results together, we obtain:
\begin{cor}
  Let, for all $(d,x) \in \Pl (X)$,  $S \cdot x = \with{(S^x_b)_{b \in
      \MMMB_d}}$ and for all $b \in \MMMB_d$, $S^x_b = \bigoplus_{i
    \in n^x_b} D^{x,b}_i$ for definite $D^{x,b}_i$.

  Then, for any $M \colon X' \proto X$, we have $(X,S) \xot{M}
  (X',S')$ iff
  \begin{itemize}
  \item for all $(d',x') \in \Pl_M (X')$, there exists $i_{x'} \in
    n^{(x')^M}_{v^{x',M}}$ such that $S' \cdot {x'} = D^{(x')^M, v^{x',M}}_{i_{x'}}$,
  \item and for
    all $(d',x') \in \Pl (X') \setminus \Pl_M (X')$, $S' \cdot {x'} = S \cdot {(x')^M}$.
  \end{itemize}
\end{cor}

\subsection{Process terms}\label{sec:syn}
In the previous section, starting from a playground $\D$, we have
constructed \anlts{} $\SSS_\D$ of strategies. We now begin the
construction of the \lts{} $\TTT_\D$ of \emph{process terms} announced
in Section~\ref{subsec:overview}, starting with process terms themselves.

\begin{defi}\label{defn:MMM}
  For any $X$, let $\MMMF_X$ be the set of isomorphism classes of full
  moves with codomain $X$, in $\DH (X)$, and let $\BsofF$ denote the map
  $$\begin{array}[t]{rcl}
    \MMMF_d & \to & \powfin(\MMMB_d) \\
    M & \mapsto & \ens{ [b] \in \MMMB_d \aalt \exists \alpha \in \DH(b,M) }.
  \end{array}$$

  Let $\MMMFB_d$ denote the subset of $\MMMF_d$ consisting of
  (isomorphism classes of) full moves $M \colon X' \proto d$ such that
  $\Pl_M (X')$ is a singleton (and hence so is $\BsofF (M)$).  Let
  $\MMMFplus_d$ denote the complement subset.
\end{defi}
The map $\BsofF$ is easily checked to be well-defined.

We state one more axiom to demand that basic sub-moves of a full move
$[M] \in \MMMF_d$ may not be sub-moves of other full moves.
\begin{ax}
  \begin{axioms}
  \item (Basic vs.\ full) For any $d \in \DI$ and $M, M'
    \in \MMMF_d$, if $M \neq M'$, then $\BsofF(M) \cap \BsofF(M')
    = \emptyset$.
    \label{basic:full}
  \end{axioms}
\end{ax}

  Let \emph{process terms} be infinite terms in the typed grammar:
    \begin{mathpar}
        \inferrule{%
          \ldots \ d_i \vdash T_i \ \ldots \ (\forall i \in n) %
        }{%
          d \vdash \sum_{i \in n} M_i.T_i %
        }~(n \in \Nat; \forall i \in n, M_i \in \MMMFB_d  \mbox{\ and\ } \BsofF[M_i] =
    \ens{b_i \colon d_i \proto d})
\and
  \inferrule{ \ldots \ d' \vdash T_b \ \ldots \ {(\forall
      (b \colon d' \proto d) \in \BsofF [M])} }{ d \vdash M \with{(T_b)_{b
        \in \BsofF [M]}} %
  }~{(M \in \MMMFplus_d).}
    \end{mathpar}

    The first rule is a guarded sum, in a sense analogous to guarded
    sum in CCS. It should be noted that guards have to be full moves
    with only one non-trivial view. There is good reason for that,
    since allowing general moves as guards would break bisimilarity
    between process terms and strategies. To understand this, consider
    a hypothetical guarded sum $R = (P|Q) + (P'|Q')$. Since this has no
    interaction before the choice is made, $R$ behaves, in CCS, just
    like an internal choice $(P|Q) \oplus (P'|Q')$. However, our
    translation to strategies does not translate guarded sum as an
    internal choice, with right, since other guarded sums, e.g., $a.P
    + b. Q$ should certainly not be translated this way. Instead, $R$
    would be translated as something equivalent to $(P|Q) \oplus
    (P'|Q) \oplus (P|Q') \oplus (P'|Q')$, which is clearly not
    bisimilar to $R$ in general.

    We could easily include internal choice in the grammar, since
    strategies do model it, directly. We refrain from doing so for
    simplicity.

% \begin{rem}
%   Again, process terms form a terminal coalgebra for a polynomial
%   functor.  There does not appear to be a standard presentation for
%   terminal coalgebras of such polynomial functors. Although this
%   question lies beyond the scope of this paper, we mention that
%   Ad\'amek and Porst~\cite{Adamek2004257} use an elementary
%   presentation of infinite trees to characterise the terminal
%   coalgebra for polynomial endofunctors on $\Set / 1$. It might be
%   efficient to use Kock's presentation of trees as polynomial
%   endofunctors~\cite{Kock01012011} to extend the result to arbitrary
%   $\Set/X$.
% \end{rem}
\begin{defi}
  Let $\TT_\D$ be the set of process terms.
\end{defi}

\begin{exa}\label{ex:ccs:terms}
  For $\Dccs$, the obtained syntax is equivalent to
  \begin{mathpar}
    \inferrule{\ldots \ \Gam \cdot \alpha_i \vdash P_i \ldots}{
      \Gam \vdash \sum_i \alpha_i.P_i
    }
    \and
    \inferrule{
      \Gam \vdash P \\       \Gam \vdash Q
    }{
      \Gam \vdash P|Q
    }~\cdot
  \end{mathpar}%
  where
  \begin{itemize}
  \item $\Gam$ ranges over natural numbers;
  \item $\alpha ::= a \aalt \abar
    \aalt \tick \aalt \nu$ (for $a \in \Gam$);
  \item $\Gam \cdot \alpha$ denotes $(\Gam + 1)$ if $\alpha = \nu$ and
    just $\Gam$ otherwise.
  \end{itemize}
  This grammar obviously contains CCS, and we let $\theta \colon
  \ob(\ccs) \into \TT_{\Dccs}$ be the injection.
\end{exa}

\subsection{The labelled transition system for process terms}
We now define the \lts{} $\TTT_\D$. States, i.e., vertices of
the graph underlying this \lts{}, are pairs $(X, T)$ of a position $X$
and a family $T$ of process terms, indexed by the players of $X$,
i.e., $T \in \prod_{(d,x) \in \Pl (X)} (\TT_\D)_d$, where $(\TT_\D)_d$
is the set of process terms of type $d$.

To define edges, we need a lemma. For any play $u \colon X' \proto X$
and $x \colon d \to X$, recalling Notation~\ref{not:cartesian},
consider the map $$
\begin{array}[t]{ccrcl}
r^u & \colon & \sum_{(d,x) \in \Pl (X)} \Pl (D_{x,u}) & \to & \Pl (X') \\
&& ((d,x),(d',x')) & \mapsto & h_{x,u} \rond x'
\end{array}$$
 sending any $(d,x) \in \Pl (X)$ and $x' \colon d' \to D_{x,u}$ to
$d' \xto{x'} D_{x,u} \xto{h_{x,u}} X'.$

Consider also the map $i^u$ in the other direction sending any $y
\colon d' \to X'$ to the pair $((d^{y,u}, y^u), (d',
\restr{y}{y^u}))$, where $\restr{y}{y^u}$ is the (domain in $\DV$ of
the) unique $\alpha'$ making the diagram
\begin{center}
      \diagramme[diag={.4}{1.5}]{baseline=(d'.base)}{}{%
      |(d')| d' \\
      \& |(D)| D_{y^u,u} \& \& |(X')| X' \\
      {} \& \\
      |(dyu)| d^{y,u} \\
      \& |(dyui)| d^{y,u} \& \& |(X)| X %
    }{%
      (X') edge[pro,twor={u}] (X) %
      (d') edge[bend left=10,labelar={y}] (X') %
      edge[pro,twol={v^{y,u}}] (dyu) %
      edge[dashed,labelbl={\restr{y}{y^u}}] (D) %
      (l) edge[bend left=10,cell=1] node[pos=.6,anchor=north] {$\scriptstyle \alpha^{y,u}$} (r) %
      (dyu) edge[identity] (dyui) %
      edge[bend left=10] node[pos=.6,anchor=north] {$\scriptstyle y^u$} (X) %
      (dyui) edge[labeld={y^u}] (X) %
      (D) edge[pro,fore,twoleft={uy}{},tworight={uy'}{}] 
      node[pos=.4,anchor=west] {$\restr{u}{y^u}$} (dyui) %
      edge[labeld={h_{y^u,u}}] (X') %
      (uy) edge[cell=.5,labeld={\scriptstyle \alpha_{y^u,u}}] (r) %
      (l) edge[cell=.3,labelbl={\scriptstyle \alpha'},dashed] (uy') %
    }
\end{center}
commute (by~\preaxref{fibration}).  This map $i^u$ is well-defined by
uniqueness of $y^u$ and cartesianness of $\alpha_{y^u,u}$.
\begin{lem}\label{lem:ru}
The maps $i^u$ are $r^u$ are mutually inverse.
\end{lem}
\begin{proof}
  Straightforward.
\end{proof}

Let us return to the definition of our \lts{}. 
We first say that for any full quasi-move $M \colon D \proto d$, a
process term $d \vdash T$ has an $M$-transition to $(D, T')$, for $T'
\in \prod_{(d',x') \in \Pl (D)} (\TT_\D)_{d'}$, when
one of the following holds:
\begin{enumeratei}
\item \label{pt:trans:i}  $\exists M' \in \MMMFplus$, $T = M' \langle T''
  \rangle$, and, for all $(d',x') \in \Pl (D)$,
  \begin{itemize}
  \item if $v^{x',M}$ is a basic move, then $v^{x',M} \in \BsofF(M')$ and $T'_{d',x'} = T''_{v^{x',M}}$;
  \item otherwise $\length{v^{x',M}} = 0$ (hence $d' = d$), and $T'_{d',x'} = T$;
  \end{itemize}
\item $[M] \in \MMMFB$, $T = \sum_{i \in n} M_i.T_i$, $M_{i_0} =
  [M]$ for some $i_0 \in n$, and for all players $x' \colon d' \to D$
  \begin{itemize}
  \item if $v^{x',M} \in \BsofF (M)$, then 
    $T'_{d',x'} = T_{i_0}$,
  \item and otherwise ($\length{v^{x',M}} = 0$),  $T'_{d',x'} = T$;
  \end{itemize}

\item $\length{M} = 0$ 
  and for all $(d',x') \in \Pl (D)$, $T'_{d',x'} = T$ (which, again,
  makes sense by Lemma~\ref{lem:hv}).
\end{enumeratei}
We denote such a transition by $T \xot{M} (D,T')$.
\begin{rem}
  The first case \ref{pt:trans:i} allows $\BsofF(M) = \emptyset$, but
  if $\BsofF(M) \neq \emptyset$, then $[M] = M'$
  by~\axref{basic:full}. Also, let us mention that $\BsofF(M) \neq
  \emptyset$ does not imply $\length{M} = 0$ in general, although it
  does in $\Dccs$.
\end{rem}

\begin{defi}
Let $\TTT_\D$ be the graph with pairs $(X,T)$ as vertices, and as
edges $(X',T') \to (X, T)$ full quasi-moves $M \colon X' \proto X$ such
that for all $(d,x) \in \Pl (X)$, $T_{d,x} \xot{\restr{M}{x}}
(D_{x,M}, (T' \rond \cocob{h_{x,M}}))$. Here, we let $\cocob{h_{x,M}}$ denote
composition with $h_{x,M} \colon D_{x,M} \to X'$, viewed as a map $\Pl
(D_{x,M}) \to \Pl (X')$.

  $\TTT_\D$ is viewed as \anlts{} over ${\QF}$, by mapping $(X,T)
  \xot{M} (X',T')$ to $X \xproot{M} X'$.
\end{defi}

\begin{exa}\label{ex:ltsccs}
  For $\Dccs$, the obtained \lts{} differs subtly, but significantly
  from the usual \lts{} for CCS. In order to explain this clearly, let
  us introduce some notation. First, let \emph{evaluation contexts} be
  generated by the grammar
  \begin{mathpar}
    \inferrule{ }{\Gam ; x \colon n \vdash x (a_1, \ldots, a_n)}
    \and
    \inferrule{\Gam ; \Del_1 \vdash \ec_1 \\ 
      \Gam ; \Del_2 \vdash \ec_2 %
    }{%
      \Gam ; \Del_1, \Del_2 \vdash \ec_1 | \ec_2 %
    },
  \end{mathpar}
  where, in the first rule, $\forall i \in n,a_i \in \Gam$, and in the
  second $\dom (\Del_1) \cap \dom (\Del_2) = \emptyset$.  Here, $x$
  ranges over a fixed set of \emph{variables}, and $\Del, \ldots$
  range over finite maps from variables to natural numbers. Evaluation
  contexts are considered equivalent up to associativity and
  commutativity of $|$.  Positions are essentially a combinatorial,
  direct representation of such contexts.

  Leaving the details aside, states in $\TTT_{\Dccs}$ may be viewed as
  pairs $(X, T)$ of an evaluation context $X$, plus, for each $n$-ary
  variable $x (a_1, \ldots, a_n)$ in $X$, a process term over $n$ in
  the grammar of Example~\ref{ex:ccs:terms}.  Instead of separately
  writing the evaluation context and the map from its variables to
  process terms, we inline process terms between brackets in the
  context, thus avoiding variables. Moves are either put in context
  similarly, or located implicitly. E.g., for a state $(X,T)$ where
  $X$ contains two players respectively mapped by $T$ to process terms
  $P$ and $Q$, we would write $[P]|[Q]$. There is some ambiguity in
  this notation, e.g., in case some channels are absent from $P$: are
  they absent from the arity of $P$, or only unused? Since we use this
  notation mostly for clarifying examples, we will avoid such
  ambiguities.  Finally, we sometimes use brackets to denote the fact
  that some holes are filled with the given state. E.g., $X[[P]|[Q]]$
  denotes a state $X$, where a hole has been replaced by a parallel
  composition of two holes, respectively filled with $P$ and $Q$.

  Returning to our comparison of $\TTT_{\Dccs}$ and $\ccs$, of course,
  a first difference is the fact that labels may contain several
  moves, as quasi-moves only locally have length 1.

  A second difference is the presence of \emph{heating} rules for
  parallel composition and channel creation, in a sense close to the
  chemical abstract machine \citep{DBLP:conf/popl/BerryB90}.  For
  example, we have transitions $X[P | Q] \xot{\pi} X[[P]|[Q]]$.

  There is a third important difference, related to channel
  creation. For instance, we have transitions
  $$[\nu a.a.P] \xot{\nu} [a.P] \xot{\iotaneg{a}} [P].$$
  The second transition cannot occur in a closed-world setting, since
  the environment cannot know $a$.  And it does not occur in $\ccs$
  either.

  A final difference is that labels contain too much information to be
  relevant for behavioural equivalences. E.g., they contain the whole
  evaluation context in which the transition takes place, as well as
  which players are involved.

  The second difference, i.e., the presence of heating rules, is not
  really problematic, and merely forces us to use weak bisimulations
  rather than strong ones.  All other defects will be corrected below.
\end{exa}

\subsection{Translation and a first correctness result}
We conclude this section on the general theory of playgrounds by
establishing a strong, functional bisimulation from process terms to
strategies.

Mimicking \eqref{eq:traduc} (page~\pageref{eq:traduc}), our
translation from process terms to definite strategies (\emph{qua}
families over $\DI$) is defined coinductively by
\begin{equation}
\begin{array}{rcl}
  \transl{\sum_{i \in n} M_i.T_i} & = & \langle b \mapsto
      \bigoplus_{\ens{i \in n \aalt b \in \BsofF(M_i)}} \transl{T_i}
      \rangle \\
      \transl{M\langle (T_{b})_{b \in \BsofF (M)} \rangle} & = & 
      \left \langle b \mapsto \left \{
        {\begin{array}[c]{ll}
          \transl{T_{b}} & \mbox{if $b \in \BsofF(M)$} \\
          \emptyset & \mbox{otherwise}
        \end{array}}
\right . \right \rangle.
\end{array}\label{eq:transl}
\end{equation}

Let us extend the map $\translfun \colon \TT_\D \to \SS_\D$ to a map
$\translfun \colon \ob(\TTT_\D) \to \ob(\SSS_\D)$, defined by
$\transl{X,T} = (X, (\transl{T_{d,x}})_{(d,x) \in \Pl (X)})$, using
Proposition~\ref{prop:spatial}.

\begin{thm}\label{thm:bisim}
  The map $\translfun \colon \ob(\TTT_\D) \to \ob(\SSS_\D)$ is a
  functional, strong bisimulation.
\end{thm}
\begin{proof}
  The theorem follows from Proposition~\ref{prop:localtrans} and the next lemma.
\end{proof}

\begin{lem}
  For any full quasi-move $M \colon X' \proto d$, for any $T \in
  % \prod_{(d,x) \in \Pl (X)}
  (\TT_\D)_d$ and $S' \in (\SS_\D)_{X'}$, we have
  \begin{center}
    $(d,\transl{T}) \xot{M} (X',S')$ \hfil iff \hfil
    $\exists T', (T \xot{M} (X',T')) \wedge ((X',S') = \transl{X',T'})$.
  \end{center}
\end{lem}
Note the implicit typing: $T' \in \prod_{(d',x') \in \Pl (X')} (\TT_\D)_{d'}$. Also
the second condition on the right is equivalent to $\forall x'
\colon d' \to X', S' \cdot {x'} = \transl{T'_{d',x'}}$.

\begin{proof}
  If $\length{M} = 0$, then both sides are equivalent to the fact that
  for all $x' \colon d' \to X'$, $S' \cdot {x'} = \transl{T}$.

  Otherwise, we proceed by case analysis on $T$. 

  If $T = M' \with{(T''_b)_{b \in \BsofF (M')}}$, then by \axref{basic:full} both sides are
  equivalent to $\BsofF(M) \subseteq \BsofF(M')$, plus
  \begin{itemize}
  \item for all $(d',x') \in \Pl_M (X')$, $S' \cdot {x'} = \transl{T''_{v^{x',M}}}$, and
  \item for all $(d',x') \in \Pl (X') \setminus \Pl_M (X')$, $S' \cdot {x'} = \transl{T}$.
  \end{itemize}
  Indeed, for any $b \in \BsofF (M)$, $\transl{T} \cdot b =
  \transl{T''_b}$ is definite.  We thus put $T'_{d',x'} =
  T''_{v^{x',M}}$ in the first case and $T'_{d',x'} = T$ in the second
  case.

  If $T = \sum_{i \in n} M_i.T_i$, then both sides are equivalent to 
  the existence of $i_0 \in n$ such that 
  $M_{i_0} = [M]$ and
  \begin{itemize}
  \item for the unique $(d',x') \in \Pl_M (X')$, $S' \cdot {x'} = \transl{T_{i_0}}$, and
  \item for all $(d',x') \in \Pl (X') \setminus \Pl_M (X')$, $S' \cdot {x'} = \transl{T}$.
  \end{itemize}
  This uses~\axref{basic:full}, since the left-hand side unfolds to
  the existence of $x' \colon d' \to X'$ such that $v^{x',M} \in
  \BsofF [M]$ and $\transl{T} \cdot v^{x',M} \neq \emptyset$, i.e., $v^{x',M}
  \in \BsofF (M_{i_0})$ for some $i_0 \in n$, by definition of
  $\transl{T}$. This entails in particular $[M] = M_{i_0}$
  by~\axref{basic:full}.
\end{proof}

% VIRE
% To conclude this section, we state a simple corollary. 
% \begin{defi}
%   Let $\QFI$ denote the graph with morphisms $k \colon I \to X$ as
%   vertices, where $I$ is an interface, and whose edges $k \to h$, for
%   $k \colon J \to Y$, are full quasi-moves $M \colon Y \proto X$.
% \end{defi}
% Let $U \colon \QFI \to \QF$ be the obvious forgetful morphism.
% Let $\SSSDI$ and $\TTTDI$ be the pullbacks of $\SSS_\D$ and $\TTT_\D$ along $U$.
% \begin{cor}\label{cor:wbisim1}
%   The map $\ob(\SSSDI) \to \ob(\TTTDI)$ obtained by pullback, which we also denote by $\translfun$,
%   is a strong bisimulation.
% \end{cor}
% \begin{proof}
%   By Proposition~\ref{prop:change of base}.
% \end{proof}

\section{Graphs and fair morphisms}\label{sec:graphs}
In this section, we derive our main result.  For this, we develop a
notion of \emph{graph with complementarity}, which aims at being a
theory of \ltss{} over which fair testing makes sense.  Although the
theory would apply with any predicate $\bot$ compatible with
$\wbisimsierp$ equivalence classes (see below), the question of
whether such a generalisation would have useful applications is
deferred for now.

For any graph with complementarity $A$ and relation $R \colon G \modto
H$ over $A$, we exhibit sufficient conditions for $R$ to be
\emph{fair}, i.e., to preserve and reflect fair testing
equivalence. We then relate this theory to our semantics, and show
that it entails our main result.  For now, this section lies outside
the scope of playground theory. Some aspects of it could be formalised
there, but we leave the complete formalisation for further
work. Because the only playground involved is $\Dccs$, we often omit
sub or superscripts, e.g., in $\D$, $\SSS_{\D}$ (even just $\SSS$),
etc.

Before we start, let us define $\ccsW$ to be the set of
\emph{closed-world quasi-moves}, i.e., vertical morphisms in $\D$
which either are closed-world moves
(Definition~\ref{def:cw:successful}) or have length 0.  
Please note: quasi-moves must locally restrict to plays of length $\leq 1$, whereas
closed-world quasi-moves have length $\leq 1$ globally.
Let $\DW$ be
the subbicategory of $\Dv$ generated by $\ccsW$, and let $\Sierp$ be the
free reflexive graph on an endo-edge $\tick$.  Finally, let
$\labelDccsfunctor \colon \DW \to \freecat{\Sierp}$ be the pseudo functor
determined by the mapping $\labelDccs \colon \ccsW \to \Sierp$ sending all closed-world
quasi-moves to $\id$ except $\tick$ moves, which are sent to $\tick$.

\subsection{Graphs with complementarity}
A \emph{relation} $A \modto B$ between two reflexive graphs $A$ and
$B$ is a subgraph $R \into A \times B$. Such a relation $R$ is
\emph{total} when, for all vertices, resp.\ edges, $x \in A$, there
exists a vertex, resp.\ an edge $y \in B$, such that $(x,y) \in R$. It
is \emph{partially functional} if there is at most one such $y$.  It is
\emph{functional} when it is total and partially functional.  The
\emph{domain} of $R$ is the subgraph of $A$ consisting of vertices and
edges related to something in $B$.
\begin{defi}
  A \emph{graph with complementarity} is a reflexive graph $A$,
  equipped with a subgraph $\aW$, a relation $\compata \colon A^2
  \modto \aW$, and a map $\labela \colon \aW \to \Sierp$, such that
  the composite $A^2 \modto \aW \to \Sierp$ is partially functional
  and symmetric.
\end{defi}
We let $\acoh = \dom(\compata)$ and write $a \coh a'$ for $(a,a') \in
\acoh$. We further denote the map $\acoh \into A^2 \modto \aW \to
\Sierp$ by $(a,b) \mapsto (a \dpara b)$, and deem edges in $\aW$
\emph{closed-world}.

\begin{rem}
  $\acoh$ has to be symmetric as the domain of a symmetric relation.
 % a relation $A \modto A$,
 %  as if $a \coh a'$, then there exists $w \in \aW$ such that $(a,a')
 %  \compata w$, hence $(a \dpara a') = \labela(w)$, so by symmetry of
 %  $\dpara$, $(a' \dpara a) = \labela(w)$, hence there exists $w'$ such
 %  that $(a',a) \compata w'$ and $\labela(w') = \labela(w)$.
\end{rem}

\begin{defi}
  A morphism of graphs with complementarity is a morphism $f \colon A
  \to B$ of reflexive graphs such that
  \begin{mathpar}
    f(\aW) \subseteq \bW
    % Redundant!!!  \and *) f^2(\acoh) \subseteq \bcoh 
    \and
    \labelb \rond \fW = \labela \and ((a_1,a_2) \compata a_3)
    \Rightarrow ((f(a_1),f(a_2)) \compata f(a_3)),
  \end{mathpar}
  where $\fW \colon \aW \to \bW$ is the restriction of $f$.
\end{defi}

\begin{prop}
  Graphs with complementarity and morphisms between them form a category $\GCompl$.
\end{prop}

We now introduce the graph $\LLL$, which as announced in the
introduction will serve as a base for making $\SSS_{\Dccs}$ and
$\TTT_{\Dccs}$ into graphs with complementarity. It is an
\emph{interfaced} variant of $\QF$, hence its name.
\begin{exa}
  Let $\LLL$ be the graph with as vertices all horizontal morphisms $h
  \colon I \to X$ from some interface to some position, and whose
  edges $k \to h$ are given by diagrams
\begin{equation}
  \Diag{%
    \twocellbr{m-2-1}{m-1-1}{m-1-2}{\alpha} %
  }{%
    I \& Y \\
    I \& X %
  }{%
    (m-1-1) edge[labelu={k}] (m-1-2) %
    edge[pro,identity] (m-2-1) %
    (m-2-1) edge[labeld={h}] (m-2-2) %
    (m-1-2) edge[pro,labelr={M}] (m-2-2) %
  }\label{eq:limove}
\end{equation}
in $\DH$, where $M$ is either a full move or an identity, such that if
$M$ is an input or an output, then the corresponding channel is in the
image of $I$.  $\LLL$ forms a reflexive graph with identities given by
the case where $M = \idv$, which forms a graph with complementarity as
follows.

Let $(\LLL)^\W$ consist of all closed-world quasi-moves in $\QFI$. For any $h
\colon I \to X$, $k \colon J \to Y$, and $c \colon K \to Z$, let
$(h,k) \compatof{\LLL} c$ iff $I = J = K$, $Z = h +_I k$, and $c$ is
the corresponding map $I \to Z$. On edges, for any $M_h \colon h' \to
h$, $M_k \colon k' \to k$, and $M_c \colon c' \to c$, let $(M_h,M_k)
\compatof{\LLL} M_c$ iff there exists a diagram
      \begin{equation}\label{eq:trans}
      \Diag{%
        \pbk{X}{Z}{Y} %
        \pbkk{X}{X'}{Z'} %
        \pbkk{Y}{Y'}{Z'} %
        \pbkk[2em]{I}{I'}{Y'} %
        \pbkk[2em]{I}{I'}{X'} %
        \pullback[2em]{X'}{I'}{Y'}{draw,-,fore} %
        \path[->,draw] %
        (X') edge[fore] (Z') %
        (I') edge[labelblat={h'}{.8},fore] (X') %
        ; %
      }{%
       |(I')| I \&  \& \&|(Y')| Y' \\
        \&|(X')| X' \& \& \&|(Z')| Z' \\
        \\
       |(I)| I \&  \& \&|(Y)| Y \\
        \&|(X)| X \& \& \&|(Z)| Z, %
      }{%
        (I') edge (Y') %
        edge[pro,identity] (I)
        (Y') edge (Z') %
        (I) edge (Y) %
        edge[labelbl={h}] (X) %
        (Y) edge (Z) %
        (X) edge (Z) %
        %
        % (I') edge (I) %
        (Y') edge[pro,labellat={M_k}{0.7}] (Y) %
        (X') edge[pro,fore,labell={M_h}] (X) %
        (Z') edge[pro,labell={M_c}] (Z) %
      }
    \end{equation}
    where $M_c$ is a closed-world quasi-move and double cells with a
    `double pullback' mark are cartesian, as below
    Axiom~\axref{fibration} (page~\pageref{fibration}). (One
    easily shows that the upper square is also a pushout.)  Then
    $(\LLL)^\coh$, consists of all pairs $(M_h,M_k)$ for which there
    exists a diagram of the shape~\eqref{eq:trans}.

    Let $\labelLLL$ be the composite $(\LLL)^\W \into \ccsW
    \xto{\labelDccs} \Sierp$. It thus maps tick moves to $\tick$ and
    all other closed-world moves to $\id$. The composite $(\LLL)^2
    \xmodto{\compatof{\LLL}} (\LLL)^\W \to \Sierp$ is indeed partially
    functional and symmetric.

    There is an obvious morphism $\chi \colon \LLL \to \QF$ of reflexive graphs.
\end{exa}

\begin{exa}\label{exa:Axi}
  Recall the alphabet $\A$ for CCS. It also forms a graph with complementarity, as follows.
Let $\AW$ consist of all vertices and of all $\tick$ and $\id$ edges.
Let $\Acoh$ consist, on vertices, of the diagonal, i.e., all pairs $(n,n)$.
On edges, let $e \coh e'$ when $\dom(e) \coh \dom(e')$ and:
\begin{itemize}
\item one of $e$ and $e'$ is in $\AW$, the other being an identity,
\item or one of $e$ and $e'$ is an input on some $i \in \dom(e)$,
  the other being an output on $i$.
\end{itemize}
Define now our relation $\compatA$ to be the graph of the map sending
all coherent pairs $e \coh e'$ to $\id$, except when one is a $\tick$,
in which case the pair is sent to $\tick \colon n \to n$. The axioms
are easily satisfied.

Let $\xi \colon \LLL \to \A$ map any vertex $h \colon I \to X$ to 
$n = I(\star)$, and any edge~\eqref{eq:limove} to
\begin{itemize}
\item $\id_n$ if $M$ is an identity, a synchronisation, a fork, or a channel creation,
\item $\tickn$ if $M$ is a tick move,
\item $i$ if $M$ is an input on $h_\star(i)$,
\item $\overline{i}$ if $M$ is an output on $h_\star(i)$.
\end{itemize}
This map $\xi$ is a morphism of graphs with complementarity.
\end{exa}

We have the following general way of constructing graphs with
complementarity.  For any graph with complementarity $A$ and morphism
of reflexive graphs $p \colon G \to A$, consider the following
candidate complementarity structure on $G$. 

  Let $\GW = G \times_A \aW$ denote the pullback
  \begin{equation}
      \Diag{%
        \pbk{m-2-1}{m-1-1}{m-1-2} %
      }{%
        \GW \& \aW \\
        G \& A. %
      }{%
        (m-1-1) edge[labelu={p^\W}] (m-1-2) %
        edge[labell={}] (m-2-1) %
        (m-2-1) edge[labeld={p}] (m-2-2) %
        (m-1-2) edge[into,labelr={}] (m-2-2) %
      }\label{eq:GW}
    \end{equation}
Further, let $\labelG$ be the composite $\GW \xto{p^\W}
  \aW \xto{\labela} \Sierp$, and let $(x,y) \compatG z$ iff $(p(x),p(y))
  \compata p(z)$ (for both vertices and edges). In other words,
  $\compatG$ is the relational composite 
    $$G^2 \xto{p^2} A^2 \xmodto{\compata} \aW \xot{p^\W} G^\W,$$
    where the backwards $p^\W$ arrow denotes the converse of the graph
    of $p^\W$.

    \begin{prop}\label{prop:compatpR}
      For any subrelation $R \subseteq \compatG$, if $R$
      is symmetric, then $(G,\GW,R,\labelG)$ forms a graph
      with complementarity, and $p$ is a morphism of graphs with
      complementarity to $A$.
    \end{prop}
    \begin{proof}
    By standard relational algebra, the composite relation
    $$G^2 \xmodto{\compatG} G^\W \xto{p^\W} \aW,$$
    which is equal to
    $$G^2 \xto{p^2} A^2 \xmodto{\compata} \aW \xot{p^\W} G^\W \xto{p^\W} \aW,$$
    is included in
    $$G^2 \xto{p^2} A^2 \xmodto{\compata} \aW.$$

    Composing with $\labela$, we obtain that $(\labelG \rond \compatG)
    \subseteq (\labela \rond \compata \rond p^2)$, which is
    straightforwardly symmetric and partially functional.  A
    subrelation of a partially functional relation is automatically
    partially functional, so $\labelG \rond R$ is partially functional.
    It is symmetric because $R$ is, hence the result.
    \end{proof}

\begin{exa}\label{ex:lgraphs}
  $\ccs$ forms a graph with complementarity over $\A$ by the last
  proposition, taking $R$ to relate
  \begin{itemize}
  \item all pairs $(n \vdash P, n \vdash Q)$ to $n \vdash P \para Q$
    on vertices,
  \item any transitions $(\Gam \vdash P_1) \xot{\alpha} (\Gam \vdash
    P'_1)$ and $(\Gam \vdash P_2) \xot{\id} (\Gam \vdash P_2)$ with
    $(\Gam \vdash P_1 \para P_2) \xot{\alpha} (\Gam \vdash P'_1 \para
    P_2)$, and symmetrically,
  \item and any two transitions $(\Gam \vdash P_1) \xot{\alpha} (\Gam
    \vdash P'_1)$ and $(\Gam \vdash P_2) \xot{\overline{\alpha}} (\Gam
    \vdash P'_2)$ with $(\Gam \vdash P_1 \para P_2) \xot{\id} (\Gam
    \vdash P'_1 \para P'_2)$.
  \end{itemize}
\end{exa}

    \begin{prop}\label{prop:constructcompl} Suppose given a choice, for all $x, y \in G$ and $a
      \in A$ such that $(p(x),p(y)) \compata a$, of a vertex $[x,y]_a
      \in G$ such that $p([x,y]_a) = a$, satisfying the following
      condition: for all edges $e_x \colon x' \to x$ and $e_y \colon
      y' \to y$ in $G$, and $e_a \colon a' \to a$ in $A$, if
      $(p(e_x),p(e_y)) \compata e_a$, then there exists a
      $[e_x,e_y]_{e_a} \colon [x',y']_{a'} \to [x,y]_a$ such that
      $p([e_x,e_y]_{e_a}) = e_a$.

      Then, $(G,\GW,\compatG,\labelG)$ forms a graph with
      complementarity, and $p$ is a morphism of graphs with
      complementarity.
  \end{prop}
  \begin{proof}
    Recalling the beginning of the proof of Proposition~\ref{prop:compatpR},
    the hypothesis implies that the inclusion
    $$(G^2 \xmodto{\compatG} G^\W \xto{p^\W} \aW) \subseteq (G^2 \xto{p^2} A^2 \xmodto{\compata} \aW)$$
    is actually an equality.
    
    Composing with $\labela$, we obtain that $\labelG \rond \compatG
    = \labela \rond \compata \rond p^2$, which is straightforwardly
    symmetric and partially functional. The morphism $p$ is a morphism
    of graphs with complementarity by construction.
\end{proof}
    % In particular, we have $\Gcoh = G^2 \times_{A^2} \acoh$, i.e., we have a pullback square
    % \begin{equation}
    %   \Diag{%
    %     \pbk{m-2-1}{m-1-1}{m-1-2} %
    %   }{%
    %     \Gcoh \& \acoh \\
    %     G^2 \& A^2. %
    %   }{%
    %     (m-1-1) edge[labelu={}] (m-1-2) %
    %     edge[labell={}] (m-2-1) %
    %     (m-2-1) edge[labeld={p^2}] (m-2-2) %
    %     (m-1-2) edge[into,labelr={}] (m-2-2) %
    %   }
    %   \label{eq:Gcoh}
    % \end{equation}

\begin{defi}
  Let $\SSSL = \cob{\chi}(\SSS)$ and $\TTTL = \cob{\chi} (\TTT)$ be
  the pullbacks of $\SSS \to \QF$ and $\TTT \to \QF$ along $\chi \colon \LLL \to
  \QF$.   
\end{defi}

\begin{exa}
  $\SSSL$ and $\TTTL$ form graphs with complementarity over $\QFI$ by
  Proposition~\ref{prop:constructcompl}.  The canonical relation
  $\compatof{\ccs}$ does not satisfy the condition of
  Proposition~\ref{prop:constructcompl}, however. Indeed, e.g., any
  non-silent transition $(\Gam \vdash P_1) \xot{\alpha} (\Gam \vdash
  P'_1)$ and silent but non-identity transition $(\Gam \vdash P_2)
  \xot{\id} (\Gam \vdash P'_2)$ are not coherent in $\ccs$, although
  their images under the projection to $\A$ are so.  (Amalgamating two
  such transitions in $\ccs$ requires a path of length 2, as will be
  used below.) What saves $\SSSL$ and $\TTTL$ from this issue is that
  projecting to $\QFI$ does not hide away, e.g., synchronisations.
\end{exa}

\subsection{Modular graphs and fair testing equivalence}
We now introduce the notion of \emph{modular} graph, which is
appropriate for defining fair testing. We could actually introduce
fair testing for arbitrary graphs with complementarity, but the extra
generality would make little sense.

For any graph with complementarity $G$, $\Gcoh$ forms \anlts{} over
$\Sierp$, through $\Gcoh \xto{\dpara} \Sierp$.
\begin{defi}
  $G$ is \emph{modular} iff for all $(x,y) \compatG z$ we have both:
  \begin{enumerate}
  \item for all $e \colon z' \to z$, there exists $e_x \colon x' \to x$ and
    $e_y \colon y' \to y$ such that $(e_x,e_y) \compatG e$; and\label{modularity:i}
  \item for all $e_x \colon x' \to x$ and $e_y \colon y' \to y$ such that $e_x \coh e_y$
    there exists $e \colon z' \to z$ such that $(e_x,e_y) \compatG e$. \label{modularity:ii}
  \end{enumerate}
\end{defi}
\begin{rem}
  The second condition is almost redundant: in any graph with
  complementarity $G$, there exists $e'$ such that $(e_x,e_y) \compatG
  e'$, but the target of $e'$ may be any $u$ such that $(x,y) \compatG
  u$; it does not have to be $z$.
\end{rem}

\begin{prop}\label{prop:modbis}
   $G$ is \emph{modular} iff ${\compatG}$ is a strong bisimulation
  over $\Sierp$.
\end{prop}
We here implicitly view $\compatG$ as a relation $\Gcoh \modto \GW$.
\begin{proof} 
Since $\compatG$ is a relation over $\Sierp$, it is enough to prove that both 
projections are graph fibrations, which is directly equivalent to modularity.
% We show that modularity implies $\compatG$ is a
%   bisimulation, the converse being just as easy.  If $x' \xto{e_x} x$,
%   $y' \xto{e_y} y$, and $(x,y) \compatG z$, with $e_x \dpara e_y =
%   \sigma$, then by~\eqref{modularity:ii} there exists $e \colon z' \to
%   z$ such that $(e_x,e_y) \compatG e$ and $\labelG(e) = \sigma$ (by
%   definition of $\dpara$). Hence, $z'
%   \mathrel{{}_{\Sierp}\!\xot{\sigma}} z$ and $(x',y') \compatG z'$ as
%   desired.
%
%   Conversely, if $(x,y) \compatG z$ and $e \colon z' \to z$ in $G^\W$,
%   then by~\eqref{modularity:i} there exist $e_x \colon x' \to x$ and
%   $e_y \colon y' \to y$ such that $(e_x,e_y) \compatG e$.
\end{proof}

\begin{exa}
  $\SSSL$ and $\TTTL$, as well as $\ccs$, are modular.
\end{exa}

We now define fair testing in any modular graph, and compare with both
semantic fair testing equivalence ($\faireq$) for strategies and
standard fair testing equivalence ($\faireqs$) for CCS processes.
Recall that $\bisimsierp$ denotes strong bisimilarity over $\Sierp$.
\begin{lem}
  For any modular graph with complementarity $G$ and $x,y,z,t \in
  G$, if $(x,y) \compatG z$ and $(x,y) \compatG t$, then $z
  \bisimsierp t$.
\end{lem}
\begin{proof}
  We have $z \bisimsierp (x,y) \bisimsierp t$.
\end{proof}

Any modular graph may be equipped with a choice of $z$ such that
$(x,y) \compatG z$, for all $x \coh y$. We denote such a choice by
$[x,y]$. By the lemma, the choice of $z$ does not matter as long as we
only consider properties invariant under $\bisimsierp$.  Here, we only
need the standard predicate for fair testing.

\begin{defi}
  For any reflexive graph $G$ over $\Sierp$, let $\bot^G$ denote the set of all
  $x \in G$ such that for all $x \Leftarrow x'$ there exists $x'
  \xLeftarrow{\tick} x''$.
\end{defi}
When $G$ is a graph with complementarity,
we often denote $\bot^{\GW}$ by $\botG$. There is no confusion because 
$G$ is not even a graph over $\Sierp$ in general.

In any modular graph with complementarity $G$, let, for any $x \in G$,
$\testable{x} = \ens{y \aalt x \coh y}$, and let $\eqtestable{x}{y}$
iff $\testable{x} = \testable{y}$.
\begin{defi}
  For any $x,y \in G$, let $x \faireqG y$ iff $\eqtestable{x}{y}$ and
  for all $z \in \testable{x}$, $[x,z] \in \botG$ iff $[y,z] \in
  \botG$.
\end{defi}

We may at last define fair relations:
\begin{defi}\label{def:fairrel}
  For all modular graphs with complementarity $G$ and $H$, and full
  relations $R \colon G \modto H$, let $R$ \emph{preserve fair testing
    equivalence} when, for all $x \relR x'$ and $y \relR y'$, $(x
  \faireqof{G} y)$ implies $(x' \faireqof{H} y')$.  $R$ \emph{reflects
    fair testing equivalence} when the converse implication holds.
  $R$ is \emph{fair} when it preserves and reflects fair testing
  equivalence.
\end{defi}

Modularity enables a first, easy characterisation of fair testing.
\begin{prop}
  If $G$ is modular, then for any $x \coh y$, $[x,y] \in \bot^G$ iff
  $(x,y) \in \bot^{\Gcoh}$.
\end{prop}
\begin{proof}
A direct consequence of Proposition~\ref{prop:modbis}.
\end{proof}

We now prove that the general definition of fair testing equivalence
instantiates correctly for $\SSSL$ and $\ccs$.  First, we easily have 
\begin{prop}\label{prop:fairccs}
  For any two CCS processes $P$ and $Q$ over $n$, $P \faireqs Q$ iff
  $P \faireqof{\ccs} Q$.
\end{prop}
\begin{proof}
  Straightforward.
\end{proof}

We now wish to compare $\bot^{\SSSL}$, as defined in this section, and
the semantic $\bbot$. As an intermediate step, we consider the
following, bare $\barebot$, which lives over $\QF$, but is defined in
terms of \ltss{} (as opposed to successful states of strategies).  Let
$\SSS^\W$ be the restriction of $\SSS$ to closed-world transitions,
i.e., the pullback of $\SSS \to \QF$ along the inclusion $\ccsW \into
\QF$; this is \anlts{} over $\Sierp$ via $\labelDccs$.  Let $\barebot
= \bot^{\SSS^\W}$ denote the set of pairs $(X,S) \in \SSS$ such that
for all $(X,S) \Leftarrow (X',S')$ there exists $(X',S')
\xLeftarrow{\tick} (X'',S'')$.
\begin{lem}\label{lem:bbot:bot}
  For all $(X,S) \in \SSS$,  $\exta{X}{S} \in \bbot_X$ iff $(X,S) \in \barebot$.
\end{lem}

This essentially amounts to checking that the notions of closed-world,
successful, and unsuccessful play
(Definition~\ref{def:cw:successful}), correspond with closed-world,
successful, and unsuccessful transition sequences.  The former are
defined in terms of plays and moves therein, while the latter rest
upon the map $\labelDccs \colon \ccsW \to \Sierp$.

  We first observe:
\begin{lem}\label{lem:closedworld}
  For any two closed-world plays $W,W'$ over $X$, and $\alpha \colon W
  \to W'$ in $\DH(X)$, $\alpha$ is an isomorphism, and it is unique.
\end{lem}

\begin{proof}[Proof of Lemma~\ref{lem:bbot:bot}]
  Let $S \in \SSX$ and assume $\exta{X}{S} \in
  \bbot_X$. Let $S \Leftarrow S'$ (over $\Sierp$).  This
  means that there exists a path $p$
$$X = X_0 \xot{M_1} X_1 \xot{M_2} \ldots X_n = X',$$
such that, omitting positions, 
$$S = S_0 \xot{M_1} S_1 \xot{M_2} \ldots S_n = S',$$
and $p$ is mapped by $\labelDccs^\star$ to the path of length $n$ consisting
only of $\id$ edges. This implies by induction the existence of
$\state \in \exta{X}{S}(W)$, where $W = M_1 \vrond \ldots \vrond M_n$
is closed-world and unsuccessful, such that $S' = \restr{(S \cdot
  W)}{\psi(\state)}$.  Because $\exta{X}{S} \in \bbot_X$, there exists a
successful, closed-world play $W'$, a morphism $f \colon W \to W'$ in $\E(X)$, and $\state' \in
\exta{X}{S}(W')$ such that $\state' \cdot f = \state$.  By
Lemma~\ref{lem:closedworld}, $W'$ is isomorphic to an extension of $W$ with
closed-world moves, say $W' \iso W \vrond M_{n+1} \vrond \ldots \vrond
M_{n+m}$. By induction on $m$, we obtain a path
$$S' = S_n \xot{M_{n+1}} S_{n+1} \xot{M_{n+2}} \ldots S_{n+m},$$ where $S_{n+m} = \restr{(S \cdot
  W')}{\psi(\state')}$. Because $W'$ is successful, there exists $i \in m$
such that $\labelDccs(M_{n+i}) = \tick$, hence $S' \xLeftarrow{\tick}
S_{n+i}$. Thus, $(X,S) \in \barebot$.

Conversely, assume $(X,S) \in \barebot$. Let $W$ be an
unsuccessful, closed-world play over $X$ and $\state \in
\exta{X}{S}(W)$. Picking a decomposition $W = M_1 \vrond \ldots \vrond
M_n$ of $W$, we obtain a path $p$
$$S = S_0 \xot{M_1} S_1 \ldots \xot{M_n} S_n = S'$$
in $\SSS$ such that $S' = \restr{(S \cdot W)}{\psi(\state)}$, which 
yields $S \Leftarrow S'$. Because $(X,S) \in \barebot$, there exists
$S' \Leftarrow S'' \xot{\tick} S'''$, with underlying path
$$S' = S_n \xot{M_{n+1}} S_{n+1} \ldots \xot{M_{n+m}} S_{n+m} = S'' \xot{M_{n+m+1}} S'''$$
in $\SSS^\W$, such
that $\labelDccs(M_{n+i}) = \id$ for all $i \in m$ and $\labelDccs(M_{n+m+1}) =
\tick$.  But by definition this means that $S''' = \restr{(S' \cdot
  W')}{\psi(\state')}$ for some
$$\state' \in \exta{}{S'}(W') = \exta{}{\restr{(S \cdot W)}{\psi(\state)}}(W') 
= \ens{\state'' \in \exta{}{S}(W \vrond W') \aalt \state''
  \cdot f = \state},$$ where $W' = M_{n+1} \vrond \ldots \vrond
M_{n+m+1}$ and $f \colon W \to (W \vrond W')$ is the extension.
By construction, $\state' \cdot f = \state$. Hence, $\exta{X}{S} \in \bbot_X$.  
\end{proof}

We furthermore have:
\begin{lem}\label{lem:barebat:botL}
  For any vertex $h \colon I \to X$ of $\LLL$ and
  $S \in \SSX$, $(I,h,S) \in \bot^{\SSSL}$ iff $(X,S) \in \barebot$.
\end{lem}
\begin{proof}
  The map $\chi^\W \colon \LLL^\W \to \QF^\W$ is a strong, functional bisimulation,
  because for any $h \colon I \to X$ and closed-world move $M \colon Y
  \proto X$, there exists a diagram~\eqref{eq:limove}.  Thus, the
  projection $(\SSSL)^\W \to \SSS^\W$ is a strong, functional bisimulation by
  Proposition~\ref{prop:change of base}.
\end{proof}
\begin{rem}
  Interfaces are pretty irrelevant here, and indeed we could have
  decreed that closed-world moves only relate vertices with empty
  interfaces in $\LLL$. This is unnecessary here, though, so we stick
  to the simpler definition, but it will be crucial for the
  $\pi$-calculus.
\end{rem}

This entails:
\begin{cor}\label{cor:SSSLI}
  For any $h \colon I \to X$, $h' \colon I \to X'$, $S \in \SSX$, and
  $S' \in \SS_{X'}$, $(I,h,S) \faireq (I,h',S')$ iff $(I,h,S)
  \faireqof{\SSSL} (I,h',S')$.
\end{cor}
\proof We have
  \begin{center}
    $(I,h,S) \faireq (I,h',S')$ 
    
      $\Updownarrow$\makebox[0pt][l]{(by definition)}

      $\forall Y, k \colon I \to Y, T \in \SS_Y, (\exta{}{[S,T]} \in \bbot_{X +_I Y} \Leftrightarrow
      \exta{}{[S',T]} \in \bbot_{X' +_I Y})$

      $\Updownarrow$\makebox[0pt][l]{(by Lemma~\ref{lem:bbot:bot})}

      $\forall Y, k \colon I \to Y, T \in \SS_Y, ((X +_I Y, {[S,T]}) \in \barebot \Leftrightarrow
      (X' +_I Y, {[S',T]}) \in \barebot)$

      $\Updownarrow$\makebox[0pt][l]{(by Lemma~\ref{lem:barebat:botL})}

      $\forall Y, k \colon I \to Y, T \in \SS_Y, ((I \to X +_I Y, {[S,T]}) \in \bot^{\SSSL} \Leftrightarrow
      (I \to X' +_I Y, {[S',T]}) \in \bot^{\SSSL})$

      $\Updownarrow$\makebox[0pt][l]{(by definition)}

      $(I,h,S) \faireqof{\SSSL} (I,h',S'),$
\end{center}
which concludes the proof.
\qed

\subsection{Adequacy}
Until now, our study of graphs with complementarity and fair testing
therein is intrinsic, i.e., fair testing equivalence in a modular
graph with complementarity $G$ does not depend on any alphabet.  We
now address the question of what an alphabet should be, for $G$. The
main idea is that such an alphabet $A$ should be a graph with
complementarity, and that viewing it as an alphabet for $G$ is the
same as providing a morphism $p \colon G \to A$ in $\GCompl$, satisfying
a certain condition called \emph{adequacy}. To understand the role of
this condition, one should realise that edges in $G$ may be much too
fine a tool for checking fair testing equivalence. E.g., in $\SSSL$,
they include information about which players played which move.  Thus,
although it is true that weak bisimilarity implies fair testing
equivalence, this property is essentially useless for fair testing,
because too few strategies are weakly bisimilar.  Any morphism $p
\colon \SSSL \to A$ induces an \emph{a priori} coarser version of fair
testing for $\SSSL$, where one only looks at labels in $A$.
\emph{Adequacy} is a sufficient condition for this latter version to
coincide with the original. This will in particular entail that weak
bisimilarity over $A$ is finer than fair testing equivalence.

Adequacy relies on the following:
\begin{defi}
  Consider, for any $p \colon G \to A$ and $q \colon H \to A$ the pullback
    \begin{center}
      \Diag{%
        \pbk{m-2-1}{m-1-1}{m-1-2} %
      }{%
        \combinea{G}{H} \& \acoh \\
        G \times H \& A^2. %
      }{%
        (m-1-1) edge[labelu={}] (m-1-2) %
        edge[labell={}] (m-2-1) %
        (m-2-1) edge[labeld={p \times q}] (m-2-2) %
        (m-1-2) edge[into,labelr={}] (m-2-2) %
      }
    \end{center}
    We call $\combinea{G}{H}$ the \emph{blind composition} of $G$ and
    $H$ over $A$, viewed as \anlts{} over $\Sierp$ via
    $\combinea{G}{H} \to \acoh \to \Sierp$.
\end{defi}

Recall from Section~\ref{subsec:prelim:lts} that $\wbisim_A$ denotes
weak bisimilarity for reflexive graphs over $A$.
\begin{defi}
  Let $p \colon G \to A$ be a morphism of graphs with
  complementarity. We say that $p$ is \emph{adequate} iff
  \begin{itemize}
  \item the graph of $\ob \Gcoh \into \ob(\combinea{G}{G})$ is
    included in $\wbisimsierp$, and
  \item for all $x,y \in \ob(G)$, $x \coh y$ iff $p(x) \coh p(y)$.
  \end{itemize}

\end{defi}
Concretely, any transition $(e_1,e_2) \in \Gcoh$ is matched, without
any hypothesis on $G$, by $(e_1,e_2)$ itself. Conversely, having a
transition $(x_1,x_2) \xot{e_1,e_2} (x'_1,x'_2)$ in $\combinea{G}{G}$
means that $p(e_1) \dpara p(e_2) = \sigma$. Adequacy demands that
there exists a path $(r_1,r_2) \colon (x_1,x_2) \ot^\star
(x''_1,x''_2)$ in $\Gcoh$, such that $\idfree{r_1 \dpara^\star
  r_2} = \idfree{(\sigma)}$, and $(x''_1,x''_2) \wbisimsierp
(x'_1,x'_2)$, where the left-hand side is in $\Gcoh$ and the
right-hand side is in $\combinea{G}{G}$.

Recall the map $\xi \colon \LLL \to \A$ from Example~\ref{exa:Axi}.
Via this map, $\SSSL$ and $\TTTL$ form \ltss{} and even graphs with
complementarity over $\A$.
\begin{prop}
  The maps from $\ccs$, $\SSSL$, and $\TTTL$ to $\A$ are adequate.
\end{prop}
\begin{proof}
  For all three graphs $p \colon G \to \A$ over $\A$, both $\Gcoh$ and
  $\combinea{G}{G}$ form graphs over $\A^\W$, because $\compatA \colon
  \A^2 \modto \A^\W$ is actually partially functional.  In each case,
  the graph of $\ob \Gcoh \into \ob(\combinea{G}{G})$ is a weak
  bisimulation over $\AW$, because for all $e$ and $e'$ in $G$, if
  $p(e) \coh p(e')$, then either $e \coh e'$, or both interleavings
  are coherent, i.e., $(e,\id) \coh (\id,e')$ and $(\id,e) \coh (e',
  \id)$, pointwise. (Here, e.g., $(e,\id)$ denotes the \emph{path}
  $\cdot \xot{e} \cdot \xot{\id} \cdot$.)
\end{proof}
The only subtle point is that this only holds thanks to the
restrictions put on edges of $\LLL$. E.g., consider the graph the
graph $\LLL_{-}$ with the same vertices as $\LLL$, and edges $(I
\xto{k} Y) \to (I \xto{h} X)$ given just as for $\LLL$, except that we
do not require existence of a diagram~\eqref{eq:limove}.  Pullback
yields a graph $\SSS_{-}$ over $\LLL_{-}$.  Extending $\xi$ to $\xi'
\colon \LLL_{-} \to \A$ in the obvious way, we obtain a graph over
$\A$.  Consider now the moves $\iotapos{2,1}, \iotaneg{2,1} \colon [2]
\proto [2]$, let $I = 2 \cdot \star$, and let $f$ be one of the two
embeddings $I \to [2]$, say the one which is an inclusion at $\star$,
$f'$ being the other.  Recalling labels in $\A$ from
Definition~\ref{def:A}, we have edges $(I, f,[2]) \xot{\onebar} (I,
f,[2])$ and $(I, f,[2]) \xot{1} (I, f',[2])$, and $(I, f,[2]) \coh (I,
f,[2])$. However, the two edges are not coherent, because any attempt
to construct a diagram~\eqref{eq:trans} (with here $h = h' = k = f$,
and $k' = f'$) fails (even if we forget about the vertical
identity). This is the very reason we use $\LLL$ instead of
$\LLL_{-}$.

We have the following two easy properties of blind composition.
\begin{prop}\label{prop:combine:adeq}
  For any modular $G$, adequate $p \colon G \to A$, and $x \coh y$ in $G$, we have 
  $[x,y] \in \botG$ iff $(x,y) \in \bot^{\combinea{G}{G}}$.
\end{prop}
\begin{proof}
We have  $[x,y] \wbisimsierp ((x,y) \in \Gcoh) \wbisimsierp ((x,y) \in \combinea{G}{G})$.
\end{proof}

\begin{prop}
  For any $H$ over $A$, modular $G$, adequate $p \colon G \to A$, $x_1,x_2 \in G$, and
  $y$ in $H$, if $x_2 \wbisima y$, then
  \begin{center}
    $[x_1,x_2] \in \bot^G$ iff $(x_1,y) \in \bot^{\combinea{G}{H}}.$
  \end{center}
\end{prop}
\begin{proof}
  By Proposition~\ref{prop:combine:adeq}, it is enough to prove that
  the right-hand side is equivalent to $(x_1,x_2) \in
  \bot^{\combinea{G}{G}}$, which is straightforward by hypothesis.
\end{proof}

%% B

We conclude this section by stating the main property of blind
composition, Proposition~\ref{prop:combinecombine} below, which will
be used extensively in the next section.

To start with, recall the following notation from
Section~\ref{subsubsec:notation:lts}. There, considering a morphism $p
\colon G \to A$ of reflexive graphs, we defined $x \xxOt{A}{r} x'$,
for $x,x' \in \ob(G)$ and $r \colon p(x) \ot p(x')$ in
$A^\star$.  Namely, this denotes any path $r' \colon x \ot^\star
x'$ in $G$, such that $\idfree{p^\star(r')} = \idfree{r}$.

In order to state Proposition~\ref{prop:combinecombine}, we now need
to equip $\freecat{A}$ with complementarity structure, but we cannot
do it over the graph $\Sierp$, because closed-world paths may contain
more than one $\tick$ edge, hence cannot all be mapped to $\Sierp$.  We
thus define \emph{categories} with complementarity.

The notions of relation, partial functionality, functionality,
totality, and domain on reflexive graphs carry over to categories,
e.g., a relation $A \modto B$ is a subcategory $R \subseteq A \times
B$. The only subtlety is that the definitions imply certain
functoriality properties.  E.g., for any composites $g \rond f$ in $A$
and $g' \rond f'$ in $B$, if $(f,f') \in R$ and $(g,g') \in R$,
because $R$, as a subcategory, is stable under composition, we have
for free that $(g \rond f, g' \rond f') \in R$.  Similarly, if $(x,y)
\in R$ for objects $x \in A$ and $y \in B$, then $(\id_x, \id_y) \in
R$.  We thus rename partial functionality and functionality into
partial functoriality and functoriality in this setting.
\begin{defi}
  A \emph{category with complementarity} is a category $A$, equipped
  with a subcategory $\aW$, a relation $\compata \colon A^2 \modto
  \aW$, and a functor $\labela \colon \aW \to \freecat{\Sierp}$, such
  that the composite $A^2 \modto \aW \to \freecat{\Sierp}$ is
  partially functorial and symmetric.
\end{defi}
Again, we let $\acoh = \dom(\compata)$ and write $a \coh a'$ for
$(a,a') \in \acoh$. We further denote the map $\acoh \into A^2 \modto
\aW \to \freecat{\Sierp}$ by $(a,b) \mapsto (a \dpara b)$, and deem
morphisms in $\aW$ \emph{closed-world}.

Defining functors with complementarity in the obvious way, we obtain:
\begin{prop}
  Categories with complementarity form a (locally small) category
  $\CCompl$.
\end{prop}

Consider the functor $\UCompl \colon \CCompl \to \GCompl$
mapping any category with complementarity $\C$ to its underlying graph, say $G$, which we equip
with complementarity structure as follows. First, define
$\GW$ and $\labelG$ by the pullback
\begin{center}
  \Diag{%
    \pbk{m-2-1}{m-1-1}{m-1-2} %
  }{%
    \GW \& \C^\W \\
    \Sierp \& \freecat{\Sierp}. %
  }{%
(m-1-1) edge[labelu={i}] (m-1-2) %
edge[labell={\labelG}] (m-2-1) %
(m-2-1) edge[labeld={\eta}] (m-2-2) %
(m-1-2) edge[labelr={\labelC}] (m-2-2) %
  }
\end{center}
Furthermore, let $\compatC$ consist of all triples $(x,y,z)$ of
vertices (resp.\ edges) such that $z \in \GW$ and
$\compatibleC{x}{y}{z}$ (which is a pullback of $(\compatC) \into \C^2
\times \CW$ along $\C^2 \times \GW \into \C^2 \times \CW$).  This
clearly equips $G$ with complementarity structure and extends to the
announced functor $\UCompl$.

This functor does not appear to have a left adjoint, because
complementarity in $G$ may behave badly w.r.t.\ composition in
$\freecat{G}$. However, we may define the following candidate structure on $\freecat{G}$.
Consider any graph with complementarity $G$, and
  let us start by defining a complementarity structure on $G^\star$.
  Let $(G^\star)^\W$ denote the subcategory of closed-world paths in $G$, 
  i.e., $(G^\star)^\W = (G^\W)^\star$.
  Accordingly, let $\labelof{{G^\star}}$ be the composite
  $$(G^\W)^\star \xto{(\labelG)^\star} \Sierp^\star \xto{\idfree{-}} \freecat{\Sierp}.$$
  Finally, consider the functor
  $$(G^2)^\star \xto{\pairing{\pi^\star}{(\pi')^\star}} (G^\star)^2.$$
  It yields a relation $(G^2)^\star \modto (G^\star)^2$, whose
  converse we use to define $\compatGstar$ as the composite relation
  $$(G^\star)^2 \xmodto{\converse{\pairing{\pi^\star}{(\pi')^\star}}} (G^2)^\star \xmodto{(\compatG)^\star}
  G^\star.$$ Concretely, $\compatGstar$ is $\compatG$ on objects, 
  and on paths, we have $(r_1,r_2) \compatGstar r$ iff all three
  paths $r_1, r_2,$ and $r$ have the same length $n$ and
  $(r^i_1,r^i_2) \compatG r^i$ for all $i \in n$.  This clearly makes
  $G^\star$ into a category with complementarity.
  
  Let us now define our candidate complementarity structure on $\freecat{G}$,
  for any $G \in \GCompl$.  Let first $\freecat{G}^\W$ be the image of
  $(G^\star)^\W \into G^\star \xto{\idfree{-}} \freecat{G}$, i.e., all
  $\id$-free, closed-world paths.  This in particular induces a
  functor $(G^\star)^\W \to (\freecat{G})^\W$, with which
  $\labelof{(G^\star)}$ is obviously compatible, hence we define
  $\labelof{\freecat{G}}$ to be the induced functor.  Finally, let
  $\compatfcG$ be the following relational composite, where the backwards
  arrow denotes a converse:
  $$(\freecat{G})^2 \xot{(\idfree{-})^2} (G^\star)^2 \xmodto{\compatGstar} (G^\star)^\W \to
  \freecat{G}^\W.$$ Concretely, $(\rho_1,\rho_2) \compatfcG \rho_3$
  iff there exist $(r_1,r_2) \compatGstar r_3$ such that $\idfree{r_i}
  = \rho_i$ for $i = 1,2,3$.  Intuitively, $\rho_1$ and $\rho_2$ are
  coherent if upon insertion of identities at appropriate places they
  become pointwise coherent.  

  The relational composite $\freecat{G}^2 \xmodto{\compatfcG} \freecat{G}^\W \xto{\labelfcG}
  \Sierp$ is obviously symmetric and furthermore partially functional
  on objects, so in order to equip $\freecat{G}$ with complementarity
  structure, it only misses partial functoriality on morphisms.

  \begin{defi}
    Let $\GCompl_+$ denote the full subcategory of $\GCompl$ spanning
    objects $G$ such that the above composite is partially functorial
    on morphisms, which we call \emph{functorial} graphs with
    complementarity.
  \end{defi}

  \begin{exa}
    A sufficient condition for a graph with complementarity $G$ to be
    functorial is to satisfy
     \begin{enumerate}[label=(\roman*)]
     \item for any two edges $e$ and $e'$, and object $x$, if $e \coh
       e'$, $e \neq id$, and $e \coh \id_x$, then $e'$ is an identity;
     \item for all edges $e$ and $e'$, $\idfree{(e \dpara \id) ; (id
         \dpara e')}$ and $\idfree{(\id \dpara e') ; (e \dpara \id)}$
       are defined at the same time and then equal.
     \end{enumerate}    
     The three graphs with complementarity $\ccs$, $\SSSL$, and
     $\TTTL$ satisfy these conditions, hence are functorial.
  \end{exa}

  The forgetful functor $\UCompl$ of course lands into $\GCompl_+$ and
  we view it as a functor $\CCompl \to \GCompl_+$ from now on.
  \begin{prop}
    The above construction of $\freecat{G}^\W$,
    $\labelof{\freecat{G}}$, and $\compatfcG$ extends to a left
    adjoint to $\UCompl$, which coincides with $\freecatfun$ on
    underlying graphs.
  \end{prop}
  We henceforth denote the left adjoint by $\freecatfun$. 
\begin{proof}
  Proving that this is left adjoint to $\UCompl$ reduces to showing
  that the composite
  $$\GCompl_+(G, \UCompl (C)) \into \Gph(G, U (C)) \xto{\iso} \Cat(\freecat{G},C)$$
  factors through $\CCompl(\freecat{G},C) \into \Cat(\freecat{G},C)$, and conversely the composite
  $$\CCompl(\freecat{G}, C) \into \Cat(\freecat{G},C) \xto{\iso} \Gph(G, U (C))$$
  factors through $\GCompl_+ (G, \UCompl (C)) \into \Gph(G, U (C))$,
  which is routine.
\end{proof}

We may now state the main property of blind composition:
\begin{prop}\label{prop:combinecombine}
  For any graphs with complementarity $G$ and $H$ over $A$, and
  transition sequences $x \xTrans{A}{\rho_x} x'$ and $y
  \xTrans{A}{\rho_y} y'$ respectively in $G$ and $H$, if
  $(\rho_x,\rho_y) \compatfca \rho$, then $(x,y) \xTrans{A}{\rho}
  (x',y')$ in $\combinea{G}{H}$.
\end{prop}
\begin{proof}
  Let $p \colon G \to A$ and $q \colon H \to A$ be the given
  projections. Let also $(r^i_1,r^i_2) \compata r^i_3$ for
  all $i \in n$ witness the fact that $(\rho_x,\rho_y) \compatfca
  \rho$.  It is enough to prove $(x,y)
  \xtrans{A^\star}{r_3} (x',y')$, which is in fact a
  trivial induction on the length of $r_3$ using the definition of
  $\combinea{G}{H}$.
\end{proof}

\subsection{Trees}
Returning to our main question, we know by Theorem~\ref{thm:bisim}
that the graph morphism $\TTT \to \SSS$ is a functional, strong
bisimulation over $\QF$. Hence, by Proposition~\ref{prop:change of
  base}, we have:
\begin{prop}\label{prop:strongbisima}
  The graph morphism $\TTTL \to \SSSL$ is a functional, strong
  bisimulation over $\LLL$, and thus also over $\A$.
\end{prop}
In this section, we introduce a criterion for a relation $R \colon G
\modto H$ between modular graphs with complementarity over some
adequate alphabet $A$, which essentially ensures that if $R \subseteq
\wbisima$, then $R$ is fair.  This will reduce our main question to
proving that the full relation induced by the map $\ccs \into \TTTL$
is included in weak bisimilarity over $\A$, which we do in
Section~\ref{subsec:horror}.

Our criterion will rest upon the notion of $A$-tree, for any graph
with complementarity $A$, which is directly inspired by the work of
Brinksma et al.\ on
\emph{failures}~\cite{DBLP:journals/iandc/RensinkV07}.

Let the set $\atrees$ of \emph{$A$-trees} consist of possibly infinite terms in the
grammar
\begin{mathpar}
  \inferrule{%
    \ldots \\ v_i \vdash t_i \\ \ldots \\ (\forall i \in n) %
  }{% 
    v \vdash \sum_{i \in n} a_i.t_i %
  }~(n \in \Nat)
\end{mathpar}
where for all $i \in n$, $a_i \colon v_i \to v$ in $A$ is not silent,
i.e., $a_i \in \aW$ implies $\labela(a_i) \neq \id$. $A$-trees form a
reflexive graph over $A$ with edges determined by
$$(v \vdash \sum_{i \in n} a_i.t_i) \xot{a_i} (v_i \vdash t_i).$$

\begin{defi}
  A modular graph $p \colon G \to A$ over $A$ \emph{has enough
    $A$-trees} iff for all $x \in G$, $v \in A$ such that $p(x) \coh
  v$, for all $A$-trees $v \vdash t$, there exists $x_t \in G$ such
  that $x \coh x_t$ and $x_t$ is weakly bisimilar to $t$ (over $A$).
\end{defi}
\begin{rem}\label{rem:trees}
  In the case where $x \coh x'$ iff $p(x) \coh p(x')$,
  this is equivalent to requiring that for all $a \in A$ and $A$-tree
  $t$ over $a$, there exists $x_t \in G$ such that $p(x_t) = a$ and
  $x_t \wbisima t$.
\end{rem}
\begin{exa}\label{ex:enoughatrees}
  $\ccs$, $\SSSL$, and $\TTTL$ have enough $\A$-trees, and
  Remark~\ref{rem:trees} applies.
\end{exa}

$A$-trees yield a new testing equivalence, called \emph{$A$-tree}
equivalence, as follows.
\begin{defi}
  For any modular $p \colon G \to A$, let
  $\treeqGa$ be the relation defined by $x \treeqGa y$ iff
  $\eqtestable{x}{y}$ and for all $v \in A$ such that $p(x) \coh v$
  and $A$-trees $t \in \atrees_v$,
  \begin{center}
    $(x,t) \in \bot^{\combinea{G}{\atrees}}$ iff $(y,t) \in
    \bot^{\combinea{G}{\atrees}}$.
  \end{center}
\end{defi}

A graph with complementarity $A$ \emph{has enough ticks} iff for all
$a \in A$, there exists an edge $\tick_a \colon a' \to a$ such that
$\labela(\tick_a) = \tick$. Furthermore, $A$ is \emph{inertly silent}
iff for all $e \colon b \to a$ in $\aW$ such that $\labela(e) = \id$,
we have $a = b$ and $e = \id_a$.
\begin{defi}\label{def:nice}
  A graph with complementarity $A$ is a \emph{nice alphabet} iff it
  has enough ticks, and is finitely branching and inertly silent.
\end{defi}

\begin{exa}
  $\A$ is a nice alphabet, but $\LLL$ is not, because it is not
  inertly silent.
\end{exa}

The main property of $A$-trees is:

\begin{prop}\label{prop:failures}
  Consider any modular $G$ and adequate $p \colon G \to A$, where $G$
  has enough $A$-trees and $A$ is a nice alphabet. Then,
  ${\faireqof{G}} = {\treeqGa}$.
\end{prop}

We start with some preparation.  Let a path in $A$ be \emph{loud} iff
it contains no silent (=identity if $A$ is inertly silent) edge, and
\emph{$\tick$-free} iff no edge is in $\inv{(\labela)} (\tick)$.  Let
the set $\FFF_a$ of \emph{failures} over $a \in A$ consist of all
pairs $(p,L)$, where $p \colon a' \tostar a$ is any loud, $\tick$-free
path in $A$ and $L \subseteq A^\star$ is a set of loud paths such that
for all $q \in L$, $\cod(q) = a'$.

  We define a map $\failoffun \colon \FFF_a \to \atrees_a$ to $A$-trees
  over $a$, for all $a$, by induction on $p$, followed by coinduction
  on $L$:
  $$
  \begin{array}[t]{rcll}
      (e \rond p, L) & \mapsto & e.(\failof{p,L}) + \tick_a.0 \\
    (\epsilon, L) & \mapsto & 
    % \left \{  
    %   \begin{array}[c]{ll} 
    %     0 & \mbox{if $L = \emptyset$} \\ 
        \failof{L} 
      %   & \mbox{otherwise} 
      % \end{array} 
      % \right .  
      \\[.5em]
      L & \mapsto & 
      % \left \{  
      % \begin{array}[c]{ll} 
      %   \tick & \mbox{if $\epsilon \in L$} \\ 
        \sum_{\ens{e \in A(-,a) \aalt L \cdot e \neq \emptyset}}
        e.\failof{L \cdot e} 
      %   & \mbox{otherwise}, 
      % \end{array} 
      % \right .  
  \end{array}
  $$
  where $L \cdot e$ is the set of paths $p$ such that $(e \rond p) \in
  L$.  Note in particular that if $L = \emptyset$ or $\ens{\epsilon}$,
  then $\failof{L} = 0$.  The sum is finite at each stage because $A$
  is finitely branching, and we use the fact that $A$ has enough
  ticks.

\proof[Proof of Proposition~\ref{prop:failures}]
  It is straightforward to show that ${\faireqG} \subseteq
  {\treeqGa}$, by Proposition~\ref{prop:combine:adeq}. For the
  converse, assume $\eqtestable{x}{y}$ and $x \fairneqG y$.  This
  means that there exists $z$ such that $x \coh z$ and $y \coh z$,
  and, w.l.o.g., $(y,z) \in \bot^{\combinea{G}{G}}$ and $(x,z) \notin
  \bot^{\combinea{G}{G}}$.
  
  By the latter, we obtain a transition sequence $(x,z) \Leftarrow
  (x',z')$, such that for no $(x'',z'')$ we have $(x',z')
  \xLeftarrow{\tick} (x'',z'')$.  Let $r$ be the given path witnessing
  $(x,z) \Leftarrow (x',z')$.  Its second projection $(\pi')^\star(r)$
  is mapped by $p^\star$ to a path in $A$, from which we remove all
  identity edges (which are also all silent edges by $A$ being inertly
  silent) to obtain $\rho = \idfree{(p \rond \pi')^\star(r)}$, a loud,
  $\tick$-free path in $A$.  Further let $L \subseteq A^\star$ be the
  set of all $\idfree{p^\star(r')}$ for paths $r' \colon z' \ot^\star
  z''$.  Let $t = \failof{\rho,L}.$ We show $(x,t) \notin
  \bot^{\combinea{G}{\atrees}}$ and $(y,t) \in
  \bot^{\combinea{G}{\atrees}}$.

  For the first point, $t \xLeftarrow{\rho} t'$, with $t' =
  \failof{\epsilon,L}$, hence $(x,t) \Leftarrow (x',t')$, by
  Proposition~\ref{prop:combinecombine}.  Now, assume $(x',t')
  \xLeftarrow{\tick} b''$.  By definition of $\combinea{G}{\atrees}$,
  we split this into $x' \xLeftarrow{\rho_1} x''$ and $t'
  \xLeftarrow{\rho_2} t''$, with $b'' = (x'',t'')$. But then $z'
  \xLeftarrow{\rho_2} z''$ by construction of $t$, and hence $(x',z')
  \xLeftarrow{\tick} (x'',z'')$ (by
  Proposition~\ref{prop:combinecombine}), contradicting $(x,z) \notin
  \bot^{\combinea{G}{G}}$.

  Let us now show $(y,t) \in \bot^{\combinea{G}{\atrees}}$.  For any
  $(y,t) \Leftarrow (y',t')$, we have accordingly $t
  \xLeftarrow{\rho'} t'$. 
  By construction, $\rho'$ is in the prefix closure of $\rho \rond L$ ,i.e.,
  $$\rho' \in \ens{r \in A^\star \aalt \exists r' \in A^\star, l \in L, r \rond r' = \rho \rond l}.$$
\begin{itemize}
\item If $\rho'$ is a strict prefix of $\rho$, then by construction $t'
  \xLeftarrow{\tick} 0$ and we are done by
  Proposition~\ref{prop:combinecombine}, since $(\id,\tick) \compatA
  \tick$.
\item Otherwise, $\rho$ is a prefix of $\rho'$. Let $\rho''$ be the unique path such
  that $\rho' = \rho \rond \rho''$.  We have $\rho'' \in L$,
  hence by construction of $L$ there exists $z''$ such that $z
  \xLeftarrow{\rho} z' \xLeftarrow{\rho''} z''$, and thus $(y,z)
  \Leftarrow (y',z'')$, by Proposition~\ref{prop:combinecombine}. By
  $(y,z) \in \bot^{\combinea{G}{G}}$, there exists $(y',z'')
  \xLeftarrow{\tick} (y'',z''')$, which projects to $y'
  \xLeftarrow{\rho_y} y''$ and $z'' \xLeftarrow{\rho_z} z'''$.  But
  then $t' \xLeftarrow{\rho_z} t''$, hence $(y',t') \xLeftarrow{\tick}
  (y'',t'')$, by Proposition~\ref{prop:combinecombine} again, which
  concludes the proof. \qed
\end{itemize}%

\begin{cor}\label{cor:wbisim:fair}
  For any nice alphabet $A$, modular $G$ and $H$, adequate $p \colon G
  \to A$ and $q \colon H \to A$, and relation $R \colon G \modto H$
  over $A$ such that $R \subseteq {\wbisima}$, if $G$ and $H$ have
  enough $A$-trees and $R$ preserves and reflects $\eqtestable{}{}$,
  then for any $x R x'$ and $y R y'$, we have ${x \faireqG y}$ iff $x'
  \faireqH y'$.
\end{cor}
\proof
  We have

  \begin{minipage}[b]{.8\linewidth}
    \centering
  $x \faireqG y$

$\Updownarrow$\makebox[0pt][l]{(by Proposition~\ref{prop:failures})}

$\eqtestable{x}{y}$ and $\forall v \in p(x)^\coh.\forall t \in \atrees_v.
(x,t) \in \bot^{\combinea{G}{\atrees}} 
\Leftrightarrow 
(y,t) \in \bot^{\combinea{G}{\atrees}}$

$\Updownarrow$\makebox[0pt][l]{(by weak bisimilarity over $A$)}

$\eqtestable{x'}{y'}$ and $\forall v \in p(x)^\coh.\forall t \in \atrees_v.
(x',t) \in \bot^{\combinea{H}{\atrees}} 
\Leftrightarrow 
(y',t) \in \bot^{\combinea{H}{\atrees}}$

$\Updownarrow$\makebox[0pt][l]{(by Proposition~\ref{prop:failures} again)}

  $x' \faireqH y'$.
\end{minipage}
\qed

\subsection{Main result}\label{subsec:horror}
We now provide the missing piece to our main result, and then conclude. 

\begin{lem}\label{lem:ccsTTTL}
  The graph of $\theta \colon \ob\ccs \to \ob\TTTL$ is included in weak
  bisimilarity over $\A$.
\end{lem}
\proof We would like, for any $h \colon I \to X$ and family $P
\in \prod_{n \in \Nat} \prod_{x \in X[n]} \ccs_n$, to define 
a process term $h[P]$ with interface $I(\star)$,
which would amount to 
$$(\para_{n} \para_{x \in X[n]} P_x
[l \mapsto x \cdot s_l]),$$ 
but restricting all channels in $X (\star) \setminus h(I (\star))$. 
When $h$ is not an inclusion, this is a bit tricky, because in our De Bruijn-like syntax
$\Gam \vdash \nu.P$ may be understood as $\Gam \vdash \nu (\Gam+1).P$. That is,
$\nu$-bound channels are always strictly greater than names in $\Gam$.

The correct way of doing this is to use subtraction, i.e.,
restrict channels in $X(\star) - I(\star)$, and accordingly rename channels in the body.
Formally, let $\gamma_h$ be the unique non-decreasing
isomorphism $(X (\star) \setminus h(I (\star))) \to (X (\star) -
I(\star))$ (which exists thanks to $h$ being monic), and let
$h[P]$ be
$$I(\star) \vdash \nu^{X (\star) - I (\star)}. \left (\para_{n} \para_{x \in X[n]} P_x
\left [
  \begin{array}[c]{ll}
    l \mapsto \epsilon a.(h_\star(a) = x \cdot s_l) & \mbox{if $x \cdot s_l \in h_\star(I(\star))$} \\
    l \mapsto \gamma_h(x \cdot s_l) & \mbox{otherwise} 
  \end{array}
\right ] \right ),$$ where $\epsilon$ is Hilbert's definite
description operator, i.e., $\epsilon a.A(a)$ denotes the unique $a$
such that $A(a)$ holds, and $\nu^n.P$ denotes $\nu. \ldots \nu.P$, $n$
times. 

\begin{defi}
  Let $\III \colon \ob \ccs \modto \ob\TTTL$ consist, for any $P \in
  \prod_{n \in \Nat} \prod_{x \in X[n]} \ccs_n$, of all pairs
  $(h[P],(I,h,\theta(P))$.
\end{defi}

  Let $\RRR$ be the composite relation
  $$\ob\ccs \xmodto{\scong} \ob\ccs \xmodto{\III} \ob\TTTL.$$

  We show that $\RRR$ is an expansion~\cite[Chapter 6]{SangioRutten}, which implies
  that it is a weak bisimulation.  Hence, since the graph of $\theta$
  is included in $\RRR$, this entails the desired result.  Let $x
  \xhatot{\alpha} x'$ iff
  \begin{itemize}
  \item either $\alpha$ is an identity and $x \xOt x'$ in zero or one step,
  \item or $\alpha$ is not an identity and $x \xLeftarrow{\alpha} x'$.
  \end{itemize}
  Recall:
  \begin{defi}
    $\RRR$ is an expansion iff for all $P \relRRR T$,
    \begin{itemize}
    \item if $P \xot{\alpha} P'$, then there exists $T'$ such
      that $P' \relRRR T'$ and $T \xLeftarrow{\alpha} T'$; and
    \item if $T \xot{\alpha} T'$, then there exists $P'$ such that
      $P \xhatot{\alpha} P' \relRRR T'$.
    \end{itemize}
  \end{defi}

  First, one easily shows that transitions in $\ccs$ are dealt with by
  `heating' the right-hand side until it may match the given
  transition.

  Conversely, we show below in (1) that for any transition
  $(I,h,\theta(P)) \xot{M} (I,k,T')$, for $M \colon k \to h$ in
  $\LLL$, where $M$ is either a fork or a channel creation, then $T' =
  \theta(P')$, for some $P' \in \prod_{n \in \Nat} \prod_{y \in Y[n]}
  \ccs_n$, and $h[P] \scong k[P']$.

  Thus, any such transition, which is silent, is matched by the empty
  transition sequence, as in

\begin{center}
  \diag(.15,.3){%
    Q \& \scong \& h[P] \& \III \& (I,h,\theta(P)) \\
    \mbox{\rotatebox{90}{$=$}} \& \& \mbox{\rotatebox{90}{$\scong$}} \\
    Q \& \scong \& k[P'] \& \III \& (I,k,T'). %
  }{%
    (m-3-5) edge[labelr={M}] (m-1-5) %
  }
\end{center}

Similarly, for any transition $(I,h,\theta(P)) \xot{M}
(I,k,T')$ not falling in the previous cases,  
we prove below in (2) that there exists $P' \in
\prod_{n \in \Nat} \prod_{y \in Y[n]} \ccs_n$ and $Q'$ such that 
$h[P] \xot{\xi(M)} Q' \scong k[P']$. Thus, any such transition is matched
as in 

\begin{center}
  \diag(.6,.3){%
    \& \& Q \& \scong \& h[P] \& \III \& (I,h,\theta(P)) \\
    Q'' \& \scong \& Q' \& \scong \& k[P'] \& \III \& (I,k,T'), %
  }{%
    (m-2-7) edge[labelr={M}] (m-1-7) %
    (m-2-3) edge[labelr={\ \,\xi(M)}] (m-1-5) %
    (m-2-1) edge[labelal={\xi(M)}] (m-1-3) %
  }
\end{center}
where $Q''$ is obtained by $\scong$ being a bisimulation.

(1) As announced, let us now consider the case of a transition
$(I,h,\theta(P)) \xot{M} (I,k,T')$, for $M \colon k \to h$ in $\LLL$,
where $M$ is either a fork or a channel creation. Consider first the
case where $M$ is a fork. Let $x_1, \ldots, x_n$ be the players of
$X$, let $m_1, \ldots, m_n$ be their respective arities, and let $i_0
\in n$ be the forking player.  Let, for any $i \in n+1$,
$$\mu(i) = \left \{
\begin{array}[c]{ll}
  i & \mbox{if $i < i_0$} \\
  i_0 & \mbox{if $i = i_0$ or $i = i_0 + 1$} \\
  {i-1} & \mbox{if $i > i_0 + 1$}
\end{array} \right .$$
and $$P'_i = \left \{
\begin{array}[c]{ll}
  P_i & \mbox{if $i < i_0$} \\
  P_{i_0}^1 & \mbox{if $i = i_0$} \\
  P_{i_0}^2 & \mbox{if $i = i_0 + 1$} \\
  P_{i-1} & \mbox{if $i > i_0 + 1$},
\end{array} \right .$$
where $P_{i_0} = P^1_{i_0} \para P^2_{i_0}$.
For all $j \in n+1$, we have that $y_j$ is an avatar of
$x_{\mu(j)}$ (i.e., $x_{\mu(j)} = (y_j)^M$), 
and $P'_j = P_{\mu(j)}$ if $\mu(j) \neq i_0$, while $P_{i_0} = P'_{i_0} \para P'_{i_0 + 1}$.

Thanks to the restriction of edges
  \begin{center}
    \diag(.6,.6){%
      \& |(I)| I \\
      |(X)| X \& |(M)| M \&|(Y)| Y %
  }{%
    (I) edge[labelal={h}] (X) %
    edge[labelo={u}] (M) %
    edge[labelar={k}] (Y) %
    (X) edge[labeld={t}] (M) %
    (Y) edge[labeld={s}] (M) %
  }
\end{center}
 in $\LLL$,
for any $j \in n+1$, if $\mu(j) = i$, $l \in m_i$ and $a,b
\in I(\star)$, we have that if $h_\star(a) = x_i \cdot s_l$
and $k_\star(b) = y_j \cdot s_l$, then,
since $s \rond y_j \rond s_l = t \rond x_i \rond s_l$, both squares
\begin{center}
  \diag{%
    \star \& I \\
    {[m_i]} \& M %
  }{%
(m-1-1) edge[labelu={a,b}] (m-1-2) %
edge[labell={s_l}] (m-2-1) %
(m-2-1) edge[labeld={s \rond y_j, t \rond x_i}] (m-2-2) %
(m-1-2) edge[labelr={u}] (m-2-2) %
  }
\end{center}
commute, hence $a = b$ by monicity of $u$. 

So, for all $j \in n+1$ and $l \in m_i$, for $i = \mu(j)$, we have
$x_i \cdot s_l \in h_\star (I (\star))$ iff $y_j \cdot s_l \in k_\star
(I (\star))$, in which case
$$\epsilon a.(h_\star(a) = x_i \cdot s_l) =
\epsilon b.(k_\star(b) = y_j \cdot s_l).$$

Furthermore, we have a commuting diagram
\begin{center}
  \diag{%
    X (\star) \setminus h(I (\star)) \& M (\star) \setminus u(I (\star)) \&
    Y (\star) \setminus k(I (\star)) \\
    X (\star) - I (\star) \& \&
    Y (\star) - I (\star), %
  }{%
    (m-1-1) edge[iso] (m-1-2) %
    edge[labell={\gamma_h}] (m-2-1) %
    edge[bend left,labelu={\delta_M}] (m-1-3) %
    (m-1-3) edge[iso] (m-1-2) %
    edge[labelr={\gamma_k}] (m-2-3) %
    (m-2-1) edge[dashed,labeld={\delta'_M}] (m-2-3) %
  }
\end{center}
of bijections, where $\delta_M$ and $\delta'_M$ are obtained by
composition and the arrows marked 
$\iso$ are the respective restrictions of $t$ and $s$.
This diagram is such that for all $j \in n+1$ and $i = \mu(j)$, $l \in m_i$, if
$x_i \cdot s_l \notin h (I (\star))$, then $\delta_M(x_i \cdot s_l) =
y_j \cdot s_l$.
We have
$$h[P] =  \nu^{X (\star) - I (\star)}. \left (\para_{i \in n} P_i\left [
  \begin{array}[c]{ll}
    l \mapsto \epsilon a.(h_\star(a) = x_i \cdot s_l) & \mbox{if $x_i \cdot s_l \in h_\star (I(\star))$} \\
    l \mapsto \gamma_h(x_i \cdot s_l) & \mbox{otherwise} 
  \end{array}
\right ] \right ) $$
and
$$k[P'] =  \nu^{Y (\star) - I (\star)}. \left (\para_{j \in n+1} P'_j\left [
  \begin{array}[c]{ll}
    l \mapsto \epsilon b.(k_\star(b) = y_j \cdot s_l) & \mbox{if $y_j \cdot s_l \in k_\star (I(\star))$} \\
    l \mapsto \gamma_k(y_j \cdot s_l) & \mbox{otherwise} 
  \end{array}
\right ] \right ). 
$$
Via the renaming $\delta'_M$,
we have
$$\begin{array}{rcl}
h[P] & \scong & \nu^{Y (\star) - I (\star)}. \left (\para_{j \in n+1, j \neq i_0 + 1} P_{\mu(j)} \left [
  \begin{array}[c]{ll}
    l \mapsto \epsilon b.(k_\star(b) = y_j \cdot s_l) & \mbox{if $y_j \cdot s_l \in k_\star (I(\star))$} \\
    l \mapsto \gamma_k(\delta_M(x_i \cdot s_l)) & \mbox{otherwise} 
  \end{array}
\right ] \right ) \\
& \scong & \nu^{Y (\star) - I (\star)}. \left (\para_{j \in n+1, j \neq i_0 + 1} P_{\mu(j)} \left [
  \begin{array}[c]{ll}
    l \mapsto \epsilon b.(k_\star(b) = y_j \cdot s_l) & \mbox{if $y_j \cdot s_l \in k_\star (I(\star))$} \\
    l \mapsto \gamma_k(y_j \cdot s_l)) & \mbox{otherwise} 
  \end{array}
\right ] \right ) \\
& \xot{\id} & \nu^{Y (\star) - I (\star)}. \left (\para_{j \in n + 1} P'_j\left [
  \begin{array}[c]{ll}
    l \mapsto \epsilon b.(k_\star(b) = y_j \cdot s_l) & \mbox{if $y_j \cdot s_l \in k_\star (I(\star))$} \\
    l \mapsto \gamma_k(y_j \cdot s_l)) & \mbox{otherwise} 
  \end{array}
\right ] \right ) \\
& \scong & k[P'].
\end{array}$$
The case of a channel creation move is similar.

(2) Consider now any transition $(I,h,\theta(P)) \xot{M} (I,k,T')$,
where $M$ is an input or an output on some channel $c \in h_\star (I
(\star))$, or a synchronisation, or a tick. Then, proceeding as for the forking move above,
we may take $\mu = \id$, and still obtain $\delta_M$ and $\delta'_M$.
In all cases, we have $T'_i = \theta(P'_i)$, for some family $P'$
of CCS processes. E.g., if $M$ is an input on $c$ by $x_{i_0}$,
then $P'_i = P_i$ for all $i \neq i_0$, and $P_{i_0} \scong c.P'_{i_0} + P''$.
We have $h[P] \xot{\xi(M)} Q$, where
$$Q = \nu^{X (\star) - I (\star)}. \left (\para_{i \in n} P'_i\left [
  \begin{array}[c]{ll}
    l \mapsto \epsilon a.(h_\star(a) = x_i \cdot s_l) & \mbox{if $x_i \cdot s_l \in h_\star (I(\star))$} \\
    l \mapsto \gamma_h(x_i \cdot s_l) & \mbox{otherwise} 
  \end{array}
\right ] \right ), $$
which via the renaming $\delta'_M$, is structurally congruent to 
\begin{center}
  $ \begin{array}[b]{l} \nu^{Y (\star) - I (\star)}. \left (\para_{i
        \in n} P'_i\left [
        \begin{array}[c]{ll}
          l \mapsto \epsilon b.(k_\star(b) = y_i \cdot s_l) & \mbox{if $y_i \cdot s_l \in k_\star (I(\star))$} \\
          l \mapsto \gamma_k(\delta_M(x_i \cdot s_l)) & \mbox{otherwise} 
        \end{array}
      \right ] \right ) \\
    {} \scong 
    \nu^{Y (\star) - I (\star)}. \left (\para_{i \in n} P'_i\left [
        \begin{array}[c]{ll}
          l \mapsto \epsilon b.(k_\star(b) = y_i \cdot s_l) & \mbox{if $y_i \cdot s_l \in k_\star (I(\star))$} \\
          l \mapsto \gamma_k(y_i \cdot s_l) & \mbox{otherwise} 
        \end{array}
      \right ] \right ) \\
    {} \scong k[P'],
  \end{array}$
\end{center}
which concludes the proof.\qed

This leads to our first full abstraction result:
\begin{cor}\label{cor:wbisim}
  The composite $\ob(\ccs) \into \ob (\TTTL) \to \ob (\SSSL)$ is included in weak bisimilarity.
\end{cor}
\begin{proof}
  By the previous lemma, Proposition~\ref{prop:strongbisima}, and
  the fact that weak bisimulations are closed under composition.
\end{proof}

\begin{cor}\label{cor:final}
  The composite $\ob\ccs \xto{\theta} \ob\TTTL \xto{\translfun}
  \ob\SSSL$ is fair, and we have for all CCS processes $P$ and $Q$
  over any common $n$:
  \begin{mathpar}
    P \faireqs Q \and \mbox{iff} \and \transl{\theta(P)} \faireq
    \transl{\theta(Q)}.
  \end{mathpar}
\end{cor}
\proof
We have:
\begin{center}
\hspace*{-10em}  \begin{minipage}[t]{.8\linewidth}
\centering

    $P \faireqs Q$

    $\Updownarrow$\makebox[0pt][l]{(by Proposition~\ref{prop:fairccs})}

    $P \faireqof{\ccs} Q$

    % $\Updownarrow$\makebox[0pt][l]{(by Lemma~\ref{lem:ccsTTTL})}

    % $\theta(P) \faireqof{\TTTL} \theta(Q)$

    % $\Updownarrow$\makebox[0pt][l]{(by Proposition~\ref{prop:change
    % of
    % base})}

    $\Updownarrow$\makebox[0pt][l]{(by
      Corollaries~\ref{cor:wbisim:fair} and~\ref{cor:wbisim}, and
      Example~\ref{ex:enoughatrees})}

    $\transl{\theta(P)} \faireqof{\SSSL} \transl{\theta(Q)}$
    
    $\Updownarrow$\makebox[0pt][l]{(by Corollary~\ref{cor:SSSLI})}

    $\transl{\theta(P)} \faireq \transl{\theta(Q)}$,
  \end{minipage}
\end{center}
as desired.\qed

\section{CCS as a playground}\label{sec:ccs}
At last, we prove that $\Dccs$ forms a playground. We rewind to the
beginning of Section~\ref{subsec:pseudodouble}, to state things a bit
more formally.

\subsection{A pseudo double category}
Recall from \citetalias{2011arXiv1109.4356H} the notion of dimension
in $\C$: $\star$ is the sole object of dimension $0$, all $[n]$'s have
dimension $1$, all $\iotaposni$, $\iotanegni$, $\paraln$, $\pararn$,
$\tickn$, and $\nun$ have dimension $2$, all $\paran$ have dimension
3, and all $\taunimj$ have dimension $4$.  By extension, a presheaf
$F$ has dimension $i$ if $F$ is empty over objects of dimension
strictly greater than $i$.  We call \emph{interfaces} the presheaves
of dimension $0$ (i.e., empty beyond dimension $0$), \emph{positions}
the finite presheaves of dimension $1$.

We start by viewing the base pseudo double category of our playground, $\Dccs$,
as a sub-pseudo double category of the following pseudo double category $\Dccso$.
\begin{defi}
  Let $\Dccso$ have:
  \begin{itemize}
  \item as objects all positions,
  \item horizontal category $\Dccsoh$ the subcategory of $\Chatf$
    consisting of positions and monic arrows between them;
  \item vertical (bi)category $\Dccsov$ the sub-bicategory of
    $\Cospan{\Chatf}$ consisting of positions and cospans of monic
    arrows between them;
  \item and all commuting diagrams
    \begin{equation}
      \diag{%
        X \& X' \\
        U \& V \\
        Y \& Y' %
      }{%
        (m-1-1) edge[into,labelu={h}] (m-1-2) %
        (m-2-1) edge[into,labelo={k}] (m-2-2) %
        (m-3-1) edge[into,labeld={l}] (m-3-2) %
        (m-1-1) edge[into,labell={s}] (m-2-1) %
        (m-1-2) edge[into,labelr={s'}] (m-2-2) %
        (m-3-1) edge[linto,labell={t}] (m-2-1) %
        (m-3-2) edge[linto,labelr={t'}] (m-2-2) %
      } \hfil \mbox{as double cells} \hfil
      \doublecellpro{X}{X'}{Y}{Y',}{h}{U}{V}{l}{(h,k,l)}
      \label{eq:doublecelll}
    \end{equation}
    with all $\into$ arrows monic.
  \end{itemize}
\end{defi}
Horizontal composition of double cells is induced by composition in
$\Chatf$.  Vertical composition of double cells is induced by pushout
in $\Chatf$.  It is of course the vertical direction here which is
pseudo.

\begin{prop}
  $\Dccs$ is the pseudo double category obtained by restricting
  $\Dccso$ to vertical morphisms which are either equivalences or
  finite composites of moves.
\end{prop}
Since $\Dccs$ is again the only involved (candidate) playground in this section,
we often omit the superscript. E.g., $\Do$ denotes $\Dccso$.

The rest of Section~\ref{sec:ccs} is devoted to proving:
\begin{thm}
  $\D$, equipped with
  \begin{itemize}
  \item as individuals, all positions of the shape $\C (-,[n])$, i.e.,
    all strictly representable presheaves,
  \item moves as moves, seeds as basic
    moves, and full moves as full moves.
  \end{itemize}
forms a playground.
\end{thm}

We start with a combinatorial correctness criterion for characterising
plays $U \colon X \proto Y$ among general cospans $X \into U \otni Y$,
which we then put to use in proving the theorem. Our convention for
plays $X \into U \otni Y$ is that the (candidate) final position is
always on the left.

\subsection{Correctness}
We prove a few properties of plays, which we then find are sufficient
for a cospan to be a play.  

Given a play $X \into U \otni Y$, we start by forgetting the cospan
structure and exhibiting some properties of $U$ alone.

 \begin{defi}
   A \emph{core} of a presheaf $U \in \Chatf$ is an element of
   dimension $> 1$ which is not the image (under the action of some
   morphism of $\C$) of any element of higher dimension.
 \end{defi}

 Here is a first easy property of plays. Observing that for all seeds
 $Y \into M \otni X$, $M$ is a representable presheaf, we put:
\begin{defi}
  A presheaf $U$ is \emph{locally 1-injective} iff for any
  seed $Y \into M \otni X$ with interface $I$ and core $\mu \in U (M)$, if
  two elements of $M$ are identified by the Yoneda morphism $\mu
  \colon M \to U$, then they are in (the image of) $I (\star)$.
\end{defi}
The name `locally 1-injective' is designed to evoke the fact that $M
\to U$ is injective above dimension 0.
\begin{prop}
  Any play $U$ is locally 1-injective.
\end{prop}
\begin{proof}
  Choose a decomposition of $U$ into moves; $\mu$ corresponds
  to precisely one such move, say $M'$, obtained, by
  definition, from some seed $M$ as a
  pushout~\eqref{eq:extend}.  By construction of pushouts in presheaf
  categories, $M'$ is obtained from $M$ by identifying some channels
  according to $I \to Z$.
\end{proof}

% \begin{lem}
%   $U$ is locally 1-injective iff for any non-extended move $Y \into M
%   \otni X$ and core $\mu \in U (M)$, the square
%   \begin{center}
%     \diag{%
%       Y \& M \& U \fatbslash \mu \\
%       M \& \& U
%     }{%
%       (m-1-1) edge[labelu={}] (m-1-2) %
%       edge[labell={}] (m-2-1) %
%       (m-1-2) edge (m-1-3)
%       (m-2-1) edge[labeld={}] (m-2-3) %
%       (m-1-3) edge[labelr={}] (m-2-3) %
%     }
%   \end{center}
%   is a pushout and $Y \to U \fatbslash \mu$ is injective on $Y -
%   \past{\mu}$.
% \end{lem}
% \begin{proof}
%   \textit{(if)}\ Consider any two elements $x,x' \in M$, identified by $M \to U$.
%   
% \end{proof}

We now extract from any presheaf a graph, which represents its
candidate causal structure.  Observe that, in $\C$, for any object
$\mu$ of dimension $>1$ (i.e., a move), all morphisms from a player,
i.e., an object of the shape $[n]$, to $\mu$ have exactly one of the
shapes $f \rond s \rond f'$ and $f \rond t \rond f'$. In the former
case, the given player belongs in the final position of $\mu$ and we
say that it is a \emph{source} of $\mu$; in the latter, it belongs in
the initial position and we call it a \emph{target}. We extend these
notions to arbitrary presheaves.
 \begin{defi}
   In any $U$, the \emph{sources} of a core $\mu$ are the players $x$
   with a morphism, in $\elements U$ (the category of elements of $U$,
   recalled in Section~\ref{subsec:diagrams}), of the shape $x \xto{f
     \rond s \rond f'} \mu$ to $\mu$; its \emph{targets} are the
   players $y$ with a morphism of the shape $y \xto{f \rond t \rond
     f'} \mu$.
 \end{defi}
 \begin{exa}
   In the representable $\paran$, there is one target, $l \rond t$ (or
   equivalently $r \rond t$), and two sources, $s_1 = l \rond s$ and
   $s_2 = r \rond s$, respectively the left and right players obtained
   by forking.  Another example is $\taunimj$, which has two targets,
   the sender $\sender \rond t$ and the receiver $\receiver \rond t$,
   and two sources $\sender \rond s$ and $\receiver \rond s$.
 \end{exa}
 \begin{defi}
   A channel $a \in M (\star)$ is \emph{created} by a seed $Y
   \xinto{s} M \xotni{t} X$ iff $a \in Y (\star) \setminus X (\star)$.
 \end{defi}
 Recall that in $\C$, the channels known to a player $[n]$ are
 represented by morphisms $s_1, \ldots, s_n \colon \star \to [n]$, so
 that in a presheaf $U \in \Chatf$, the channels known to $x \in U[n]$
 are $x \cdot s_1$, \ldots, and $x \cdot s_n$.

 Given a presheaf $U$, we construct its \emph{causal
   (simple) graph} $G_U$ as follows:
 \begin{itemize}
 \item its vertices are all channels, players, and cores in $U$;
 \item there is an edge to each core from its sources and one from
   each core to its targets, as in
   \begin{center}
     \diag{%
       \mathrm{source}_1 \& \&        \mathrm{source}_2;\\
       \&        \mathrm{core} \\
       \mathrm{target}_1 \& \&        \mathrm{target}_2; %
     }{%
       (m-1-1) edge (m-2-2) %
       (m-1-3) edge (m-2-2) %
       (m-2-2) edge (m-3-1) %
       edge (m-3-3) %
     }
   \end{center}
 \item there is an edge $x \to x \cdot s_i$ for all $x \in U[n]$ and
   $i \in n$;
 \item there is an edge $a \to \mu$ for each channel $a$ created by $\mu$.
 \end{itemize}
 This graph is actually a binary relation, since there is at most one
 edge between any two vertices. It is also a coloured graph, in the
 sense that it comes equipped with a morphism to the graph $L$:
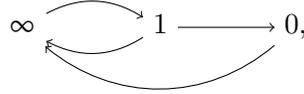
 \begin{center}
   \diag{%
     \infty \& 1 \& 0, %
   }{%
     (m-1-1) edge[bend left] (m-1-2) %
     (m-1-2) edge[bend left] (m-1-1) %
     (m-1-3) edge[bend left=40] (m-1-1) %
     (m-1-2) edge (m-1-3) %
   }
 \end{center}
 mapping cores to $\infty$, players to $1$, and channels to
 $0$. (Observe in particular that there are no edges from channels to
 players nor from cores to channels.)  For any simple graph $G$,
 equipped with a morphism $l \colon G \to L$, we call vertices of $G$
 channels, players, or cores, according to their label.

\begin{defi}
  Seen as an object of $\Gph / L$, $G$ is \emph{source-linear} iff
  for any cores $\mu, \mu'$, and other vertex (necessarily a player or a channel) $x$,
  $\mu \ot x \to \mu'$ in $G$, then $\mu = \mu'$.
  $G$ is \emph{target-linear} iff for any cores $\mu,\mu'$ and player $x$, if
  $\mu \to x \ot \mu'$ in $G$, then $\mu = \mu'$.
  $G$ is \emph{linear} iff it is both source-linear and target-linear.
\end{defi}

 \begin{prop}
   For any play $Y \xinto{s} U \xotni{t} X$, $G_U$ is linear.
 \end{prop}
 \begin{proof}
   By induction on any decomposition of $U$ into moves.
 \end{proof}

 \begin{prop}
   For any play as above, $G_U$ is acyclic (in the directed sense).
 \end{prop}
 \begin{proof}
   Again by induction on any decomposition of $U$.
 \end{proof}

 \begin{defi}
   A player $x$ in $U$ is \emph{final} iff it is not the target of any move, i.e., for no move $\mu \in U$, $x = \mu \cdot t$.
 \end{defi}
 \begin{lem}
   A player is final in $U$ iff it has no edge from any core in $G_U$.
 \end{lem}
 \begin{defi}
   A player is \emph{initial} in $U$ when it is not the source of any
   move, i.e., for no move $\mu \in U$, $x = \mu \cdot s$. A channel
   is initial when it is not created by any move.
 \end{defi}
 \begin{lem}
   A player is initial in $U$ iff it has no edge to any core in $G_U$.
 \end{lem}

Now, here is the expected characterisation:
\begin{thm}\label{thm:completeness}
  A cospan $Y \xinto{s} U \xotni{t} X$ is a play iff
  \begin{enumerate}[label=(\roman*)]
  \item $U$ is locally 1-injective, \label{cond:locinj}
  \item $X$ contains precisely the initial players and channels in
    $U$, \label{cond:X}
  \item $Y$ contains all channels, plus precisely the final players in
    $U$, \label{cond:Y}
  \item and $G_U$ is linear and acyclic. \label{cond:causality}
  \end{enumerate}
\end{thm}
Of course, we have almost proved the `only if' direction, and the rest
is easy, so only the `if' direction remains to prove. The rest of this
section is devoted to this. First, let us familiarise ourselves with
removing elements from a presheaf. For two morphisms of presheaves $U
\xto{f} V \xot{g} W$, we denote by $U \setminus W$ the topos-theoretic
difference $U \cap \neg W$ of (the images of) $f$ and $g$ in the
lattice $\Sub (V)$ of subobjects of $V$. This differs in general from
what we denote $U - W$, which is the set of elements in $V$ which are
in the image of $U$ but not that of $W$, i.e., $\sum_{c \in \C} U (c)
\setminus W (c)$. More generally, for any morphism of presheaves $f
\colon U \to V$ and set $W$, let $U - W = \sum_{c \in \C} \im(U (c))
\setminus W$.  $U - W$ is generally just a set, not a presheaf; i.e.,
its elements are not necessarily stable under the action of morphisms
in $\C$. Proposition~\ref{prop:maxcore} below exhibits a case where
they are, which is useful to us.
\begin{defi}
  For any seed $Y \into M \otni X$, let the \emph{past}
  $\past{M} = M - Y$ of $M$ be the set of its elements not in the
  image of $Y$.  For any such $M$, presheaf $U$, and core $\mu \in U
  (M)$, let $\past{\mu} = \im(\past{M})$ consist of all images of
  $\past{M}$.
\end{defi}
To explain the statement a bit more, by Yoneda, we see $\mu$ as a map
$M \to U$, so we have a set-function $$\past{M} \into {\elements M}
\to {\elements U}.$$ Observe that $\past{\mu}$ is always a set of
players and moves only, since channels present in $X$ always are in
$Y$ too.

Given a core $\mu \in U$, an important operation for us will be $$U
\fatbslash \mu = \bigcup \ens{V \into U \aalt \elements V \cap \past{\mu} =
  \emptyset}.$$ $U \fatbslash \mu$ is thus the largest subpresheaf of
$U$ not containing any element of the past of $\mu$. The good property
of this operation is:
\begin{prop}\label{prop:maxcore}
  If $\mu$ is a maximal core in $G_U$ (i.e., there is no path to any
  further core) and $G_U$ is target-linear, then $U \fatbslash \mu =
  U - \past{\mu}$, i.e., $(U \fatbslash \mu) (c) = U (c) \setminus
  \past{\mu}$ for all $c$.
\end{prop}
\begin{proof}
  The direction $(U \fatbslash \mu) (c) \subseteq U (c) \setminus
  \past{\mu}$ is by definition of $\fatbslash$. Conversely, it is
  enough to show that $c \mapsto U(c) \setminus \past{\mu}$ forms a
  subpresheaf of $U$, i.e., that for any $f \colon c \to c'$ in $\C$,
  and $x \in U(c') \setminus \past{\mu}$, $x \cdot f \notin
  \past{\mu}$. Assume on the contrary that $x' = x \cdot f \in
  \past{\mu}$. Then, of course $f$ cannot be the
  identity. Furthermore, $x'$ is either a player or a move; so, up to
  pre-composition of $f$ with a further morphism, we may assume that
  $x'$ is a player.  But then, since $f$ is non-identity, $x$ must be
  a move, with $x'$ being one of its sources or targets.  Now, up to
  post-composition of $f$ with a further morphism, we may assume that
  $x$ is a core.  So, there is either an edge $x \to x'$ or an edge
  $x' \to x$ in $G_U$. However, $x \neq \mu$, so $x \to x'$ is
  impossible by target-linearity of $G_U$, and $x' \to x$ is
  impossible by maximality of $\mu$.
\end{proof}

\begin{proof}[Proof of Theorem~\ref{thm:completeness}]
  We proceed by induction on the number of moves in $U$. If it is
  zero, then $U$ is a position; by~\ref{cond:X}, $t$ is an iso, and
  by~\ref{cond:Y} so is $s$, hence the cospan is a play. For the
  induction step, we first decompose $U$ into
  $$Y \xinto{s_2} U' \xotni{t_2} Z \xinto{s_1} M' \xotni{t_1} X,$$
  and then show that $M'$ is a move and $U'$ satisfies the
  conditions of the theorem. 

  So, first, pick a maximal core $\mu$ in $G_U$, i.e., one with no
  path to any other core. Let
  \begin{center}
    \diag{%
      \& |(I_0)| I_0 \\
      |(Y_0)| Y_0 \& |(M_0)| M_0 \& |(X_0)| X_0 %
    }{%
      (I_0) edge (Y_0) %
      edge (M_0) %
      edge (X_0) %
      (Y_0) edge (M_0) %
      (X_0) edge (M_0) %
    }
  \end{center}
  be the seed with interface corresponding to $\mu$, so
  we have the Yoneda morphism $\mu \colon M_0 \to U$.

  Let $U' = (U \fatbslash \mu)$, and $X_1 = X - \Pl(X_0)$. $X_1$ is a
  subpresheaf of $X$, since it contains all names. The square
  \begin{center}
    \Diag{%
      \pbk{m-2-1}{m-2-2}{m-1-2} %
    }{%
      I_0 \& X_1 \\
      X_0 \& X %
    }{%
      (m-1-1) edge[labelu={}] (m-1-2) %
      edge[labell={}] (m-2-1) %
      (m-2-1) edge[labeld={}] (m-2-2) %
      (m-1-2) edge[labelr={}] (m-2-2) %
    }
  \end{center}
  is a pushout, since it just adds the missing players to $X_1$. 
  Define now $Z$, $M'$, $s_1$, and $t_1$ by the pushouts
  \begin{center}
    \Diag(.02,.6){%
      \pbk[1.1em]{Y}{Y'}{Z} %
      \pullback[1.1em]{Z}{M'}{M}{draw,-,fore} %
      \pullback[1.1em]{Z}{X'}{X}{draw,-,fore} %
    }{%
      \& |(X)| Y_0 \& \& |(X')| {Z} \\
      \& \ \& \\
      \& |(M)| M_0 \& \& |(M')| M'  \& \& |(U)| U \\
      |(I)| I_0 \&\& |(Z)| X_1 \\
      \& |(Y)| X_0 \& \& |(Y')| X %
    }{%
      (X') edge[into] (U) %
      (M') edge[into] (U) %
      (Z)  edge[into] (U) %
      (Y') edge[into] (U) %
      (Z) edge[into] (Y') %
      edge (M') %
      edge (X') %
      (I) edge[into] (Y) %
      edge node[pos=.8,anchor=south] {$C$} (Z) %
      edge (M) %
      edge (X) %
      (X) edge[fore] (X') %
      (M) edge[fore] (M') %
      (Y) edge (Y') %
      (Y') edge[dashed,into,fore,labelr={t_1}] (M') %
      (X') edge[dashed,into,labelr={s_1}] (M') %
      (Y) edge[fore,into] (M) %
      (X) edge[fore,into] (M) %
    }
  \end{center}
  and the induced arrows.  We further obtain arrows to $U$ by
  universal property of pushout, which are monic because $X \into U$
  is, using~\ref{cond:locinj}. We observe that $U = M' \cup U'$, i.e.,
  the square
  \begin{center}
    \Diag{%
      \pbk{M'}{U}{U'} %
    }{%
      |(Z)| Z \& |(U')| U' \\
      |(M')| M' \& |(U)| U %
    }{%
      (Z) edge[into,labeld={}] (M') %
      edge[into] (U') %
      (U') edge[into] (U) %
      (M') edge[into] (U) %
    }
  \end{center}
  is a pushout, so $U$ is indeed a composite as claimed, with $Z \into
  M' \otni X$ a move by construction. So, it remains to
  prove that $Y \into U' \otni Z$ satisfies the conditions. First, as
  a subpresheaf of $U$, $U'$ is locally 1-injective and has a linear
  and acyclic causal graph, so satisfies~\ref{cond:locinj}
  and~\ref{cond:causality}. $U'$ furthermore
  satisfies~\ref{cond:X} by construction of $Z$ and source-linearity
  of $G_U$, and~\ref{cond:Y} because removing $\past{\mu}$ cannot
  make any non-final player final.
\end{proof}

\subsection{CCS as a pre-playground}
We now start proving:
\begin{thm}
 $\D$ forms a playground.
\end{thm}\label{thm:playground}
Axioms~\preaxrefrange{discreteness}{atomicity} are easy, as well
as~\preaxref{ax:views}, \preaxref{finiteness}
and~\preaxref{basic:full}. Furthermore, once~\preaxref{fibration} is
clear, \preaxref{fibration:continued} is also easy. This
leaves~\preaxref{fibration} and the decomposition axioms.

For~\preaxref{fibration}, i.e., the fact that $\cod \colon \DH \to
\Dh$ is a fibration, we introduce the notion of `history' for plays.
For a presheaf $U \in \Chatf$, let $\strip{U}$ be its restriction to
dimension 3, i.e., $\strip{U}(\taunimj) = \emptyset$ for all $n,i,m,j$, and
$\strip{U}(c) = U(c)$ on other objects. Further let $\El{U} = \sum_{c \in
  \ob(\C)} \strip{U} (c)$ be the set of elements of $\strip{U}$. We
have a category $\El{\Chatf}$, whose objects are those of $\Chatf$,
and whose morphisms $U \to U'$ are set-functions $\El{U} \to
\El{U'}$. We denote such morphisms with special arrows $U \histto
U'$. There is a forgetful functor $\Elfun \colon \Chatf \to
\El{\Chatf}$, which we implicitly use in casting arrows $U \to U'$ to
arrows $U \histto U'$.
\begin{defi}
  Consider any seed $X \into M \otni Y$ which is not a
  synchronisation, where $Y$ is the initial position and $X$ is the
  final position.  Then $Y$ is a representable position, say $[n]$,
  and we let the \emph{history} of $M$ be the map $p_M \colon \El{M}
  \to \El{Y}$ sending
  \begin{itemize}
  \item all channels in $\El{M} \cap \El{Y}$ to themselves,
  \item all other elements to $\id_{[n]}$.
  \end{itemize}
  The history $p_{M'}$ of a move $M'$ is the map obtained by
  pushout of the history of its generating seed $M$, as in
  \begin{center}
    \Diag(.5,1){%
      \pbk{Y}{Y'}{Z} %
      \pullback{Z}{M'}{M}{draw,-,fore} %
    }{%
      \& |(M)| M \& \& |(M')| M'  \\
      |(I)| I \&\& |(Z)| Z \\
      \& |(Y)| Y \& \& |(Y')| Y'. %
    }{%
      (Z) edge[into] (Y') %
      edge (M') %
      (I) edge[into] (Y) %
      edge node[pos=.8,anchor=south] {$C$} (Z) %
      edge (M) %
      (M) edge[fore] (M') %
      (Y) edge (Y') %
      (M') edge[dashed,history] (Y') %
      (M) edge[fore,history] (Y) %
    }
  \end{center}

\end{defi}
This defines the history of moves. We have:
\begin{prop}
  For any move $X \xinto{s} M \xotni{t} Y$, we have $p_M \rond t = \id$.
\end{prop}

We graphically represent histories by arrows between the presheaves,
as $p$ in
\begin{equation}
  \diag{%
    |(X)| X \& |(U)| U \& |(Y)| Y. %
  }{%
    (X) edge[into,labeld={s}] (U) %
    (Y) edge[linto,bend left=20,labeld={t}] (U) %
    (U) edge[history,bend left=20,labelu={p}] (Y) %
  }\label{eq:vmor}
\end{equation} 

We now define the history of sequences of moves, which we
here call \emph{sequential plays}. We denote
such a sequence $X_n \xto{M_n} X_{n-1} \ldots X_1 \xto{M_1} X_0$ by
$(M_n,\ldots,M_1)$.
\begin{defi}
  Define now the history of a sequential play $X \to (M_n,\ldots,M_1)
  \ot Y$, letting $U = M_1 \vrond \ldots \vrond M_n$ be the
  corresponding play, to be the map $U \histto Y$ defined by induction
  on $n$ as follows:
  \begin{itemize}
  \item if $\length{U} = 0$, then $t$ is an isomorphism, and the
    history is the inverse of the corresponding bijection on
    elements;
  \item if $\length{U} = 1$, then $U$ is a move $M$ and its
    history is that of $M$;
  \item if $\length{U} > 1$, then $U = (U', M)$ for some move
    $M$ and sequential play $U'$; letting $p_{U'}$ be the history of
    $U'$ obtained by induction hypothesis, we let $p_U = p_M \rond q$,
    where $q$ is defined by universal property of pushout in
    \begin{center}
      \Diag{%
        \pbk[.6cm]{U'}{U}{M} %
      }{%
        |(X)| X \& |(U')| U' \& |(Y)| Y \& |(M)| M \& |(Z)| Z \\%
        \& |(Yi)| Y \& |(U)| U \\
        \& \& |(Mi)| M. %
      }{%
        (X) edge[into,labelu={s_{U'}}] (U') %
        (Y) edge[linto,labeld={t_{U'}}] (U') %
        (U') edge[history,bend left=30,labelu={p_{U'}}] (Y) %
        (Y) edge[into,labelu={s_M}] (M) %
        (Z) edge[linto,labeld={t_M}] (M) %
        (M) edge[history,bend left=30,labelu={p_M}] (Z) %
        (U') edge[history,labell={p_{U'}}] (Yi) %
        (Y) edge[identity] (Yi) %
        (M) edge[identity] (Mi) %
        (U') edge[fore] node[pos=.95,anchor=east] {$\scriptstyle s\ $} (U) %
        (M) edge node[pos=.7,anchor=west] {$\scriptstyle t$} (U) %
        (U) edge[dashed,history,labell={q}] (Mi) %
        (Yi) edge[into,labelbl={s_M}] (Mi) %
      }
    \end{center}
  \end{itemize}
\end{defi}

\begin{prop}
  For any sequential plays $U_1, U_2 \colon X \proto Y$ with isomorphic
  compositions, we have $p_{U_1} = p_{U_2}$.
\end{prop}
\begin{proof}
  For any presheaf $U$ such that $G_U$ is source-linear and acyclic,
  consider the function $h_U \colon \El{U} \to \El{U}$ mapping
  \begin{itemize}
  \item initial players and channels to themselves,
  \item non-initial players and channels to the (unique by
    source-linearity of $G_U$) core that created them,
  \item elements of dimension 2 to their image under $t$,
  \item elements of higher dimensions to the image of one of their
    images in dimension 2 (which all map to the same element by a
    simple case analysis).
  \end{itemize}
  Observe that this map is ultimately idempotent because it is
  strictly increasing w.r.t.\ $G_U$, and let $H_U$ be the
  corresponding idempotent function.

  It is easy to see that if $X \into U \otni Y$ is a move,
  then $\im (H_U) = Y$ and $p_U = H_U$.

  Furthermore, for all composable plays $X \xproto{U'} Y \xproto{U} Z$, we
  have $H_{U \vrond U'} = H_U \rond H^U_{U'}$, where $H^U_{U'} \colon
  \El{U \vrond U'} \to \El{U}$ is the extension of $H_{U'}$ to $\El{U
    \vrond U'}$ which is the identity on $\El{U} \setminus
  \El{U'}$. Because $\im (H_{U'}) = Y$, this indeed goes to $\El{U}$.

  When $U$ is a move, this is actually equivalent to the
  diagrammatic definition of $p_{M \vrond U'}$, which entails by
  induction that for any play $U$, $p_U = H_U$, which does not depend
  on the decomposition of $U$ into moves.
\end{proof}

Just as for moves, the target map is a section of the history:
\begin{prop}
  For any play $X \into U \xotni{t} Y$, we have $p_U \rond t = \id_Y$.
\end{prop}

\begin{prop}\label{prop:compathist}
  Any double cell $(h,k,l)$ as on the left below
  \begin{mathpar}
          \diag{%
      X \& X' \\
      U \& V \\
      Y \& Y' %
    }{%
      (m-1-1) edge[labelu={h}] (m-1-2) %
      (m-2-1) edge[labelo={k}] (m-2-2) %
      (m-3-1) edge[labeld={l}] (m-3-2) %
      (m-1-1) edge[labell={s}] (m-2-1) %
      (m-1-2) edge[labelr={s'}] (m-2-2) %
      (m-3-1) edge[labell={t}] (m-2-1) %c`
      (m-3-2) edge[labelr={t'}] (m-2-2) %
    } 
\and
    \diag (1,1) {%
      U \& V \\
      Y \& Y' %
    }{%
      (m-1-1) edge[labelu={k}] (m-1-2) %
      edge[history,labell={p}] (m-2-1) %
      (m-2-1) edge[labeld={l}] (m-2-2) %
      (m-1-2) edge[history,labelr={p'}] (m-2-2) %
    }
  \end{mathpar}
 is compatible with histories $p
  \colon U \histto Y$ and $p' \colon U' \histto Y'$, in the sense that
  the square on the right
  commutes.
\end{prop}
The important point for us is:
\begin{prop}
  The vertical codomain functor $\cod \colon \DH \to \Dh$ is a
  fibration.
\end{prop}
\begin{proof}
  We first consider the restriction of $\cod$ to the full subcategory
  of $\DH$ consisting of moves and isomorphisms.  Given a move $X
  \xinto{s} M \xotni{t} Y$ and a morphism $l \colon Y' \to Y$ in
  $\Dh$, consider the pullback (in sets) and the induced arrow $t'$:
  \begin{center}
    \Diag(.5,1){%
      \pbk{Y'}{U'}{U} %
    }{%
      |(Y'i)| Y' \& \& |(Yi)| Y \\
      \& |(U')| U_0 \& \& |(U)| M \\
      \& \  \\
      \& |(Y')| Y' \& \& |(Y)| Y. %
    }{%
      (Y'i) edge[bend right=20,identity] (Y') %
      edge[labelu={l}] (Yi) %
      (Yi) edge[labelar={t}] (U) %
      edge[identity,bend right=20] (Y) %
      (U') edge[into,fore] node[pos=.3,anchor=south] {$\scriptstyle k_0$} (U) %
      (U') edge[history,labell={p'}] (Y') %
      (Y') edge[into,labeld={l}] (Y) %
      (Y'i) edge[dashed,labelo={t'}] (U')
      (U) edge[history,labelr={p}] (Y) %
    }
  \end{center}
  Now, consider $U_0$ as a presheaf over $\C_3$ by giving each element
  the type of its image under $k_0$, and checking that $U_0$, viewed
  as an $\ob (\C_3)$-indexed family of subsets of $M$, is stable under
  the action of morphisms in $\C_3$. This, in passing, equips $k_0$
  and $t'$ with the structure of maps in $\Chatf$.

  Furthermore, let the $(n,i,m,j)$-\emph{horn} (see, e.g., Joyal and
  Tierney~\cite{JoyalTierney} for the origin of our terminology)
  $\hornnimj$ be the representable presheaf on $\taunimj$, minus the
  element $\id_{\taunimj}$, and consider the family $A$ of commuting
  squares
  \begin{center}
    \diag{%
      \hornnimj \& U_0 \\
      \taunimj \& M, %
    }{%
      \sq{w}{i}{k_0}{w'} %
    }
  \end{center}
  where $i$ is the inclusion. Define then $U$ and $k$ by pushout as in
   \begin{center}
    \Diag{%
      \pbk{taus}{U'}{U_0} %
    }{%
      |(horns)| \sum_{a \in A} \hornof{n_a}{i_a}{m_a}{j_a} \& |(U_0)| U_0 \\
      |(taus)| \sum_{a \in A} \tauof{n_a}{i_a}{m_a}{j_a} \& |(U')| U \\
      \& \& |(U)| M. %
    }{%
      (horns) edge[labelu={[w_a]_{a \in A}}] (U_0) %
      edge[labell={\sum_{a \in A} i_a}] (taus) %
      (taus) edge (U') %
      edge[bend right,labelbl={[w'_a]_{a \in A}}] (U) %
      (U_0) edge (U') %
      edge[bend left,labelar={k_0}] (U) %
      (U') edge[dashed,labelo={k}] (U) %
    }
  \end{center}
  Informally, $U$ is $U_0$, where we add all the $\taunimj$'s that
  exist in $M$ and whose horn is in $U_0$. We have by construction
  $\El{U} = \El{U_0}$, so $p'$ is indeed a left inverse to $t' \colon
  \El{Y'} \to \El{U}$.

  Finally, define $X'$, $h$, and $s'$ by the pullback
  \begin{center}
    \Diag{%
      \pbk{U'}{X'}{X} %
    }{%
      |(X')| X' \& |(X)| X \\
      |(U')| U \& |(U)| M. %
    }{%
      (X') edge[labelu={h}] (X) %
      edge[into,labell={s'}] (U') %
      (U') edge[labeld={k}] (U) %
      (X) edge[into,labelr={s}] (U) %
    }
  \end{center}
  This altogether yields a vertical morphism
  \begin{center}
      \diag{%
    |(X)| X' \& |(U)| U \& |(Y)| Y', %
  }{%
    (X) edge[into,labeld={s'}] (U) %
    (Y) edge[linto,bend left=20,labeld={t'}] (U) %
    (U) edge[history,bend left=20,labelu={p'}] (Y) %
  }
  \end{center}
  in $\Dov$. A tedious case analysis (made less tedious by $l \colon
  Y' \into Y$ being monic) shows that, because $M$ is a move, $U$ is
  either a move or isomorphic to $Y'$.  So it is in $\Dv$. $U$ is our
  candidate cartesian lifting of $M$ along $l$. More generally, for
  any play $X \xinto{s} U \xotni{t} Y$, choose a decomposition into
  moves. We obtain a candidate cartesian lifting $X' \xinto{s'} U'
  \xotni{t'} Y'$ for $U$, with morphism $(h,k,l)$ to $U$, along any $l
  \colon Y' \into Y$ by taking the successive candidates for each move
  in the obvious way, and composing them.

  To show that this indeed yields a cartesian lifting, consider any
  vertical morphism $X'' \xinto{s''} U'' \xotni{t''} Y''$ and diagram
  \begin{center}
    \diag{%
      |(X'')| X'' \& |(X)| X \\
      |(U'')| U'' \& |(U)| U \\
      |(Y'')| Y'' \& |(Y)| Y, %
    }{%
      (X'') edge[into,labelu={h''}] (X) %
      (U'') edge[into,labelo={k''}] (U) %
      (Y'') edge[into,labeld={l''}] (Y) %
      (X'') edge[into,labell={s''}] (U'') %
      (X) edge[into,labelr={s}] (U) %
      (Y'') edge[linto,labell={t''}] (U'') %
      (Y) edge[linto,labelr={t}] (U) %
    }
  \end{center}
  together with a map $l' \colon Y'' \to Y'$ such that $l \rond l' = l''$.
  By Proposition~\ref{prop:compathist}, letting $p''$ be the history
  of $U''$, the diagram
  \begin{center}
    \diag{%
      U'' \& U \\
      Y'' \& Y %
    }{%
      (m-1-1) edge[labelu={k''}] (m-1-2) %
      edge[history,labell={p''}] (m-2-1) %
      (m-2-1) edge[labeld={l''}] (m-2-2) %
      (m-1-2) edge[history,labelr={p}] (m-2-2) %
    }
  \end{center}
  commutes, so by universal property of pullback, we obtain a map
  $k'_0 \colon \El{U''} \to \El{U'}$, such that $k_0 \rond k'_0 =
  k''_0$, where $k''_0$ is the restriction of $k''$ to dimensions $<
  4$. Furthermore,  the expected map $k' \colon U'' \to U'$,
  is given by universal property of pushout in
  \begin{center}
    \Diag{%
      \pbk[.8cm]{tausb}{U''}{U''_0} %
      \pbk[.8cm]{tausa}{U'}{U'_0} %
    }{%
      \& |(hornsa)| \sum_{a \in A} \hornof{n_a}{i_a}{m_a}{j_a} \& \&
      \& |(U'_0)| U'_0 \\ %
      |(hornsb)| \sum_{b \in B} \hornof{n_b}{i_b}{m_b}{j_b} \& \& \&
      |(U''_0)| \strip{U''} \\ %
      \& \\ %
      \& |(tausa)| \sum_{a \in A} \tauof{n_a}{i_a}{m_a}{j_a} \& \& \&
      |(U')| U' \\ %
      |(tausb)| \sum_{b \in B} \tauof{n_b}{i_b}{m_b}{j_b} \& \& \& |(U'')| U'' %
    }{%
      (hornsa) edge[labelu={}] (U'_0) %
      edge[labell={}] (tausa) %
      (tausa) edge[labeld={}] (U') %
      (U'_0) edge[labelr={}] (U') %
      (hornsb) edge[fore,labelu={}] (U''_0) %
      edge[fore,labell={}] (tausb) %
      (tausb) edge[fore,labeld={}] (U'') %
      (U''_0) edge[fore,labelr={}] (U'') %
      (hornsb) edge (hornsa) %
      (U''_0) edge (U'_0) %
      (tausb) edge (tausa) %
      (U'') edge[dashed,labelbr={k'}] (U') %
    }
  \end{center}
  where $B$ is the family of all commuting squares
  \begin{center}
    \diag{%
      \hornnimj \& \strip{U''} \\
      \taunimj \& U''. %
    }{%
      \sq{w}{i}{k_0}{w'} %
    }
  \end{center}

  Finally, the desired map $h' \colon X'' \to X'$ follows from
  universal property of $X'$ as a pullback, and the square
  \begin{center}
    \diag{%
      |(U'')| U'' \& |(U')| U' \\
      |(Y'')| Y'' \& |(Y')| Y' %
    }{%
      (Y'') edge[into,labell={t''}] (U'') %
      edge[into,labeld={l'}] (Y') %
      (U'') edge[into,labelu={k'}] (U') %
      (Y') edge[into,labelr={t'}] (U') %
    }
  \end{center}
  commutes by uniqueness in the universal property of $U'$ as a
  pullback.
\end{proof}

\subsection{Towards CCS as a playground}
In this section, we prove an intermediate result for proving the
decomposition axioms.

Consider a double cell $\alpha$ of the shape
\begin{center}
  \Diag{%
    \twocellbr{B}{A}{X}{\alpha} %
  }{%
    |(A)| A \& |(X)| X \\
    |(B)| B \&  \\
    |(C)| C \& |(Y)| Y, %
  }{%
    (A) edge[labelu={h}] (X) %
    edge[pro,labell={w}] (B) %
    (X) edge[pro,labelr={u}] (Y) %
    (B) edge[pro,labell={v}] (C) %
    (C) edge[labeld={k}] (Y) %
  }
\end{center}
where $v$ is a view. Let now $\Decomp{\alpha}$ denote the category with
\begin{itemize}
\item objects all tuples $T =
  (Z,l,u_1,u_2,\alpha_1,\alpha_2,\alpha_3)$ such that
\begin{center}
    \Diag{%
      \twocellbr{B}{A}{X}{\alpha_2} %
      \twocellbr{C}{B}{Z}{\alpha_1} %
      \twocell{Z}{X}{R}{}{celllr={0}{.07},bend
        right=10,labeld={\alpha_3}} %
    }{%
      |(A)| A \& |(X)| X \\
      |(B)| B \& |(Z)| Z \& |(R)| \\
      |(C)| C \& |(Y)| Y, %
    }{%
      (A) edge[labelu={h}] (X) %
      edge[pro,labell={w}] (B) %
      (X) edge[pro,labell={u_2}] (Z) %
      (Z) edge[pro,labell={u_1}] (Y) %
      (X) edge[pro,bend left=70,labelr={u}] (Y) %
      (B) edge[pro,labell={v}] (C) %
      (C) edge[labeld={k}] (Y) %
      (B) edge[labelu={l}] (Z) %
    }
\end{center} 
equals $\alpha$ and $\alpha_3$ is an isomorphism;
\item with morphisms $T \to T'$ given by tuples
  $(U,f,\beta,\gamma,\delta)$ (where $f$ is vertical) such that
\begin{center}
  \Diag(1,2){%
    \twocell[.5][.5]{C}{Y}{Z'}{}{celllr={0}{0},bend %
      left,labelo={\scriptscriptstyle \alpha'_1}} %
    \twocell[.5]{B}{Z}{Z'}{}{celllr={0}{0},bend left,%
      labelo={\scriptscriptstyle \delta}} %
    \twocell[.4][.45]{B}{Z}{M}{}{celllr={0}{-.1},bend left=10,%
      labelo={\scriptscriptstyle \alpha_2}} %
    \twocell[.2][.4]{B}{Z'}{X}{}{celllr={0}{0.01},bend
      left,labelo={\scriptscriptstyle \alpha'_2}} %
    \twocellll[.3]{C}{Y}{Z}{\scriptscriptstyle \alpha_1} %
    \twocell[.35][.45]{Z}{Y}{Z'}{}{celllr={0}{0},bend
      left,labelo={\scriptscriptstyle \gamma}} %
    \twocell[.45]{Z'}{Y}{R}{}{celllr={0}{.1},bend
      left,labelo={\scriptscriptstyle \alpha'_3}} %
    \twocell[.3]{Z}{Y}{R}{}{celllr={0}{-.14},bend
      left,fore,labelo={\scriptscriptstyle \alpha_3}} %
    \twocell[.45][.3]{M}{Z}{Z'}{}{celllr={0.1}{0},bend
      left=10,fore,labelo={\scriptscriptstyle \beta}} %
    \path[draw,->] %
    (X) edge[pro,fore,bend right,labell={u_2}] (Z) %
    (Z) edge[pro,fore,labell={%u_1%
    }] (Y) %
    ; %
  }{%
    |(A)| A \& \& |(X)| X \\
    |(B)| B \& |(M)| \& |(Z')| Z' \&  \\
    \& |(Z)| Z \& \& |(R)| \\
    |(C)| C \& \& |(Y)| Y, %
  }{%
    (A) edge[labelu={h}] (X) %
    edge[pro,labell={w}] (B) %
    (X) edge[pro,bend left=70,labelr={u}] (Y) %
    (B) edge[pro,labell={v}] (C) %
    (C) edge[labeld={k}] (Y) %
    (B) edge[labelu={%l%
    }] (Z) %
      % primes
      (X) edge[pro,labell={%u'_2%
      }] (Z') %
      (Z') edge[pro,labell={%u'_1%
      }] (Y) %
      (B) edge[labelu={%l'%
}] (Z') %
      (Z') edge[pro,labelbr={f}] (Z) %
    }
\end{center} 
commutes, i.e., $\gamma \rond (\alpha_1 \vrond \delta) = \alpha'_1$,
$\beta \rond \alpha_2 = \delta \vrond \alpha'_2$, and $\alpha'_3 \rond
(\gamma \vrond u'_2) \rond (u_1 \vrond \beta) = \alpha_3$, and $\beta$
and $\gamma$ are isomorphisms;
\item composition and identities are obvious.
\end{itemize}
So, objects of $\Decomp{\alpha}$ are decompositions of $u$ permitting corresponding decompositions of $\alpha$.
The rest of this section is a proof of:
\begin{lem}\label{lem:decomp}
  $\Decomp{\alpha}$ has a weak initial object, i.e., an object $T$
  such that for any object $T'$ there is a morphism $T \to T'$.
\end{lem}

We start by extending the assignment $U \mapsto G_U$ to a functor, at
least for source-linear $U$.  Let $\SLin$ denote the full subcategory
of $\Chat$ spanning source linear presheaves.  The assignment $U
\mapsto G_U$ actually extends to a functor $G_{-} \colon \SLin \to
\Gph/L$, as follows.  Let, first, for any move $x \in U$, the
\emph{core associated to $x$}, $\core{x}$, be the unique core
reachable from $x$ in $\elements U$, i.e., the unique core $\mu$ for which
there exists $f$ in $\C$ such that $\mu \cdot f = x$.  Now, for any
$\alpha \colon U \to U'$ in $\Chat$, let $G_\alpha \colon G_U \to
G_{U'}$ map any core $x$ in $G_U$ to $\core{\alpha (x)} \in G_{U'}$,
and any non-core vertex $x \in G_U$ to $\alpha (x) \in G_{U'}$. By
naturality, this indeed defines a unique morphism of simple graphs
over $L$.
 \begin{prop} \label{prop:Gfunctor}
   $G_{-} \colon \SLin \to \Gph / L$ is a functor. 
 \end{prop}

We continue with some properties of $\D$.
%%%%% \begin{lem}\label{lem:images}
%%%%%   $\DH$ has images, which are preserved by $\cod,\dom \colon \DH \to
%%%%%   \Dh$, and by vertical composition.
%%%%% \end{lem}
%%%%% \begin{proof}
%%%%%   Easy.
%%%%% \end{proof}

\begin{defi}
  A \emph{filiform} play is any play $U$ such that the restriction of
  $G_U$ to cores and players is a filiform graph, i.e., a graph of the
  shape $\cdot \to \cdot \to \cdots$
\end{defi}
E.g., all views are filiform.

\begin{lem}\label{lem:descentview}
  Any epimorphic (in $\DH$, hence isomorphic) double cell 
\begin{equation}
  \Diag{%
    \twocellbr{B}{A}{X}{\alpha} %
  }{%
    |(A)| A \& |(X)| X \\
    |(B)| B \&  \\
    |(C)| C \& |(Y)| Y, %
  }{%
    (A) edge[labelu={h}] (X) %
    edge[pro,labell={w}] (B) %
    (X) edge[pro,labelr={u}] (Y) %
    (B) edge[pro,labell={v}] (C) %
    (C) edge[labeld={k}] (Y) %
  }\label{eq:alphaa}
\end{equation}
where $v$ is filiform decomposes as 
\begin{center}
  \Diag{%
    \twocellbr{B}{A}{X}{\alpha_2} %
    \twocellbr{C}{B}{Z}{\alpha_1} %
    \twocell{Z}{X}{R}{}{celllr={0}{0},bend
        right,fore,labeld={\scriptscriptstyle \alpha_3}} %
  }{%
    |(A)| A \& |(X)| X \\
    |(B)| B \&  |(Z)| Z \& |(R)| \\
    |(C)| C \& |(Y)| Y, %
  }{%
    (A) edge[labelu={h}] (X) %
    edge[pro,labell={w}] (B) %
    (X) edge[pro,labell={u_2}] (Z) %
    edge[pro,bend left=70,labelr={u}] (Y) %
    (Z) edge[pro,labell={u_1}] (Y) %
    (B) edge[pro,labell={v}] (C) %
    (C) edge[labeld={k}] (Y) %
    (B) edge (Z) %
  }
\end{center}
with $\alpha_3$ an isomorphism, $\alpha_1$ and $\alpha_2$ epimorphic,
uniquely up to isomorphism. In this case, $u_1$ is filiform.
\end{lem}
\begin{proof}
  $B$ has just one player, say $b$. Let $b' = \alpha(b)$. Because $\alpha$ is epi, 
  $\alpha$ induces a morphism $G_\alpha \colon G_{v \vrond w} \to G_u$ of graphs, 
  which is also epi.
  So, $G_u$ may be decomposed as a pushout
  \begin{center}
    \Diag{%
      \pbk{m-2-1}{m-1-1}{m-1-2} %
    }{%
      b' \& G_1 \\
      G_2 \& G_u %
    }{%
      (m-1-1) edge[into,labelu={}] (m-1-2) %
      edge[into,labell={}] (m-2-1) %
      (m-2-1) edge[into,labeld={}] (m-2-2) %
      (m-1-2) edge[into,labelr={}] (m-2-2) %
    }
  \end{center}
  with $G_1 = \im_{G_\alpha} (G_v)$ and $G_2 = \im_{G_\alpha}
  (G_w)$. From this one deduces a decomposition of $u$ and $\alpha$.
\end{proof}

\begin{lem}\label{lem:reflepi}
  For any vertically composable $\alpha$ and $\beta$, if $\alpha \vrond \beta$ is
  epi, then so are $\alpha$ and $\beta$.
\end{lem}
\begin{proof}
  Easy.
\end{proof}

\begin{proof}[Proof of Lemma~\ref{lem:decomp}]
  The double cell $\alpha$ induces morphisms of graphs $G_v \to G_u
  \ot G_w$, by Proposition~\ref{prop:Gfunctor}.  Let $$u_1 = \bigcap
  \ens{u' \subseteq u \aalt (Y \subseteq u') \wedge (\im_\alpha (G_v)
    \subseteq G_{u'})}.$$ Thus, $v \to u$ factors as $v \to u_1 \to
  u$.  Let $Z$ be the position containing all channels of $u_1$, and
  all final players of $u_1$.  Further let ${\uparrow} Z$ denote the
  full subgraph of $G_u$ containing all vertices $x$ with a path to
  some vertex of $Z$. Let then
$$u_2 = \bigcap \ens{u'' \subseteq u \aalt G_{u''} \supseteq {{\uparrow} Z}}.$$
 The union $u_1 \cup u_2$
  is $u$, i.e., the square
  \begin{center}
    \Diag{%
      \pbk{m-2-1}{m-1-1}{m-1-2} %
    }{%
      Z \& u_1 \\
      u_2 \& u %
    }{%
      (m-1-1) edge[labelu={}] (m-1-2) %
      edge[labell={}] (m-2-1) %
      (m-2-1) edge[labeld={}] (m-2-2) %
      (m-1-2) edge[labelr={}] (m-2-2) %
    }
  \end{center}
  is a pushout, i.e., $u_2 \vrond u_1 \iso u$ in $\Cospan{\Chatf}$. So
  it only remains to prove that $Z \to u_1 \ot X$ and $Y \to u_2 \ot
  Z$ are plays, for which we use
  Theorem~\ref{thm:completeness}. First, $u_1$ and $u_2$, as
  subpresheaves of $u$, both are locally 1-injective. Furthermore,
  $G_{u_1}$ and $G_{u_2}$, as subgraphs of a linear and acyclic graph,
  are also linear and acyclic. Now, by definition of $Z$, $Z \to u_1$
  contains all channels and the final players of $u_1$. Further, since
  $X \subseteq u_1$, being initial in $u$ implies being initial in
  $u_1$, so $Z \to u_1 \ot X$ indeed is a play. Symmetrically, no
  player of $u_1$ not in $Z$ is final, so $Y \subseteq u_2$, and hence
  $Y \to u_2$ indeed contains all channels and final players.
  Finally, the players and channels of $Z$ are precisely the initial
  players and channels of $u_2$.  

  It remains to show that the induced decomposition of $\alpha$ is
  weakly initial.  But any decomposition, inducing a decomposition
  $u'_1 \vrond u'_2$ of $u$, should satisfy $Y \subseteq u'_1$,
  $\im_\alpha (G_v) \subseteq G_{u'_1}$, and $G_{u'_2} \subseteq
  {{\uparrow} Z}$, so, ignoring isomorphisms for readability, $u_1
  \subseteq u'_1$ and $u'_2 \subseteq u_2$, as desired.
\end{proof}

\subsection{CCS as a playground}
We are now ready to prove the decomposition axioms, which entail
Theorem~\ref{thm:playground}. They are proved in
Lemmas~\ref{lem:decompleft} and~\ref{lem:decompsym} below.

Let us start with the following easy lemma.
\begin{lem}\label{lem:causalsplit}
  If $u = u_2 \vrond u_1$, then, in $G_u$
  \begin{itemize}
  \item no player of $u_1$ is reachable from any core of
    $u_2$;
  \item no core of $u_1$ is reachable from any element of
    $u_2$.
  \end{itemize}
\end{lem}
\begin{proof}
  For the first point, cores of $u_2$ only reach initial channels of $u_1$.

  For the second point, we further observe that channel and players of $u_2$
  only reach initial players and channels of $u_1$, hence no core.
\end{proof}

The easiest decomposition axiom is~\axref{views:decomp}.
\begin{lem}\label{lem:decompsym}
  $\D$ satisfies~\axref{views:decomp}.
\end{lem}

\begin{proof}
  Although the statement is complicated, this is rather easy: $\alpha$
  restricts to a map of presheaves $f \colon b \to (M \vrond u)$, on which we
  proceed by case analysis.

  If $\im (f) \subseteq M$, then by Lemma~\ref{lem:decomp} and
  correctness we are in the left-hand case.  Otherwise, assume that a
  move $\mu' \in M$ is in the image of $\alpha$, say of a move $\mu
  \in w$. We have a path $\mu \to b$ in $b \vrond w$, hence a path
  $\core{\mu'} \to \alpha(b)$ in $M \vrond u$, contradicting
  Lemma~\ref{lem:causalsplit}.
\end{proof}

Let us now attack the last axiom.
\begin{lem}\label{lem:decompleft}
$\D$ satisfies~\axref{leftdecomposition}.
\end{lem}

We need a few lemmas.

\begin{lem}\label{lem:causalcompo}
  For any plays $A \xproto{u_1} B \xproto{u_2} C$, for any player or channel
  $x \in u_2$ and core $\mu \in u_1$, there is no edge $x \to \mu$ in
  $u_2 \vrond u_1$.
\end{lem}
\begin{proof}
  The existence of $e \colon x \to \mu$ implies $x \in B$, hence $x$
  initial in $u_1$, which contradicts the very existence of~$e$.
\end{proof}

\begin{lem}\label{lem:presfinal}
  Morphisms of plays preserve finality.
\end{lem}
\begin{proof}
  If a player is final in the domain, then it is in the final
  position, hence has an image in the final position of the codomain,
  hence is final there.
\end{proof}

\begin{lem}\label{lem:mapfib}
  For any map $\alpha \colon u \to w$ in $\DH$, for any player $x$ in
  $u$ and edge $e' \colon \mu' \to \alpha (x)$ from a core in $w$,
  there exists a core $\mu \in u$ and an edge $e \colon \mu \to x$ in
  $u$ such that $G_\alpha (e) = e'$.
\end{lem}
\begin{proof}
  Let first $X \to u \ot Y$ and $X' \to w \ot Y'$ be the considered
  morphisms.

  Then, observe that $x$ is not final in $u$, for otherwise it would
  be in $X$, hence $\alpha (x)$ would be in $X'$ and final,
  contradicting the existence of $e'$.

  So there exists $e \colon \mu \to x$ in $u$. But now, by
  target-linearity, $G_\alpha(\mu) = \mu'$, which entails the result.
\end{proof}

\begin{lem}\label{lem:morpbk}
  In any double cell~\eqref{eq:doublecelll}, both squares are pullbacks.
\end{lem}
\begin{proof}
  $X$ must consist precisely of all final players and channels of
  $G_U$, which must also be final in $G_V$, so finality in $G_U$
  implies finality in $G_V$. Conversely, any player or channel mapped
  to a final one in $G_V$ has to be final. So $X$ is a pullback of $U$
  and $X'$.  The lower square being a pullback follows from similar
  reasoning.
\end{proof}

\begin{proof}[Proof of Lemma~\ref{lem:decompleft}]
  Consider any $\alpha$, and construct $C, u_1, u_2,$ and the
  morphisms in Figure~\ref{fig:proof:lem:decompleft}, as follows. %
  \begin{figure}[t]
    \centering
      \Diag (.4,.6) {%
        \pbk[.6cm]{w_1}{w}{w_2} %
        \path[->,draw] %
        (C) edge (Y) %
        (Y) edge (w_1) %
        edge (w_2) %
        (w_1) edge (w) %
        (w_2) edge (w) %
        (C) edge[fore] (u_1) %
        edge[fore] (u_2) %
        (u_1) edge[fore] (u) %
        (u_2) edge[fore] (u) %
        (u) edge[fore,labelu={f}] (w) %
        (u_1) edge[fore,labelu={f_1}] (w_1) %
        (u_2) edge[fore,labeld={f_2}] (w_2) %
        ; %
        \pbk[.6cm]{u}{u_1}{w_1} %
        \pbk[.6cm]{u}{u_2}{w_2} %
        \pbk[.6cm]{u_2}{C}{Y} %
        \pbk[.6cm]{u_1}{C}{Y} %
        \pullback[.6cm]{u_1}{u}{u_2}{draw,-,fore} %
        \path[->,draw] %
        (A) edge[dashed] (u_1) %
        edge[bend left,fore] (u) %
        edge[labelu={f_s}] (X) %
        (X) edge (w_1) %
        edge[bend left,fore] (w) %
        (B) edge[dashed] (u_2) %
        edge[bend right,fore] (u) %
        edge[labeld={f_t}] (Z) %
        (Z) edge (w_2) %
        edge[bend right,fore] (w) %
        ; %
      }{%
        \& \& |(A)| A \& \& \& \& \& \& \& |(X)| X \& \\ %
        \& \& \& \& \& \& \& \& \& \& \\ %
        \& \& \& \& \& \& \& \& \& \& \\ %
        \& \& |(u_1)| u_1 \& \& \& \& \& \& \& |(w_1)| w_1 \& \\ %
        |(C)| C \& \& \& \& \& \& \& |(Y)| Y \& \& \& \\ %
        \& \& \& \& \& \& \& \& \& \& \\ %
        \& \& \& |(u)| u \& \& \& \& \& \& \& |(w)| w \\ %
        \& |(u_2)| u_2 \& \& \& \& \& \& \& |(w_2)| w_2 \& \& \\ %
        \& \& \& \& \& \& \& \& \& \& \\ %
        \& \& \& \& \& \& \& \& \& \& \\ %
        \& |(B)| B \& \& \& \& \& \& \& |(Z)| Z, \& %
      }{%
      }
    \caption{Proof of Lemma~\ref{lem:decompleft}}
\label{fig:proof:lem:decompleft}
\end{figure}%
First, let $u_1$ be the pullback $u \times_{w} w_1$, and then $C = u_1
\times_{w_1} Y$. Let then $u_2 = u \times_w w_2$, and the arrow $C \to
u_2$ be induced by universal property of pullback. By the pullback
lemma, $C = u_2 \times_{w_2} Y$. Because presheaf categories are
adhesive~\cite{DBLP:conf/fossacs/LackS04}, $\Chatf$ is, and, $Y \to
w_1$ being monic, we have a Van Kampen square. Thus, by the main axiom
for adhesive categories, $u$ is a pushout $u_1 +_{C} u_2$, i.e., $u
\iso u_2 \vrond u_1$ in $\Cospan{\Chatf}$. Letting $\alpha_i$ be the
arrow $u_i \to w_i$, for $i = 1,2$, this yields the desired
decomposition of
$\alpha$.  \\

We still need to show that $A \to u_1 \ot C$ and $C \to u_2 \ot B$ are
plays, and that the obtained decomposition is unique.  Uniqueness
follows from adhesivity of $\Chat$ and Lemma~\ref{lem:morpbk}. Indeed,
any decomposition looks like Figure~\ref{fig:proof:lem:decompleft},
except that $u_1$, $u_2$, and $C$ are not \emph{a priori} obtained by
pullback. But by Lemma~\ref{lem:morpbk}, both back faces have to be
pullbacks, hence so are the front faces by adhesivity.

Let us finally show that $u_1$ and $u_2$ are plays. It is easy to see
that non-linearity or non-acyclicity of $G_{u_1}$ (resp.\ $G_{u_2}$)
would entail non-linearity or non-acyclicity of $u$ or $w_1$ (resp.\
or $w_2$). Local 1-injectivity
is also easy. \\

Let us now prove the missing conditions for $A \to u_1 \ot C$.

a) Any player $x$ of $u_1$ in the image of $A$ is final, for
otherwise its image in $w_1$ would be in the image of $X$ and
non-final. 

b) Conversely, if a player $x \in u_1$ is final but not in $A$, then
its image in $u$ must be non-final by Theorem~\ref{thm:completeness},
because $u_1 \to u$ is monic. But then there is a core $\mu$ of $u_2$
with a path $\mu \to x$ in $G_u$, whose images in $w$ yield a path
from a core of $w_2$ to a player of $w_1$, contradicting
Lemma~\ref{lem:causalsplit}. So $A$ contains precisely the final
players of $u_1$. 

c) Now, if a channel $x \in u_1$ is not in $A$, then its image in $u$
must be in $A$, hence $u_1 \to u$ cannot be mono, so neither can $w_1
\to w$, so neither can $Y \to w_2$, contradiction.

d) Finally, by construction, $C$ contains precisely the initial
players
and channels of $u_1$.\\

Now, for $C \to u_2 \ot B$.

a) By universal property of pullback, $C$ contains all channels of
$u_2$. 

b) For players, clearly, for any player $x$ in $C$, $x$ is final in
$u_2$. Indeed, otherwise, there would be a path $\mu \to x$ from a
core $\mu$ in $u_2$, yielding a path $f_2 (\mu) \to f_2 (x)$ in
$w_2$. But since $x$ is in $C$, $f_2 (x) \in Y$, which hence contains
a non-final player, contradiction.

c) Conversely, if $x$ is final in $u_2$, then $x' = f_2 (x)$ is final
in $w_2$. Indeed, otherwise, there would be an edge $\mu' \to x'$ from
a core in $w_2$, so, by Lemma~\ref{lem:mapfib}, an edge $\mu \to x$ in
$u$ with $f (\mu) = \mu'$. But then, $\mu \in u_2$, so $x$ cannot be
final. This shows that $x'$ is final in $w_2$. But then $x' \in Y$,
so, because $C = u_2 \times_{w} Y$, $x \in C$.

d) Consider now any player or channel $x$ initial in $u_2$. First, $x$ is also
initial in $u$: otherwise, there would be an edge $x \to \mu$ to a
core in $u$, with $\mu \in u_1$, hence an edge $f (x) \to f (\mu)$ in
$w$ from a channel of $w_2$ to a core of $w_1$, which is impossible by
Lemma~\ref{lem:causalcompo}. So $x$ is initial in $u$, hence $x \in
B$.

e) Now, for any player or channel $x \in B$, $x$ is initial in $u$, hence $x$ is
\emph{a fortiori} initial in $u_2$. 
\end{proof}

\section{Conclusion and perspectives}\label{sec:conc}

\subsection{Conclusion}
We have described a denotational semantics of CCS based on presheaves,
with a strong game-semantical flavour. Some aspects of the approach
look promising to us. 

First, our result is encouraging for potential applications of Kleene
coalgebra to programming language theory, i.e., ascribing a semantics
to the `rule of the game' rather than attempting to organise
operational semantics into some categorical structure.

Second, our use of techniques from categorical combinatorics (e.g.,
defining positions and plays as finite presheaves) provide a
high-level, yet rigorous toolbox for dealing with string diagrams.
(Compare, e.g., with available definitions of linear logic proof nets
or interaction nets.)

Third, our notion of play encompassing both views and closed-world
plays, and its rich notion of morphism yields a convincing interplay
between strategies (presheaves on views) and behaviours (presheaves on
plays). In particular,
\begin{itemize}
\item passing from one to the other
is handled by standard categorical constructions,
\item the general syntax and \lts{} for strategies provides a link to
  syntactic approaches.
\end{itemize}

Other aspects of our model are not as satisfactory.

First of all, the notion of playground is very complicated.  In work
in progress on a similar approach for $\pi$-calculus, we bypass the
intermediate \lts{} $\TTT_\D$ of process terms, because it does not help so
much ---  strategies are already really close to $\pi$-calculus
terms.  This seems to hint that the main result of playground theory
is actually the characterisation of strategies by the syntax of
Section~\ref{subsec:syntax:strats}. The good point is: this result
does not at all need all axioms for playgrounds.

A second negative point is that some proofs may probably be improved.
E.g., our proof that $\theta \colon \ccs \to \TTT_{\Dccs}$ is included
in weak bisimilarity is a bit of a nightmare, with no apparent good
reason.  Similarly, we know already that our constructions for showing
the fibration axiom~\preaxref{fibration} may be improved. Indeed, the
trick we use to restore synchronisations after restriction rests upon
a \emph{factorisation system}~\cite{FK,Joyal:ncatlab:facto}.  In our
current work on $\pi$, we use factorisation systems to prove the
fibration axiom in a much more direct way (which was prompted by the
fact that the method used here does not apply).
 
\subsection{Perspectives}
Beyond these rather technical concerns, we plan to adapt our semantics
to more complicated calculi like $\pi$, the Join and Ambients calculi,
calculi with passivation, functional calculi, possibly with extra
features (e.g., references, data abstraction, encryption), with a view
to eventually generalising it, perhaps to some SOS format.  In
particular, adapting the approach to functional calculi should clarify
the relationship with Hyland-Ong innocence.  In work in progress
mentioned above, we construct a playground for $\pi$, whose proof of
full abstraction remains to be completed. More speculative directions
include
\begin{itemize}
\item designing a general way of constructing playgrounds
  automatically from more elementary data; work in progress reveals
  that this is a very subtle task;
\item defining a notion of morphisms for playgrounds, which should
  induce translation functions between strategies, and find sufficient
  conditions for such morphisms to preserve, resp.\ reflect testing
  equivalences;
\item generalising playgrounds to apply them beyond programming
  language semantics; in particular, preliminary work shows that
  playgrounds easily account for cellular automata; this raises the
  question of how morphisms of playgrounds would compare with various
  notions of simulations between cellular
  automata~\cite{DBLP:journals/tcs/DelormeMOT11};
\item incorporate quantitative aspects from Kleene coalgebra into
  playground theory; this may start by refining fair testing
  equivalence to keep track of the probability of passing each test
  successfully.
\end{itemize}

\bibliographystyle{plainnat}
\bibliography{../common/bib}

  \begin{figure}[p]
\vspace*{-2em}    \begin{multicols}{2}
\hspace*{-4em}      \begin{tabular}{lp{.8\linewidth}}
        $\C$ & base category, over which positions and plays are presheaves \\
        \begin{minipage}[t]{.3\linewidth}
          \raggedright $\star$, $[n]$, $\paraln, \pararn, \paran$, $\nun$, $\tickn$,
          $\iotani$, $o_{m,j}$, $\taunimj$, 
        \end{minipage} & objects of $\C$ \\
        $\scriptstyle [m] \paraofij{a_1,\ldots}{c_1,\ldots} [n]$ &
        two players sharing some channels \\
        $\Cospan{-}$ & bicategory of cospans of $-$ \\
        $\Dh$ & category of positions and embeddings \\
        $\Dv$ & bicategory of positions and plays \\
        $\D$ & playground \\
%        \begin{minipage}[t]{.15\linewidth}
 %         \raggedright        
        $\Dccs$
%          \end{minipage}
          & playground for CCS \\
        $\E$ & category of plays and extensions \\
        $\BB_X$ & category of behaviours on $X$ \\
        $\EVi$ & category of views and extensions \\
        $\SS_X$ & category of strategies on $X$ \\
        $\Pl(X)$ & players of position $X$ \\
        $v^{x,u}$ & view of $x \colon d \to X$ in $u \colon X \proto Y$ \\
        $x^u \colon d^{x,u} \colon Y$ & initial player of $x$ in $u$ \\
        $\restr{u}{k} \colon D_{k,u} \proto Y$ & 
        restriction of $u \colon X' \proto X$ along $k \colon Y \to X$ \\
        $\Pl_M(X)$ & players of position $X$ whose view in $M \colon X \proto Y$ is non-trivial \\
        $S_x$ & projection of $S \in \SS_X$ to $x \in \Pl(X)$ \\
        $[S,T]$ & copairing of $S$ and $T$ \\
        $S \cdot v$ & residual of $S$ after $v$ \\
        $\restr{S}{\state}$ & restriction of $S$ to antecedents of $\state$ \\
        $\QF$ & graph of full quasi-moves \\
% VIRE        $\QFI$ & graph of interfaced full quasi-moves \\        
        $\MMMB_d$ & set of isomorphism classes of basic moves over $d$ \\
        $\MMMF_X$ & set of isomorphism classes of full moves over $X$ \\
        $\BsofF[M]$ & set of basic $b$'s s.t.\ $\exists$ $b \to M$ \\
        $\MMMFB_X \subseteq \MMMF_X$ & subset of full moves $M$ such
        that $\BsofF[M]$ is a singleton \\
        $\MMMFplus_X \subseteq \MMMF_X$ & subset of full moves $M$ such
        that $\BsofF[M]$ is not a singleton \\
        $r^u, i^u$ &
        bijection, for all plays $u \colon X' \proto X$, 
        $\sum_{(d,x) \in \Pl(X)} \Pl(D_{x,u}) \to \Pl(X')$ \\

        $d \vdash S$ & strategy term \\
        $d \vdashdefinite D$ & definite strategy term \\
        $d \vdash T$ & process term \\

        $(I,h,S)^\bot$ & set of tests passed by $(I,h,S)$ \\

        $\faireqG$ & fair testing eq.\ in graph w.c.\ $G$ \\
        $\faireqs$ & standard fair testing eq.\ in CCS \\
        $\faireq$ & semantic fair testing eq. \\

      \end{tabular}

      \begin{tabular}{lp{.7\linewidth}}

        $\bot^G$ & pole for fair testing eq.\ in $G$ \\
        $\bbot$ & pole for semantic fair test.\ eq. \\
        $\bot$ & pole for CCS (Def.~\ref{def:ccsfair}) \\
        $\barebot$ & intermediate pole (Lem.~\ref{lem:bbot:bot}) \\

        $\ccs$ & \lts{} for CCS \\
        $\SSS$ & \lts{} for strategies \\
        $\TT$ & set of process terms \\
        $\TTT$ & \lts{} for $\TT$: {$\ob (\TTT) {=} \TT$} \\
        $\Translfun$ & translation $\ccs \to \SSS$ \\
        $\theta$ & translation $\ccs \to \TTT$ \\
        $\translfun$ & translation $\TTT \to \SSS$ \\

        $\ccsW$ & set of closed-world quasi-moves \\
        $\DW \subseteq \Dv$ & subbicat.\ of closed-world plays \\
        $\labelDccs$ &
        labelling of closed-world plays in $\ens{\id,\tick}$:
        $\DW \to \freecat{\Sierp}$ \\

        $A^\W$ & `closed-world' subgraph of a graph with complementarity $A$ \\
        $\compata$ & compatibility relation for $A$:  $A^2 \modto \aW$ \\
        $e \dpara e'$ & notation for the composite $\acoh \into A^2 \modto \aW \to \Sierp$ \\

        $[x,y]$ & choice of `amalgamation' in $G$ \\
        % $[x,y]_a$ & choice of `amalgamation' of $x$ and $y$ over $a$,  
        % when $(p(x),p(y)) \compata a$, 
        % for $p \colon G \to A$ \\ 

% VIRE        $U \colon \QFI \to \QF$ & morphism forgetting interfaces \\
% VIRE        $\chi \colon \LLL \to \QFI$ & subgraph of edges with double cell $\idv_I \to M$ \\
        $\chi \colon \LLL \to \QF$ & subgraph of edges with double cell $\idv_I \to M$ \\
        $\xi \colon \LLL \to \A$ & mapping to CCS labels \\
        
        $G$ modular &  $\compatG$ strong bisim over $\Sierp$ \\
        
        $\testable{x}$ & $\ens{y \aalt x \coh y}$ \\
        $\eqtestable{x}{y}$ & $\testable{x} = \testable{y}$ \\

        $\combinea{G}{H}$ & blind composition of $G$ and $H$ over $A$ \\
        \begin{minipage}[t]{.2\linewidth}
        adequacy of $G {\to} A$
        \end{minipage} &
        (essentially) $\bot^{\Gcoh} = \bot^{\combinea{G}{G}}$ \\
        
        $\atrees$ & $A$-trees \\
        $\FFF_a$ & failures over $a \in A$ \\
        $\failoffun$ & failures to $A$-trees: $\FFF \to \atrees$ \\
        nice alphabet & enough ticks, finitely branching, inertly silent 
        (Def.~\ref{def:nice}) \\
        
        core & move element of some presheaf, of maximal dimension \\
        $U$ locally 1-inj. & cores map inj.\ to $U$, except 
        perhaps for channels in the interface  \\
        $G_U$ & causal graph of $U$ \\
        $\El{-}$ & elements $\setminus$ synchronisations \\
        $\histto$ & map between $\El{-}$'s \\
        horn $\hornnimj$ & synchro. minus $\id_{\taunimj}$ \\
        $\cob{f}$ & change of base along $f$
      \end{tabular}
    \end{multicols}
    \caption{Cheat sheet}
\label{fig:cheat}
\end{figure}

\end{document}